\documentclass[a4paper,12pt]{amsart}
\usepackage{amsmath,amsfonts,amssymb,amsthm}
\usepackage[foot]{amsaddr}
\usepackage{istgame}
\usepackage[hidelinks]{hyperref}
\usepackage{mathrsfs}
\usepackage[left=1.5in,right=1.5in,bottom=1.2in]{geometry}

\newcommand{\ignore}[1]{}
\newtheorem{definition}{Definition}[section]
\newtheorem{proposition}[definition]{Proposition}
\newtheorem{theorem}[definition]{Theorem}
\newtheorem{remark}[definition]{Remark}
\newtheorem{lemma}[definition]{Lemma}
\newtheorem{corollary}[definition]{Corollary}

\numberwithin{equation}{section}

\newcommand{\0}{{\bf 0}}

\newcommand{\m}{{\bf m}}
\newcommand{\n}{{\bf n}}
\newcommand{\z}{\mathsf{z}}
\newcommand{\B}{B}
\newcommand{\D}{D}
\newcommand{\I}{\mathcal{I}}
\newcommand{\J}{\mathcal{J}}
\renewcommand{\L}{\mathsf{L}}
\newcommand{\N}{\mathbb{N}}
\newcommand{\R}{\mathbb{R}}
\newcommand{\T}{\mathsf{T}}
\newcommand{\Z}{Z}

\newcommand{\dum}[1]{\mathsf{#1}} 

\newcommand{\End}{\textnormal{End}}
\newcommand{\Der}{\textnormal{Der}}

\newcommand{\Alg}{\textnormal{Alg}}
\newcommand{\cop}{\textnormal{cop}\,}
\newcommand{\Sig}[2]{\D_{(#2,#1)}}

\newcommand{\lhom}{|}
\newcommand{\rhom}{|}

\newcommand{\prelie}{\triangleright}

\usepackage{xcolor}
\definecolor{darkred}{rgb}{0.9,0.1,0.1}
\definecolor{darkblue}{rgb}{0,0,0.7}
\definecolor{darkgreen}{rgb}{0,0.5,0}

\providecommand{\PfStep}[2]{ \ifnum\value{#1}=1
\else\medskip\fi{\sc Step }\arabic{#1}\label{#2}\refstepcounter{#1}.}

\title[The structure group for quasi-linear equations]{The structure group for quasi-linear equations via universal enveloping algebras}
\author{Pablo Linares$^1$}
\email{p.linares-ballesteros@imperial.ac.uk}
\address{${}^1\,$Imperial College London (United Kingdom)} 
\author{Felix Otto$^2$}
\email{felix.otto@mis.mpg.de} 
\author{Markus Tempelmayr$^2$}
\email{markus.tempelmayr@mis.mpg.de}
\address{${}^2\,$Max Planck Institute for Mathematics in the Sciences, Leipzig (Germany)}


\newcommand{\treeZero}{
\hspace{-0.8ex}
\raisebox{0.8mm}{
\begin{istgame}
\istroot(0) \endist
\end{istgame}}
}

\newcommand{\treeZeroo}{
\hspace{-1ex}
\raisebox{-0.2ex}{
\begin{istgame}
\istroot(0) \endist
\end{istgame}}
}

\newcommand{\treeOne}{
\hspace{-0.7ex}
\raisebox{-0.5mm}{
\begin{istgame}
\xtdistance{3mm}{1.5mm}
\setistgrowdirection{north}
\istroot(0)[null node] \istb* \endist
\end{istgame}}
}

\newcommand{\treeThree}{
2\hspace{-0.5ex}
\raisebox{-1mm}{
\begin{istgame}[font=\tiny]
\xtdistance{2mm}{2mm}
\setistgrowdirection{north}
\istroot(0)[null node] \istb* \endist
\istroot(a)(0-1)[null node] \istb*{X_1^2}[xshift=10pt, yshift=3pt] \istb* \endist
\end{istgame}}
\hspace{-1ex}+2\hspace{-0.5ex}
\raisebox{-1mm}{
\begin{istgame}[font=\tiny]
\xtdistance{2mm}{2mm}
\setistgrowdirection{north}
\istroot(0)[null node] \istb*{X_1^2}[xshift=10pt, yshift=3pt] \endist
\istroot(a)(0-1)[null node] \istb* \endist
\istroot(b)(a-1)[null node] \istb* \endist
\end{istgame}}
\hspace{-1ex}+2\hspace{-0.5ex}
\raisebox{-1mm}{
\begin{istgame}[font=\tiny]
\xtdistance{2mm}{2mm}
\setistgrowdirection{north}
\istroot(0)[null node] \istb* \endist
\istroot(a)(0-1)[null node] \istb*{X_1^2}[xshift=10pt, yshift=3pt] \endist
\istroot(b)(a-1)[null node] \istb* \endist
\end{istgame}}
}

\newcommand{\treeFour}{
\hspace{-0.8ex}
\raisebox{-1.8mm}{
\begin{istgame}[font=\tiny]
\istroot(0) \endist
\xtOwner(0){$X_1^2$}[r]
\end{istgame}}
}

\newcommand{\treeNine}[1]{
	\hspace{-0.7ex}
	\raisebox{-3mm}{
		\begin{istgame}[font=\tiny]
			\xtdistance{4mm}{2mm}
			\setistgrowdirection{north}
			\istroot(0)
			\xtOwner(0){$X^{#1}$}[r]
			\istb* \endist
	\end{istgame}}
}

\newcommand{\treeNineteen}{
\hspace{-0.7ex}
\raisebox{-0.5mm}{
\begin{istgame}[font=\tiny]
\xtdistance{3mm}{4mm}
\setistgrowdirection{north}	
\istroot(0) \istb{X^{\mathbf{n}_2}}[xshift=10pt,yshift=8pt] \istb* \istb{X^{\mathbf{n}_1}}[xshift=-5pt,yshift=8pt]  \endist
\end{istgame}}
}

\newcommand{\treeTwenty}{
	\hspace{-0.7ex}
	\raisebox{-0.5mm}{
		\begin{istgame}[font=\tiny]
			\xtdistance{3mm}{4mm}
			\setistgrowdirection{north}	
			\istroot(0) \istb{X^{\mathbf{n}}}[xshift=10pt,yshift=8pt] \istb* \istb{X^{\mathbf{n}}}[xshift=-5pt,yshift=8pt]  \endist
	\end{istgame}}
}

\newcommand{\treeThirty}[1]{
	\hspace{-1.3ex}
	\raisebox{-1.5ex}{
		\begin{istgame}[font=\tiny]
			\istroot(0) \endist
			\xtOwner(0){$#1$}[r]
	\end{istgame}}
}

\usepackage[nogroupskip]{glossaries} 

\makeglossaries

\setglossarystyle{long3col}

\renewcommand*{\glossarymark}[1]{} 

\newglossaryentry{modelspace}{
name = {\ensuremath{\T}}, 
description = {Model space},
sort = t}

\newglossaryentry{polynomialsector}{
name = {\ensuremath{\bar{\T}}},
description = {Polynomial sector of ${\T}$},
sort = t b}

\newglossaryentry{nonpolynomialsector}{
name = {\ensuremath{\mathsf{\tilde T}}},
description = {Non-polynomial subspace of ${\T}$},
sort = t c}

\newglossaryentry{homogeneities}{
name = {\ensuremath{\mathsf A}},
description = {Set of homogeneities},
sort = a}

\newglossaryentry{structuregroup}{
name = {\ensuremath{\mathsf G}},
description = {Structure group},
sort = g}

\newglossaryentry{coordinates}{
name = {\ensuremath{\z_k,\z_\n}},
description = {Coordinate functionals on $(a,p)$},
sort = z}

\newglossaryentry{polynomials}{
text = {\ensuremath{\R[\z_k,\z_\n]}},
name = {\ensuremath{\R[ \cdot ]}},
description = {Polynomials in the variables $\cdot$},
sort = R}

\newglossaryentry{powerseries}{
text = {\ensuremath{\R[[\z_k,\z_\n]]}},
name = {\ensuremath{\R[[ \cdot ]]}},
description = {Formal power series in the variables $\cdot$},
sort = R R}

\newglossaryentry{monomial}{
name = {\ensuremath{\z^\gamma}},
description = {Monomial in $\T^*$},
sort = z gamma a}

\newglossaryentry{monomialT}{
name = {\ensuremath{\z_\gamma}},
description = {Monomial in $\T$},
sort = z gamma}

\newglossaryentry{generatorTilt}{
name = {\ensuremath{D^{(\n)}}},
description = {Infinitesimal generator of constant tilt by $x^\n$},
sort = d n}

\newglossaryentry{generatorShift}{
name = {\ensuremath{\partial_i}},
description = {Infinitesimal generator of shift},
sort = del}

\newglossaryentry{generatorVariableTilt}{
name = {\ensuremath{\z^\gamma D^{(\n)}}},
description = {Infinitesimal generator of variable tilt},
sort = z gamma D}

\newglossaryentry{preLieProduct}{
name = {\ensuremath{\prelie}},
description = {Pre-Lie product on $\Der(\R[[\z_k,\z_n]])$},
sort = 10}

\newglossaryentry{homogeneity}{
name = {\ensuremath{\lhom \gamma \rhom}},
description = {Homogeneity of the multi-index $\gamma$},
sort = 40}

\newglossaryentry{noisehomogeneity}{
name = {\ensuremath{[ \gamma ]}},
description = {Noise-homogeneity of the multi-index $\gamma$},
sort = 30}

\newglossaryentry{parabolicDistance}{
name = {\ensuremath{|\n |}},
description = {Scaled length of $\n$},
sort = 50}

\newglossaryentry{LieAlgebra}{
name = {\ensuremath{\L}},
description = {Lie algebra of $\{ \z^\gamma D^{(\n)} \}_{(\gamma,\n)} \cup \{ \partial_i \}_i $},
sort = L}

\newglossaryentry{UniversalEnvelope}{
name = {\ensuremath{{\rm U} (\L)}},
description = {Universal enveloping algebra of $\L$},
sort = UL}

\newglossaryentry{action}{
name = {\ensuremath{\rho}},
description = {Representation of ${\rm U}(\L)$},
sort = rho}

\newglossaryentry{derivedAlgebra}{
name = {\ensuremath{\tilde{\L}}},
description = {Lie algebra of $\{ \z^\gamma D^{(\n)} \}_{(\gamma,\n)} $},
sort = L derived}

\newglossaryentry{basisU}{
name = {\ensuremath{\D_{(J,\m)}}},
description = {Basis element of ${\rm U}(\L)$},
sort = D }

\newglossaryentry{hopfSpace}{
name = {\ensuremath{\T^+}},
description = { Dual structure of ${\rm U}(\L)$ },
sort = T plus }

\newglossaryentry{basisTplus}{
name = {\ensuremath{\Z^{(J,\m)}}},
description = {Basis element of $\T^+$},
sort = Z J }

\newglossaryentry{coaction}{
name = {\ensuremath{\Delta}},
description = {Coaction on $\T$ over $\T^+$},
sort = delta }

\newglossaryentry{coproduct}{
name = {\ensuremath{\Delta^+}},
description = {Coproduct in $\T^+$},
sort = delta plus }

\newglossaryentry{antipode}{
name = {\ensuremath{\mathcal{S}}},
description = {Antipode in $\T^+$},
sort = S }

\newglossaryentry{embeddingJ}{
name = {\ensuremath{\mathcal{J}_\n}},
description = {Embedding of $\{\gamma\in\mathsf{\tilde T} \, | \, \lhom\gamma\rhom>|\n|\}$ into $\T^+$},
sort = J n }

\newglossaryentry{gammaF}{
name = {\ensuremath{\Gamma_f}},
description = {Generic element of the structure group ${\mathsf G}$},
sort = Gamma f }

\newglossaryentry{model}{
name = {\ensuremath{\Pi}},
description = {Model of a regularity structure},
sort = Pi }

\newglossaryentry{projectionP}{
name = {\ensuremath{P}},
description = {Projection of $\mathsf{T}^*$ on $\mathsf{\tilde T}^*$},
sort = P}

\newglossaryentry{projectionPdagger}{
name = {\ensuremath{P^\dagger}},
description = {Projection of $\mathsf{T}$ on $\mathsf{\tilde T}$},
sort = Pdag}

\newglossaryentry{GLpreLie}{
name = {\ensuremath{{\mathcal L}^1}},
description = {(pre-)Lie algebra of Connes-Kreimer in \cite{ConnesKreimer}},
sort = L GL }

\newglossaryentry{modelSpaceBranched}{
name = {\ensuremath{{\mathcal B}}},
description = {Model space for branched rough paths},
sort = B }

\newglossaryentry{preLieGL}{
name = {\ensuremath{\leadsto}},
description = {Pre-Lie product on ${\mathcal L}^1$},
sort = 15 }

\newglossaryentry{basisGL}{
name = {\ensuremath{\Z_\tau}},
description = {Basis element of ${\mathcal L}^1$},
sort = Z tau }

\newglossaryentry{genericTree}{
name = {\ensuremath{\tau}},
description = {Generic tree},
sort = tau }

\newglossaryentry{modelBranched}{
name = {\ensuremath{{\mathbb X}}},
description = {Model for branched rough paths},
sort = X }

\newglossaryentry{grafting}{
name = {\ensuremath{\curvearrowright}},
description = {Grafting of trees},
sort = 5 }

\newglossaryentry{phi}{
name = {\ensuremath{\phi}},
description = {Dictionary between multi-indices and trees},
sort = phi }

\newglossaryentry{Phi}{
name = {\ensuremath{\Phi}},
description = {Dictionary between $\T^+$ and $\T^+_H$},
sort = Phi 3 }

\newglossaryentry{phio}{
	name = {\ensuremath{\mathring{\phi}}},
	description = {Dictionary between $\T$ and $\mathscr{B}$},
	sort = phi 1 }

\newglossaryentry{phiomin}{
	name = {\ensuremath{\mathring{\phi}_-}},
	description = {Dictionary between $\mathsf{\tilde T}$ and $\mathscr{B}$},
	sort = phi 2 }

\newglossaryentry{embeddingHairer}{
name = {\ensuremath{\mathcal{J}_\n^H}},
description = {Abstract placeholder for Taylor coefficients in \cite{Hairer}},
sort = J n H }

\newglossaryentry{multFunc}{
name = {\ensuremath{f}},
description = {Generic element of $\Alg(\T^+,\R)$},
sort = f}

\newglossaryentry{coproductButcher}{
name = {\ensuremath{\Delta_B}},
description = {Coproduct on trees of \cite{Brouder}},
sort = Delta 30}

\newglossaryentry{LieAlgebraRP}{
name = {\ensuremath{\L_{RP}}},
description = {Restriction of $\L$ relevant for rough paths},
sort = L rp}

\newglossaryentry{ModelSpaceRP}{
name = {\ensuremath{\T_{RP}}},
description = {Restriction of $\T$ relevant for rough paths},
sort = T h rp}

\newglossaryentry{TplusRP}{
name = {\ensuremath{\T^+_{RP}}},
description = {Restriction of $\T^+$ relevant for rough paths},
sort = T plus rp}

\newglossaryentry{coproductRP}{
name = {\ensuremath{\Delta^+_{RP}}},
description = {Restriction of $\Delta^+$ relevant for rough paths},
sort = Delta plus rp}

\newglossaryentry{ModelSpaceH}{
name = {\ensuremath{\T_{H}}},
description = {Model space of \cite{Hairer}},
sort = T H}

\newglossaryentry{TplusH}{
name = {\ensuremath{\T^+_{H}}},
description = {Formal expressions representing Taylor coefficients of \cite{Hairer}},
sort = T plus H}

\newglossaryentry{comoduleH}{
name = {\ensuremath{\Delta_{H}}},
description = {Coaction on $\T_H$ over $\T^+_H$ of \cite{Hairer}},
sort = Delta H}

\newglossaryentry{coproductH}{
name = {\ensuremath{\Delta^+_{H}}},
description = {Coproduct on $\T^+_H$ of \cite{Hairer}},
sort = Delta plus H}

\newglossaryentry{iotan}{
name = {\ensuremath{\iota_\n}},
description = {Projection of $\{ \z^\gamma D^{(\n)} \}_\gamma$ to $\T^*$},
sort = iota n}
	
\newglossaryentry{bigrading}{
name = {\ensuremath{{\rm bi}}},
description = {Bigrading on pairs $(\gamma,\n)$ and $(J,\m)$},
sort = bi}

\newglossaryentry{grading}{
text = {\ensuremath{|(J,\m)|_{\rm gr}}},
name = {\ensuremath{| \cdot |_{\rm gr}}},
description = {Grading on $(J,\m)$},
sort = 70}

\newglossaryentry{lengthJm}{
name = {\ensuremath{| (J,\m) |}},
description = {Length of $(J,\m)$},
sort = 60}

\newglossaryentry{homogeneityHairer}{
text = {\ensuremath{| \cdot |_H}},
name = {\ensuremath{| \tau |_H}},
description = {Homogeneity of the tree $\tau$ of \cite{Hairer}},
sort = 45}

\newglossaryentry{BspaceHairer}{
name = {\ensuremath{\mathscr{B}}},
description = {Space of trees with expanded polynomial decorations in \cite{BCCH19}},
sort = B scr }

\newglossaryentry{arrown}{
name = {\ensuremath{\curvearrowright_\n}},
description = {Generalized grafting pre-Lie product in \cite[Definition 4.7]{BCCH19}},
sort = 6}

\newglossaryentry{preLien}{name={\ensuremath{\leadsto_\n}},
description={Generalized grafting pre-Lie product on $\mathcal{L}^1$},
sort = 16
}

\newglossaryentry{arrowup}{
name={\ensuremath{\uparrow_i}},
description={Operator of increasing decorations in \cite[Definition 4.7]{BCCH19}},
sort = 7}

\newglossaryentry{sharp}{name={\ensuremath{\sharp_i}},
description={Operator of increasing decorations on $\mathcal{L}^1$},
sort = 17}

\newglossaryentry{contraction}{name={\ensuremath{\mathcal{Q}}},
description={Decorations contraction operator in \cite[p. 911]{BCCH19}},
sort = Q 2}

\newglossaryentry{counitn}{name={\ensuremath{\varepsilon_{\mathbf{n}}}},
description={Generalized counit in $\mathrm{U}(\mathsf{L})$},
sort = epsilon}

\newglossaryentry{rmap}{name={\ensuremath{M_c}},
description={Renormalization map for rough paths},
sort = Mc}

\newglossaryentry{rmapBCFP}{name={\ensuremath{M_v^{BCFP}}},
description={Translation map for rough paths in \cite[Definition 14]{BCFP19}},
sort = McBCFP}

\newglossaryentry{intHairer}{name={\ensuremath{\mathcal{I}}},
description={Abstract integration map in \cite{Hairer,BCCH19}},
sort = I}

\newglossaryentry{dagger}{name={\ensuremath{\dagger}},
description={Transposition of linear maps},
sort = dagger}

\newglossaryentry{cop}{name={\ensuremath{\mathrm{cop}}},
description={Coproduct in $\mathrm{U}(\mathsf{L})$},
sort = cop}

\newglossaryentry{upsilon}{name={\ensuremath{\Upsilon}},
description={Linear map from $\mathcal{B}$ into functions, cf. \cite[Definition 2.13]{BonnefoiCMW}},
sort = upsilon}

\newglossaryentry{upsilonCirc}{name={\ensuremath{\mathring\Upsilon}},
description={Linear map from $\mathscr{B}$ into functions, cf. \cite[(4.4)]{BCCH19}},
sort = upsiloncirc}

\begin{document}


\begin{abstract}
We replace trees by multi-indices
as an index set of the abstract model space to tackle
quasi-linear singular stochastic partial differential equations.
We show that this approach is consistent with the postulates of regularity structures when it comes to the structure group, which arises from a Hopf algebra and a comodule.  

Our approach, where the dual of the abstract model
space naturally embeds into a formal power series algebra,
allows to interpret the structure group as a
Lie group arising from a Lie algebra consisting of derivations
on this power series algebra. These derivations in turn
are the infinitesimal generators of two actions on the space of pairs (nonlinearities, functions of space-time mod constants).

We also argue that there exist pre-Lie algebra and Hopf algebra morphisms between our structure and the tree-based one in the cases of branched rough paths (Grossman-Larson, Connes-Kreimer) and of the stochastic heat equation.
\end{abstract}

\maketitle

\tableofcontents



\section{Introduction}

In this article, we connect the regularity structure $(\mathsf{A},\mathsf{T},\mathsf{G})$
introduced by the second author, Sauer, Smith and Weber in \cite{OSSW} for a simple class of quasi-linear equations to
the general Hopf-algebraic framework formulated by Hairer \cite{Hairer14} and later expanded in \cite{BHZ,ChH,BCCH19}.
The main difference between \cite{OSSW} on the one hand,
and the output of the general strategy in \cite{Hairer14} 
applied to this class of equations on the other hand, lies in the effectively smaller abstract model space 
$\mathsf{T}$: The basis elements in \cite{OSSW} amount to
specific linear combinations of the basis in \cite{Hairer14}, which is indexed
by trees. Trees do not play any role in the contribution of this paper; thus, the Hopf algebras underlying rough paths \cite{Lyons98,Gub04} as worked out in \cite{Gub10,HairerKelly}, and regularity structures \cite{Hairer14,BHZ} are not at our disposal. The goal of this paper is to unveil this Hopf structure in the tree-free set-up of \cite{OSSW}; loosely speaking, this amounts to replacing combinatorics by Lie geometry. For an introduction to our framework, we refer to the series of lectures \cite{OttoLecture} and the notes \cite{LO}; for an introduction to classical regularity structures, we refer to the review article \cite{Hairer}.

\medskip

In our approach to the regularity structure $(\mathsf{A},\mathsf{T},\mathsf{G})$, we start from the space
of tuples $(a,p)$ of (polynomial) nonlinearities $a$ and space-time polynomials $p$,
which we think of parameterizing the entire manifold of solutions\footnote{ satisfying the
	equation up to space-time polynomials} $u$. 
We consider the actions of shift by a space-time vector $h\in\mathbb{R}^{d+1}$ and of tilt by a
space-time polynomial $q$ on $(a,p)$-space, where, crucially, the tilt by a constant is 
encoded as a shift of the (one-dimensional) $u$-space because we think of $p$ as $p$ mod constants.
We consider the infinitesimal generators of these actions, and pull them back
as derivations on the algebra of formal power series $\R[[\z_k,\z_\n]]$ in the natural coordinates $\{\z_k\}_{k\geq 0}$ and $\{\z_{\bf n}\}_{\N_0^{d+1}\setminus \{\0\}}$ of $(a,p)$-space\footnote{ For the sake of clarity, we will fix $d=1$, though no fundamental changes appear when increasing the spatial dimension.}, which give rise to an index set of multi-indices. These derivations define a Lie algebra $\mathsf{L}$; the corresponding Lie group coincides with the pointwise dual $\mathsf{G}^*$ of the structure group $\mathsf{G}$. However, we take a completely algebraic route to construct $\mathsf{G}$, which passes via the universal envelope ${\rm U}(\mathsf{L})$ of $\mathsf{L}$, and a module structure which identifies ${\rm U}(\mathsf{L})$ with a space of endomorphisms of $\T^*$, the algebraic dual of the model space $\T$.

\medskip

The algebraic construction of the present paper is similar in spirit to the recent work by Bruned and Manchon \cite{BrunedManchon}, who construct Hopf algebras starting from a (multi) pre-Lie algebra that encodes grafting of decorated trees, following the general theory developed by Guin and Oudom \cite{GuinOudom}. 
Also our Lie structure comes from a natural pre-Lie product on $\mathsf{L}$, which however is not closed (see Subsection \ref{Sect3.10} for details); more recently, and motivated by the present paper, Bruned and Katsetsiadis \cite{BK22} have interpreted this structure as a post-Lie algebra. 
Like in \cite{GuinOudom} we use it to canonically identify the enveloping algebra
${\rm U}(\mathsf{L})$ with the symmetric algebra ${\rm S}(\mathsf{L})$,
which we implement through the choice of a specific basis\footnote{ This basis is different from the standard basis used in the Poincar\'e-Birkhoff-Witt  
	theorem, which relies on a non-canonical ordering of the index set of $\mathsf{L}$.}  for ${\rm U}(\mathsf{L})$. This basis is
crucial for recovering the intertwining property that relates the coproduct and coaction from regularity structures \cite[(4.14)]{Hairer}.

\medskip

While in \cite{OSSW} the regularity structure $(\mathsf{A},\mathsf{T},\mathsf{G})$ was introduced for quasi-linear equations of the form
\begin{align}\label{as34}
\big(\frac{\partial}{\partial x_2}-\frac{\partial^2}{\partial x_1^2}\big)u=a(u)\frac{\partial^2 u}{\partial x_1^2} + \xi,
\end{align}
the structures defined in this paper cover other (semi-linear)
equations relying on a single scalar nonlinearity $a(u)$. In fact, the algebraic structure we build is oblivious to the form of the equation and just relies on the solution to the linearized problem being of positive regularity, as will become apparent in Section 6 and, more importantly, Section 7.\footnote{ Our structure would also work, for example, for a generalized KPZ equation with only one nonlinearity, i.~e.
\begin{equation*}
	\big(\frac{\partial}{\partial x_2}-\frac{\partial^2}{\partial x_1^2}\big)u=a(u)\big(\frac{\partial}{\partial x_1}u\big)^2 + \xi.
\end{equation*}} 
In Section \ref{Sect5} we consider a driven ODE of the form 
\begin{equation}\label{drivenODE}
\frac{du}{dx_2} = a(u)\xi,
\end{equation}
provide a dictionary $\phi$ between our index set of multi-indices and linear combinations of trees in the Connes-Kreimer Hopf algebra (which is at the basis of branched rough paths), and prove that $\phi$ generates a Hopf algebra morphism. Section \ref{Sect6} is devoted to the stochastic heat equation (SHE)
\begin{equation}\label{SHE}
\big(\frac{\partial}{\partial x_2}-\frac{\partial^2}{\partial x_1^2}\big)u=a(u)\xi,
\end{equation}
where now the morphism property is established with respect to the Hopf algebra in regularity structures \cite{Hairer}. While the morphism $\phi$ between our model space and the one based on decorated trees 
changes from one equation to another, the consistency between our geometric
definition and the combinatorial definitions persists, and we expect it to hold as well for the class \eqref{as34}.

\medskip

Working with our more parsimonious regularity structure $(\mathsf{A},\mathsf{T},\mathsf{G})$ 
and model $(\Pi_x,\Gamma_{xy})$
has the potential advantage of reducing the number of counter-terms in renormalization.
In joint work with P.~Tsatsoulis \cite{LOTT}, we show that algebraic renormalization of \eqref{as34} combines well with our greedier setting:
we show that under a natural symmetry condition on the noise $\xi$, a 
BPHZ-type choice of renormalization can be performed, leading to a renormalized equation  
of the form $\big(\frac{\partial}{\partial x_2}-\frac{\partial^2}{\partial x_1^2}\big)u=a(u)\frac{\partial^2 u}{\partial x_1^2} + h(u) + \xi$ with a deterministic and 
local counter-term $h$, as postulated in \cite[Theorem 1]{OSSW}, and -- most importantly -- we show
that the greedy model $(\Pi_x,\Gamma_{xy})$ can be naturally estimated without resorting to trees. We believe this to be a general principle, namely that multi-indices provide a more efficient bookkeeping of the renormalization constants; in Subsection \ref{Sect6.6} we show this for \eqref{drivenODE}, connecting to translation of rough paths \cite{BCFP19}. 

\medskip

Our Ansatz has two invariances built-in, with beneficial effects for renormalization.
The first invariance is the independence on the choice of an origin in $u$-space,
which is ensured by the prominent role of the infinitesimal generator of $u$-shifts.
The second invariance relates to the more specific class of quasi-linear equations
\eqref{as34}. Namely, our theory is not affected by interpreting \eqref{as34} as a perturbation
around $(\frac{\partial}{\partial x_2}-2\frac{\partial^2}{\partial x_1^2})u=\xi$, i.~e.~
by rewriting \eqref{as34} as 
$(\frac{\partial}{\partial x_2}-2\frac{\partial^2}{\partial x_1^2})u
=(a(u)-1)\frac{\partial^2}{\partial x_1^2}u+\xi$. In other words, our approach is invariant
under changing the reference value in $a$-space. Insisting on such (collective) in- or
rather covariances in renormalization has been implemented in a much broader way in \cite{Hollands05}\footnote{ There, in the absence of a linear structure, the shift of $u$ and $a$ not just by a constant,
	but by a (regular) field is considered.}.

\section{Motivation and interpretation of the main result}\label{Sect1}

\subsection{Modding out constants}
\label{Sect1.11}
$\mbox{}$

We take the perspective Butcher \cite{Butcher} introduced on the level of ODEs, 
and which was extended in \cite{Gub10} to driven ODEs of the form\footnote{We consistently denote by $x_2$ the time-like variable.} 
\eqref{drivenODE}, of viewing the solution
of the homogeneous initial value problem, i.~e.~with $u(x_2=0)=0$, 
as a function(al) of the nonlinearity $a$, i.~e.~$u=u[a](x_2)$.
Obviously, the solution $\tilde u$ for an (inhomogeneous) initial datum $u_0$ can then
be recovered by a $u$-shift: 
\begin{align}\label{fs15}
\tilde u=u[a(\cdot+u_0)]+u_0.
\end{align}
In particular, re-centering in the sense of imposing homogeneous initial conditions
at some other time instance, say $u_1(x_2=1)=0$, can be recovered by a suitable
variable $u$-shift $\pi=\pi[a]$ in the form of the Ansatz $u_1[a]=u[a(\cdot+\pi[a])]+\pi[a]$.

\medskip

The extension to a driven PDE, e.~g.~\eqref{SHE},
is more subtle, since even for fixed $a$, the solution manifold is infinite-dimensional.
Relaxing the equation to hold only modulo space-time polynomials, one expects
that the solution manifold can be (locally) parameterized by all space-time
polynomials $p$. It is therefore natural to think in terms of $u=u[a,p](x)$,
like is implicitly
done in\footnote{There, ${\mathcal O}$ corresponds to the space of all jets $p$'s.} \cite[p.879]{BCCH19}. However, this is an over-parameterization
in the sense that it does not take advantage of $u$-shifts, cf.~\eqref{fs15}. 
A key feature of our approach, which will be spelled out in the upcoming
Subsections \ref{Sect1.12} and \ref{Sect1.13} is to consider $p$ only modulo
constants (and to keep track of $u$ only modulo constants). 
In Subsection \ref{Sect3.5} we argue
that this greedy approach to the regularity structure
is actually truthful.

\subsection{The $(a,p)$-space}
\label{Sect1.12}
$\mbox{}$

At the basis of our construction is the space 
$\mathbb{R}[u]\times(\mathbb{R}[x_1,x_2]/\mathbb{R})$, 
which is the set\footnote{
Our approach ignores the linear structure of $a$-space, and only appeals to the affine structure of $p$-space.} of pairs $(a,p)$, where $a$
is a polynomial in a single variable $u$,
and $p$ is a polynomial in two variables $(x_1,x_2)=x$.
As indicated by the quotient, we consider $p$'s only up to additive constants. 
Note that $\mathbb{R}[u]\times(\mathbb{R}[x_1,x_2]/\mathbb{R})$ is the direct sum
indexed by the disjoint union of $\mathbb{N}_0$ and $\mathbb{N}_0^2\setminus\{{\bf 0}\}$;
we often write $\{k\ge 0\}\cup\{{\bf n}\not={\bf 0}\}$.

\medskip

We recall that the polynomial $a$ plays the role of the nonlinearity in case of the quasi-linear class \eqref{as34}, its argument $u$ is a placeholder for the solution. The polynomial
$p$ plays the role of a (local) parameterization of the manifold
of solutions; 
the values of $u$ and $p$ are thus thought to be in
the same space, i.~e.~the real line, whereas
the argument $x$ of $p$ is in space-time.

\subsection{Actions of shift and tilt}
\label{Sect1.13}
$\mbox{}$

There are two natural actions on the $(a,p)$-space $\mathbb{R}[u]\times(\mathbb{R}[x_1,x_2]/\mathbb{R})$,
which we shall call ``shift'' and ``tilt''.
We start by introducing the shift, by which we think of shifts of space $x_1$ and time $x_2$. 
We seek an action\footnote{By action one means that addition in the group $\mathbb{R}^2$
is compatible with composition of the transformations of $(a,p)$-space.}
of the additive group $\mathbb{R}^2\ni h$ on $(a,p)$-space. We (momentarily) identify 
\begin{align}\label{as35}
\mathbb{R}[x_1,x_2]/\mathbb{R}\cong\{\,p\in \mathbb{R}[x_1,x_2]\,|\,p(0)=0\,\},
\end{align}
which in particular allows to define the composition $a\circ p\in\mathbb{R}[x_1,x_2]$ 
on $(a,p)$-space $\mathbb{R}[u]\times(\mathbb{R}[x_1,x_2]/\mathbb{R})$.
Then for $h\in\mathbb{R}^2$, the transformation
\begin{align}\label{ao27bis}
(a,p)\mapsto \Big(a\big(\cdot+p(h)\big),p(\cdot+h)-p(h)\Big)
\end{align}
is well-defined. The action on the $p$-component is such that it corresponds to
shift projected onto \eqref{as35}. The action on the $a$-component
is made such that the composition $a\circ p$ is mapped onto its (unprojected)
shift $x\mapsto(a\circ p)(x+h)$. Thus under the lens of $a$, this action corresponds to
the plain shift by $h$.
It is easy to check that \eqref{ao27bis} is indeed an action.
The presence of the composition $a\circ p$ connects to the Fa\`a di Bruno formula, cf. \cite{Frabetti},
which expresses composition in terms of coefficients and thus encodes the chain rule.
It explicitly enters in the exponential formula \eqref{LOsg3} via \eqref{act1} and \eqref{fw13}.

\medskip

We now turn to tilt. 
By this we momentarily\footnote{We need an extension later on.}
think of an action
on $(a,p)$-space of the polynomial space $\mathbb{R}[x_1,x_2]$ (now including the constants).
It is defined by
writing $\mathbb{R}[x_1,x_2]\ni q$ $=\sum_{{\bf n}}\pi^{({\bf n})}x^{\bf n}$, 
where\footnote{With the implicit understanding that $\n\in\N_0^2$ if not stated otherwise.}
$x^{\bf n}$ $=x_1^{n_1}$ $x_2^{n_2}$,
and considering
\begin{align}\label{ao29bis}
(a,p)\mapsto\Big(a\big(\cdot+\pi^{(\0)}\big), p+\sum_{{\bf n}\not=\0}\pi^{({\bf n})}x^{\bf n}\Big).
\end{align}
This treatment of the $p$-component ensures that the transformation \eqref{ao29bis} 
is well-defined in view of \eqref{as35}. The treatment of the $a$-component is
such that the composition $a\circ p$ is mapped onto $a\circ(p+q)$
under \eqref{ao29bis}. So once more, under the lens of $a$, 
this action corresponds to the tilt of $p$ by $q$. Note that the addition of a polynomial is involved in recentering, by analogy with the addition of constants in the ODE case at the beginning of Subsection \ref{Sect1.11}.

\subsection{Seeking a representation as algebra endomorphisms}\label{Sect1.2}
$\mbox{}$

We are interested in the group $\mathsf{G}$ of transformations on $(a,p)$-space
generated by the two actions \eqref{ao27bis} and \eqref{ao29bis}.
We seek a representation of $\mathsf{G}$ as a matrix group, i.~e.~as a subgroup
of ${\rm Aut}(\mathsf{T})$ for a suitable linear space 
$\mathsf{T}$.
The natural approach is to lift \eqref{ao27bis} and \eqref{ao29bis} to 
an action on a space of nonlinear functionals $\pi$
on $(a,p)$-space by ``pull-back''. 
Indeed, it is tautological that 
\eqref{ao27bis} defines an algebra endomorphism $\Gamma^*$ of the algebra of functions $\pi$ on $(a,p)$-space
via
\begin{align}\label{ao80}
\Gamma^*\pi[a,p]=\pi\Big[a\big(\cdot+p(h)\big),p(\cdot+h)-p(h)\Big].
\end{align}
Similarly for \eqref{ao29bis}
\begin{align}\label{ao79}
\Gamma^*\pi[a,p]=\pi\Big[a\big(\cdot+\pi^{(\0)}\big),
p+\sum_{{\bf n}\not={\bf 0}}\pi^{({\bf n})}x^{\bf n}\Big].
\end{align}
For the moment, the notation $\Gamma^*$ is just suggestive; it will become meaningful when we identify this object with the dual of an element of $\mathsf{G}$. The same remark applies to the forthcoming $\mathsf{T}^*$ and $\mathsf{G}^*$.

\medskip

This pull-back also suggests to naturally extend \eqref{ao79} from
constant tilt $\pi^{({\bf n})}\in\mathbb{R}$ to variable tilt, meaning that
$\pi^{({\bf n})}$ itself is a function on $(a,p)$-space:
\begin{align}\label{ao81}
\Gamma^*\pi[a,p]=\pi\Big[a\big(\cdot+\pi^{(\0)}[a,p]\big),
p+\sum_{{\bf n}\not={\bf 0}}\pi^{({\bf n})}[a,p]x^{\bf n}\Big].
\end{align}
Note that also \eqref{ao80} has this form.

\medskip

We use the notation $\pi$ for a generic
function on $(a,p)$-space,
since it acts as a placeholder for the model $\Pi=\Pi[a,p](x)$, 
which indeed can be considered as a parameterization of the solution manifold by $p$
and depending on $a$ (next to depending on space-time $x$).

\subsection{Seeking a group structure}\label{Sect1.3}
\mbox{}

Obviously, \eqref{ao81} no longer is an action of the additive group of
functions $\{\pi^{({\bf n})}[a,p]\}$; however, it can be interpreted as an action
of the monoid given by the (non-Abelian) group operation
\begin{align}\label{as33}
\bar\pi^{({\bf n})}=\pi^{({\bf n})}+\Gamma^*{\pi'}^{({\bf n})},
\end{align}
in the sense that the corresponding\footnote{
Here, $\Gamma^*$, $\Gamma'^*$ and $\bar\Gamma^*$ are defined through \eqref{ao81} by $\{\pi^{(\n)}\}_\n$, $\{\pi'^{(\n)}\}_\n$ and $\{\bar\pi^{(\n)}\}_\n$, respectively.}
$\bar\Gamma^*$ satisfies 
$\bar\Gamma^*$ $=\Gamma^*{\Gamma'}^*$. The argument\footnote{
We will provide a rigorous proof in the context of Proposition \ref{exp01}.} for \eqref{as33} is a straightforward computation from \eqref{ao81}. The relation $\{\pi^{(\n)}\}_\n \mapsto \Gamma^*$ given by \eqref{ao81} and the composition rule \eqref{as33} reflect \cite[Definition 14]{BCFP19} on the level of branched rough paths. 

\medskip

While according to \eqref{as33}, the set of $\Gamma^*$'s defined through \eqref{ao81},
where $\pi^{({\bf n})}$ runs through all functions on $(a,p)$-space,
is closed under composition, there is in general no inverse.
For this, we will have to pass to a more restricted space for the $\pi^{({\bf n})}$'s. 
Incidentally, while \eqref{ao80} is contained in \eqref{ao81} when
$\pi^{({\bf n})}$ runs through all functions on $(a,p)$-space,
this will not be the case for the restricted space.

\medskip

\subsection{Seeking a matrix representation}\label{Sect1.4}
\mbox{}

A reason for not only restricting the space of $\pi^{({\bf n})}$'s
but also the one of $\pi$'s in (\ref{ao81}) is that the algebra of
all functions on $(a,p)$ is too large for a representation in terms
of countably many coordinates.
Let us therefore start from the following coordinates on $(a,p)$-space:
\begin{align}\label{ao20bis}
\begin{split}
\z_k[a,p]&=\frac{1}{k!}\frac{d^k a}{du^k}(0),\  k\geq 0\ \mbox{and}\\ 
\z_{\bf n}[a,p]&=\frac{1}{{\bf n}!}\frac{d^{\bf n}p}{dx^{\bf n}}(0),\  \n\neq \0.
\end{split}
\end{align}
In (\ref{ao20bis}) we use the standard
abbreviation ${\bf n}!$ $=n_1! \, n_2!$ and $\frac{d^{\bf n}}{dx^{\bf n}}$
$=\frac{d^{n_1}}{dx_1^{n_1}} \, \frac{d^{n_2}}{dx_2^{n_2}}$.
Note that $\{ \gls{coordinates}\}_{k,\n}$ can be considered as the dual basis to the standard monomial basis $\{u^k,x^\n\}_{k,\n}$ of $(a,p)$-space.
In particular, these coordinates arbitrarily fix an origin of the affine $u$-space and the affine $x$-space. The effect of changing these origins is considered in Subsection \ref{Sect3.2}. Clearly, every polynomial expression in (\ref{ao20bis}) can be identified with
a function on $(a,p)$-space. This allows us to identify the polynomial algebra
$\gls{polynomials}$ with a sub-algebra of the algebra
of functions on $(a,p)$-space.
Note that $\mathbb{R}[\z_k,\z_{\bf n}]$
is the direct sum over the index set of multi-indices\footnote{
This means that $\gamma$ is a map from the above index set into $\N_0$ that is zero for all but finitely many indices.}
$\gamma$. In particular, the monomials
\begin{align}\label{ao85}
\z^\gamma:=\prod_{k\ge 0,{\bf n}\not={\bf 0}}
\z_k^{\gamma(k)}\z_{\bf n}^{\gamma({\bf n})}
\end{align}
form a countable basis of $\mathbb{R}[\z_k,\z_{\bf n}]$. We will denote by $e_k$ and $e_\n$ the multi-indices such that $\z^{e_k} = \z_k$ and $\z^{e_\n}= \z_\n$, respectively.

\medskip

However, $\mathbb{R}[\z_k,\z_{\bf n}]$ is not preserved by the $\Gamma^*$
defined through (\ref{ao79}):
Taking $\pi^{(\0)}=v\in\mathbb{R}\setminus\{0\}$ and $\pi^{({\bf n})}=0$ for ${\bf n}\not=\0$, 
and considering the function $\pi$ $=\z_0$, we have $\Gamma^*\pi[a,p]=a(v)$.
Now $a(v)$ cannot be expressed as a finite linear 
combination of $\z^\gamma$'s.
Actually, it follows from Taylor's formula that $a(v)$ can be written as
\begin{align}\label{ao83}
a(v)=\sum_{k\ge 0}\z_k[a]\,v^k,
\end{align}
so that the function $\Gamma^*\z_0$ can be identified with a 
formal power series in the variables (\ref{ao20bis}), that is, an element of
$\gls{powerseries}$.

\medskip

Hence in coordinates, we a priori only know that (\ref{ao81})
defines an algebra morphism from $\mathbb{R}[\z_k,\z_{\bf n}]$
into the larger $\mathbb{R}[[\z_k,\z_{\bf n}]]$. Loosing
the endomorphism property of course obscures the group structure. 
We thus seek an extension of the above $\Gamma^*$'s to endomorphisms of 
$\mathbb{R}[[\z_k,\z_{\bf n}]]$. This will require restricting $\Gamma^*$ to a (linear) subspace $\T^*$ of $\R[[\z_k,\z_{\n}]]$, which amounts to the restriction of the space of $\pi$'s mentioned at the beginning of this subsection.

\subsection{Main result}

$\mbox{}$

Our main results are: \textit{i)} The goals outlined in Subsections \ref{Sect1.2},
\ref{Sect1.3}, and \ref{Sect1.4}
can be achieved, provided we restrict to a suitable subspace 
$\mathsf{T}^*\subset\mathbb{R}[[\z_k,\z_{\bf n}]]$ and restrict the admissible $\pi^{(\n)}$'s. 
\textit{ii)} The objects are dual to a regularity structure.
\textit{iii)} They arise from a natural Hopf algebra structure based on a pre-Lie algebra structure.

\begin{theorem}\label{maintheorem}
$\mbox{}$
For arbitrary yet fixed $\alpha>0$ introduce the homogeneity\footnote{See Subsection \ref{Sect3.6} for a motivation of this expression which is targeted to the application for \eqref{as34}.} of a multi-index $\gamma$ by 
\begin{align*}
&\lhom\gamma\rhom:=\alpha([\gamma]+1)+\sum_{{\bf n}\not=\bf 0}|{\bf n}|\gamma({\bf n}),\\
&\mbox{where}\quad[\gamma]:=\sum_{k\ge 0}k\gamma(k)-\sum_{{\bf n}\not=\bf 0}\gamma({\bf n}),
\end{align*}
and let\footnote{
The definition of $|\n|$ is related to the scaling of the differential operator. See Subsection \ref{Sectri} below.} 
$\N_0^2\ni \n\mapsto \gls{parabolicDistance} \in \R_+$ be additive and satisfy $|(1,0)|,|(0,1)|>0$.
Moreover, introduce the linear subspace of
$\mathbb{R}[[\z_k,\z_{\bf n}]]$
\begin{align*}
\mathsf{T}^*:=\{\,\pi\,|\,
\pi_\gamma=0\;\mbox{unless}\;[\gamma]\ge 0\;\mbox{or}
\;\gamma=e_{\bf n}\;\mbox{for some}\;{\bf n}\not=\bf 0\,\}.
\end{align*}
i) Suppose the tilt $\{\pi^{({\bf n})}\}_{{\bf n}}$ satisfies
\begin{align*}
\pi^{(\bf n)}\in\{\,\pi\,|\,
\pi_\gamma=0\;\mbox{unless}\;[\gamma]\ge 0\;\mbox{and}\;
\lhom\gamma\rhom>|{\bf n}|\,\}.
\end{align*}
Then $\Gamma^*$ defined through (\ref{ao81}) extends to an automorphism of
$\mathsf{T}^*$, which respects the algebra structure 
of the ambient $\mathbb{R}[[\z_k,\z_{\bf n}]]$, cf. \eqref{eqmult}.
The same holds true for the $\Gamma^*$ defined through (\ref{ao80}). 
These two types of $\Gamma^*$'s generate a group 
$\mathsf{G}^*\subset{\rm End}(\mathsf{T}^*)$ consistent with (\ref{as33}).
As a set, $\mathsf{G}^*$ is parameterized by
a shift $h\in\mathbb{R}^2$ and 
a tilt $\{\pi^{({\bf n})}\}_{\bf n}$ through an exponential formula\footnote{
which is distinct from the matrix exponential in ${\rm End}(\mathsf{T}^*)$}, 
cf. \eqref{LOsg3}.

\medskip

ii) There exists a linear space $\mathsf{T}$ of which $\mathsf{T}^*$ is the algebraic dual; 
for every $\Gamma^*\in\mathsf{G}^*$
there exists a $\Gamma\in{\rm End}(\mathsf{T})$ of which $\Gamma^*$ is the dual, 
thereby defining a group $\mathsf{G}\subset{\rm End}(\mathsf{T})$. 
Letting $\mathsf{A}:=(\alpha\mathbb{N}_0+\mathbb{N}_0)\setminus\{0\}$,
the triple $(\mathsf{A},\mathsf{T},\mathsf{G})$ forms a regularity structure.

\medskip

iii) Consider the Lie algebra $\mathsf{L}\subset{\rm End}(\mathsf{T}^*)$ spanned
by the infinitesimal generators of shift and tilt. 
Consider its universal enveloping algebra ${\rm U}(\mathsf{L})$ with its canonical algebra morphism
${\rm U}(\mathsf{L})\to{\rm End}(\mathsf{T}^*)$. There exists a non-degenerate pairing between ${\rm U}(\mathsf{L})$
and a linear space $\mathsf{T}^+$ such that the Hopf algebra structure on 
${\rm U}(\mathsf{L})$ defines a Hopf algebra structure on $\mathsf{T}^+$.
Likewise, the pairing allows to lift the action given through the algebra morphism 
${\rm U}(\mathsf{L})\to{\rm End}(\mathsf{T}^*)$ to a coaction $\Delta\colon\mathsf{T}
\rightarrow\mathsf{T}^+\otimes\mathsf{T}$.
In line with regularity structures, the group 
$\mathsf{G}\subset{\rm End}(\mathsf{T})$ then arises from the Hopf algebra structure
of $\mathsf{T}^+$ together with $\Delta$.
The exponential formula arises from choosing a specific basis in
${\rm U}(\mathsf{L})$, which is based on a pre-Lie algebra structure on $\mathsf{L}$. This basis determines the pairing and ensures the intertwining of $\Delta^+$ and $\Delta$ modulo the re-centering maps $\mathcal{J}_{\n}$, cf. \eqref{Hai4.14}.

\end{theorem}

\subsection{Outline of the paper}

$\mbox{}$

Section \ref{Sect3} introduces and motivates the 
main objects. More precisely, in Subsection \ref{Sect3.2},
we will introduce the infinitesimal generators
of shift $\{\partial_i\}_{i=1,2}$
and (constant) tilt $\{D^{({\bf n})}\}_{{\bf n}\in\mathbb{N}_0^2}$
as derivations on the algebra $\mathbb{R}[[\z_k,$ $\z_{\bf n}]]$.
In Subsection \ref{Sect3.3}, the polynomial sector\footnote{in the jargon of
regularity structures} $\mathsf{\bar T}$ will be defined;
in Subsection \ref{Sect3.4}, we define 
the space $\mathsf{T}^*\subset\mathbb{R}[[\z_k,\z_{\bf n}]]$, which turns out to be dual to the abstract model space $\mathsf{T}$.
The corresponding mapping properties of $\{\partial_i\}_{i}$, $\{D^{({\bf n})}\}_{{\bf n}}$, and their respective transposed versions are characterized. 
In Subsection \ref{Sect3.5}, we point out that the commutators of
$\{\partial_i\}_{i}$ and $\{D^{({\bf n})}\}_{{\bf n}}$ behave in the same
way shift and tilt operators would act on polynomials including the constants\footnote{This subsection is logically not needed, but provides a key intuition.}.
In Subsection \ref{Sect3.8}, we extend from constant to variable tilt parameters $\pi^{({\bf n})}$,
in form of monomials $\z^\gamma$,
by introducing the infinitesimal generator $\z^\gamma D^{({\bf n})}$
of variable tilt.
In Subsection \ref{Sect3.10}, we explore the natural pre-Lie algebra structure of the set of generators and a bigrading. 
In Subsection \ref{Sect3.6}, the homogeneity $\lhom\gamma\rhom$
of a multi-index $\gamma$ and thus the set of
homogeneities $\mathsf{A}$ and the ensuing grading of $\mathsf{T}$ will be introduced. 
In Subsection \ref{Sect3.7}, we define the Lie algebra $\mathsf{L}$ as the subspace
of ${\rm End}(\mathsf{T}^*)$ spanned by $\{\partial_i\}_{i=1,2}$
and $\{\z^\gamma D^{({\bf n})}\}_{\lhom\gamma\rhom>|{\bf n}|}$.

\medskip

While Section \ref{Sect3} is mostly about definitions and elementary
properties, Section \ref{Sect4} states the main, partially technical, results
that require a proof. 
In Subsection \ref{Sect4.1}, we appeal to the general theory of Hopf algebras:
We consider the universal enveloping algebra ${\rm U}(\mathsf{L})$ of the
Lie algebra $\mathsf{L}$, which is obtained from the tensor algebra factorized by
the ideal generated by the Lie bracket, and which naturally is a Hopf algebra.
Moreover, since $\mathsf{L}\subset{\rm End}(\mathsf{T}^*)$, 
there is a canonical algebra morphism
$\rho:{\rm U}(\mathsf{L})\to{\rm End}(\mathsf{T}^*)$ and the concatenation
product on ${\rm U}(\mathsf{L})$ coincides with the composition in ${\rm End}(\mathsf{T}^*)$.
This action naturally defines a (left) module structure
${\rm U}(\mathsf{L})\otimes\mathsf{T}^*\rightarrow\mathsf{T}^*$. In Subsection \ref{SectAd}, the pre-Lie product of Subsection \ref{Sect3.10} is extended to an operation of $\{\z^\gamma D^{(\n)}\}_{|\gamma|>|\n|}$ on ${\rm U}(\mathsf{L})$, cf.~\eqref{yc11}, which is shown to be consistent with the Hopf algebra structure. This operation will allow us, in Subsection \ref{Sect4.2}, to select a basis that is natural, but different from
the typical bases considered in the Poincar\'e-Birkhoff-Witt theorem, cf. \eqref{ao75}. Such a basis also provides
a non-degenerate pairing between ${\rm U}(\mathsf{L})$ and a space $\mathsf{T}^+$,
see (\ref{dual01}), which is introduced in Subsection \ref{Sect4.3}. 
Under this pairing, the coproduct on ${\rm U}(\mathsf{L})$ turns into a product
on $\mathsf{T}^+$ that allows to identify $\mathsf{T}^+$ with the polynomial algebra
in variables indexed by the index set of $\mathsf{L}$, cf. \eqref{cop01}.

\medskip

Next, we embark on the more subtle part of the dualization. 
This heavily relies on finiteness properties stated in \eqref{SG06} and \eqref{SG03}, 
which in turn are an outcome of extending the bigrading of $\mathsf{L}$ to ${\rm U}(\mathsf{L})$; 
this is carried out in Subsection \ref{Sect4.5}. 
In order to obtain these finiteness properties, it is crucial to pass from $\R[[\z_k,\z_\n]]$ to $\T^*$. 
As a consequence, the action and the product of ${\rm U}(\mathsf{L})$,
in terms of their coordinate representation with respect to our basis, 
turn into coaction and coproduct, respectively, 
for the couple $\mathsf{T}$ and $\mathsf{T}^+$, see Proposition \ref{cor02}.
More precisely, the action ${\rm U}(\mathsf{L})\otimes\mathsf{T}^*\rightarrow\mathsf{T}^*$
gives rise to a coaction $\Delta\colon\mathsf{T}\rightarrow\mathsf{T}^+\otimes\mathsf{T}$,
and the concatenation product
${\rm U}(\mathsf{L})\otimes{\rm U}(\mathsf{L})\rightarrow{\rm U}(\mathsf{L})$
gives rise to a coproduct $\Delta^+\colon\mathsf{T}^+\rightarrow\mathsf{T}^+\otimes\mathsf{T}^+$.
In particular, $\mathsf{T}^+$ carries the structure of a (graded connected) Hopf algebra.
In Subsection \ref{Sect5.6} we argue that $\Delta$ and $\Delta^+$ intertwine as postulated by regularity structures.

\medskip

Section \ref{Sect5old} deals with the group structure and connects to the goals of Theorem \ref{maintheorem}. In Subsection \ref{Sect5.1old}, we apply general Hopf algebra theory to $\mathsf{T}^+$. This allows to endow the space of multiplicative linear forms 
${\rm Alg}(\mathsf{T}^+,\mathbb{R})\subset(\mathsf{T}^+)^*$ with a group structure,
with help of the (convolution) product 
coming from the coproduct $\Delta^+$.
Together with the coaction $\Delta$, this gives rise to our 
$\mathsf{G}\subset{\rm End}(\mathsf{T})$, establishing part \textit{iii)} of Theorem \ref{maintheorem}.
We also state that $\mathsf{G}$ is consistent with the requirements
of regularity structures with respect to the gradedness of $\mathsf{T}$, cf. \eqref{tri01}, and the polynomial sector $\mathsf{\bar T}$, cf. \eqref{pol01}. 
This concludes part \textit{ii)} of Theorem \ref{maintheorem}.
In Subsection \ref{Sect5.2old}, we connect back to Section \ref{Sect1} by establishing part
\textit{i)} of Theorem \ref{maintheorem}. 
Namely, we show that the $\Gamma^*$'s extend the definition 
(\ref{ao81}) (Proposition \ref{exp01} \textit{ii)}), that they respect the
algebra structure of $\mathbb{R}[[\z_k,\z_{\bf n}]]$ (Proposition \ref{exp01} \textit{v)}), 
and that they respect the group structure (\ref{as33}) (Proposition \ref{exp01} \textit{vi)}). In Subsection \ref{Sect.ext} we enlarge our regularity structure to meet exactly\footnote{
not just up to the constant part, and including an abstract integration map} 
Hairer's axioms, 
showing that our smaller structure contains all the relevant information.

\medskip

At this stage, the reader may wonder why
the passage from the Lie algebra $\mathsf{L}\subset{\rm End}(\mathsf{T}^*)$ 
to the corresponding group $\mathsf{G}^*\subset{\rm Aut}(\mathsf{T}^*)$, which both live
on the dual side, has to pass via the primal side in form of $\mathsf{T}$ and $\mathsf{T}^+$.
The reason resides in our purely Hopf-algebraic approach, which prevents us from appealing to the matrix exponential in ${\rm End}(\mathsf{T}^*)$
that analytically links the Lie algebra $\mathsf{L}$ to its (Lie) group $\mathsf{G}^*$,
even if, as in our case, the exponential sum is effectively\footnote{meaning
	that it is finite for a given matrix element} finite because of gradedness.
Indeed, the universal enveloping algebra ${\rm U}(\mathsf{L})$, as {\it finite} linear combinations
of products of elements of $\mathsf{L}$, is obviously too small to contain matrix exponentials
of elements of $\mathsf{L}$, even if $\mathsf{T}^*$ were finite dimensional. 
First passing to $\mathsf{T}^+$, which as a linear space is isomorphic to ${\rm U}(\mathsf{L})$,
and then to its algebraic dual $(\mathsf{T}^+)^*$, 
which as a linear space is much larger than 
${\rm U}(\mathsf{L})$, is an algebraic way of extending ${\rm U}(\mathsf{L})$.
It turns out to contain the matrix exponentials of\footnote{
The elements $D^\dagger$ are well defined by \eqref{ao23}, since we have the stronger \eqref{yourchoice1}.}
$\{D^\dagger \,|\, D\in\mathsf{L}\subset\End(\T^*)\}$ 
$\subset{\rm End}(\mathsf{T})$, 
namely in form of ${\rm Alg}(\mathsf{T}^+,\mathbb{R})\subset(\mathsf{T}^+)^*$
as seen through the coaction $\Delta$.
Note that the primal $\mathsf{G}$ is more valuable than its dual $\mathsf{G}^*$,
since one may always pass from $\Gamma\in{\rm End}(\mathsf{T})$ to
$\Gamma^*\in{\rm End}(\mathsf{T}^*)$, while the opposite is only possible in finite
dimensions.

\medskip

Sections \ref{Sect5} and \ref{Sect6} are logically independent of the rest of the paper, but connect the combinatorial structures to our Lie-geometric construction. More precisely, we make this connection in the well-studied cases of branched rough paths \eqref{drivenODE} and the stochastic heat equation \eqref{SHE}. 
	
	\medskip
	
	In Section \ref{Sect5}, we consider the driven ODE example \eqref{drivenODE}. Here at a fixed $a$ the solution manifold is parameterized by $\R$, the space of initial conditions. This allows us to restrict the $(a,p)$-space to the space of nonlinearities $a$, hence the index set to multi-indices $\gamma$ over $\{k\geq 0\}$, and thus the Lie algebra to $\{\z^\gamma D^{(\0)}\}$. We show that the construction of Sections \ref{Sect3} and \ref{Sect4} is compatible with the Connes-Kreimer Hopf algebra underlying branched rough paths \cite{Gub10,HairerKelly}. More specifically, we associate our multi-indices \eqref{ao85} with linear combinations of (undecorated) trees in the Connes-Kreimer framework via a map $\phi$, and show that it is a Hopf algebra morphism. To do so, Lemma \ref{lembrp} establishes a pre-Lie algebra morphism property of the transpose of $\phi$ with respect to the grafting pre-Lie product (after a suitable normalization), which is at the core of the construction of Connes and Kreimer \cite{ConnesKreimer,Hoffman}. This morphism property is shown to be related to the one involving the map $\Upsilon$ in \cite{BonnefoiCMW} in Subsection \ref{Sect5.5}. In addition, in Subsection \ref{Sect6.6} we discuss how renormalization fits our setting and show that the transpose of $\phi$ intertwines with the translation maps of rough paths defined in \cite{BCFP19}, cf. Lemma \ref{lem6.5}. 
	
	\medskip
	
	In Section \ref{Sect6}, we connect to tree-based regularity structures \cite{Hairer14, Hairer, BCCH19} for SHE. 
	This example is simple in the sense that it only involves one integration kernel so that no edge decorations appear in the tree-based approach. We adopt the two perspectives in \cite{BCCH19} when it comes to the treatment of the polynomial decorations. In Subsections \ref{Sect7.1} to \ref{Sect7.3}, we consider a detailed description of the polynomials, in line with the space $\mathscr{B}$ in \cite[Subsection 4.1]{BCCH19}, and show that our Lie algebra $\mathsf{L}$ reflects the grafting operations in \cite[Definition 4.7]{BCCH19}; as in Section \ref{Sect5}, a pre-Lie morphism property of our dictionary is shown, connecting to the morphism properties of $\Upsilon$ in \cite{BCCH19}. In Subsection \ref{Sect7.4}, we relate to the coarser description which contracts all polynomials by multiplication giving rise to the model space $\T_H$; the morphism $\phi$ between $\T$ and $\T_H$ will no longer be one-to-one. In Proposition \ref{lemrs} we establish that $\phi$ induces a morphism $\Phi$ between $\T^+$ and the Hopf algebra $\mathsf{T}_H^+$ in \cite{Hairer} (without appealing to a pre-Lie structure).


\section{The Lie algebra structure}\label{Sect3}

\subsection{Duality and transposition $\dagger$}\label{Sect3.1}
\mbox{}

Recall that the monomials $\gls{monomial}$, defined in \eqref{ao85}, can also be considered as elements of $\R[[\z_k,\z_\n]]$ and that $\R[[\z_k,\z_\n]]$ is the direct product over the index set of multi-indices $\gamma$.
Denoting the direct sum over the same index set\footnote{
which is isomorphic to $\R[\z_k,\z_\n]$, but we choose a different notation to distinguish $\R[\z_k,\z_\n]$ as a subspace of the space of formal power series from $\R[[\z_k,\z_\n]]^\dagger$ as the pre-dual of the space of formal power series.}
by $\R[[\z_k,\z_\n]]^\dagger$ and its basis elements -- to which the monomials of $\R[[\z_k,\z_\n]]$ are dual -- by $\z_\beta$, we have $(\R[[\z_k,\z_\n]]^\dagger)^*=\R[[\z_k,\z_\n]]$ with the canonical pairing 
\begin{equation}\label{pl95}
\langle\z^\gamma,\z_{\beta}\rangle= \delta_{\beta}^\gamma.
\end{equation}

For $D\in{\rm End}(\mathbb{R}[[\z_k,\z_{\bf n}]])$ we may consider the components of the sequence $D\z^\gamma$ as a  
matrix representation $\{D_\beta^\gamma\}_{\beta,\,\gamma}$. Since the $D$'s we construct below will have the finiteness property 
\begin{align}\label{ao64}
\{\,\beta\,|\,D_\beta^\gamma\not=0\,\}\;\mbox{is finite for all $\gamma$},
\end{align}
we may write $D\z^\gamma = \sum_\beta D_\beta^\gamma \, \z^\beta$. 
Moreover, they also will have the dual finiteness property 
\begin{align}\label{ao23}
\{\,\gamma\,|\,D_\beta^\gamma\not=0\,\}\;\mbox{is finite for all $\beta$}.
\end{align}
This second finiteness property is just the one needed to have a unique 
$D^{\dagger}\in{\rm End}(\mathbb{R}[[\z_k,\z_{\bf n}]]^\dagger)$ such that $(D^\dagger)^*=D$; 
on the level of the matrix representation this just means
\begin{equation}\label{ao107}
D^{\gls{dagger}} \z_\beta = \sum_\gamma D_\beta^\gamma \, \z_\gamma.
\end{equation}

\medskip

Passing from the polynomial space $\R[\z_k,\z_\n]$ to the formal power series space
	$\mathbb{R}[[\z_k,\z_{\bf n}]]$ serves us well 
	in the actual application, cf. \cite{LOTT,LO}, where we think of (\ref{ao20bis}) as coordinates on the manifold of all solutions $u$. As a consequence, the general solution $u$ can be seen as a function of the variables $\mathsf{z}_k,\mathsf{z}_{\bf n}$ (with values in the space of functions of $x$), or rather as a formal power series in these variables, described by its coefficients $\Pi_\beta$ indexed by our multi-indices $\beta$. Formally, $\Pi_\beta$ is, up to the combinatorial factor $\beta!$, a partial derivative of the general solution $u$ w.~r.~t.~the above variables. By Leibniz' rule, \eqref{as34} gives rise to the following family of linear equations indexed by $\beta$:
	\begin{equation}\label{extra1}
		\big(\frac{\partial}{\partial x_2}-\frac{\partial^2}{\partial x_1^2}\big)\Pi_{\beta} = \sum_{k \geq 0}\sum_{e_k+\beta_1+\cdots+\beta_{k+1}=\beta}\Pi_{\beta_1}\cdots \Pi_{\beta_k}\frac{\partial^2}{\partial x_1^2} \Pi_{\beta_{k+1}} + \delta_\beta^0 \xi.
	\end{equation}
	It turns out that this $\{\Pi_\beta\}_\beta$ can be inductively rigorously constructed\footnote{Here, we think of a smoothed-out noise so that the model
		is not a distribution but actually a smooth function in $x$.}, and thus naturally gives rise to the
formal power series  $\Pi(x)\in\mathbb{R}[[\z_k,\z_{\bf n}]]$.
In particular, the algebra structure of $\mathbb{R}[[\z_k,\z_{\bf n}]]$,
which will contain the dual $\mathsf{T}^*$ of the abstract model space $\mathsf{T}$, 
is inherent to our approach.
This dual perspective is consistent with the 
definition of the model as a distribution with values in $\mathsf{T}^*$ \cite[Definition 3.3]{Hairer}.

\medskip

In the solution theory of regularity structures, one considers truncations of this formal power series as approximations of the actual solution; such truncations, in turn, shall be seen as (coherent) modelled distributions, in the language of \cite{Hairer14,BCCH19}, after the application of the model\footnote{ We invite the reader to check \cite{OSSW}, where these objects are used to obtain a priori bounds for \eqref{as34}.} $\Pi$. 
\subsection{The infinitesimal generators of shift $\{\partial_i\}_i$ and constant tilt $\{D^{(\n)}\}_\n$}\label{Sect3.2}
\mbox{}

We now come to the definition of those derivations\footnote{i.e. satisfying $D\pi\pi'=(D\pi)\pi'+\pi D\pi'$ for $\pi,\pi'\in\R[[\z_k,\z_\n]]$.}
$D\in{\rm End}(\mathbb{R}[[\z_k,\z_{\bf n}]])$ that form the building blocks for the Lie algebra $\mathsf{L}$. 
The definitions capture the infinitesimal generators of the actions of shift and tilt
on $(a,p)$-space,
see (\ref{ao27bis}) and (\ref{ao29bis}) in Section \ref{Sect1}.
Let us now introduce the main building blocks $D^{({\bf 0})}$, $\{ \gls{generatorTilt} \}_{\n\neq \0}$ and $\{ \gls{generatorShift} \}_{i=1,2}$. 
We start with $D^{({\bf 0})}$, which is to capture 
the action of $\mathbb{R}$ onto $(a,p)$-space by tilt by constants, which in view of \eqref{ao29bis} amounts to a shift of $u$-space
\begin{align}\label{ao28}
(a,p)\mapsto (a(\cdot+v),p-v).
\end{align}
Here, the action on the $p$-component, which is made such that $a\circ p$ stays invariant, 
is immaterial because we ``mod out'' constants.
As for the shift (\ref{ao27bis}) of $x$, this action lifts by pull-back to functions
$\pi$ of $(a,p)$.
We formally define $D^{({\bf 0})}$ to be the infinitesimal generator of this action\footnote{Here is yet another characterization of $D^{(\0)}$: For arbitrary $u\in\R$ consider $\pi,\pi'\in\R[[\z_k,\z_\n]]$ specified through $\pi[a,p] = a(u)$ and $\pi'[a,p] = \frac{da}{dv}(u)$; they are related by $\pi'=D^{(\0)}\pi$.}
\begin{align}\label{ao21}
D^{({\bf 0})}\pi[a,p]=\frac{d}{dv}_{|v=0}\pi[a(\cdot+v),p-v].
\end{align}
This can be given sense for $\pi\in\{\z_k,\z_{\bf n}\}$, cf. (\ref{ao20bis}),
and yields 
\begin{align}\label{ao26}
D^{({\bf 0})}\z_k=(k+1)\z_{k+1},\quad D^{({\bf 0})}\z_{\bf n}=0.
\end{align}
In addition, (\ref{ao21}) suggests that $D^{({\bf 0})}$ is a derivation, which we postulate.
This and (\ref{ao26}) 
yield that, on the space $\mathbb{R}[\z_k,\z_{\bf n}]$, $D^{({\bf 0})}$ assumes
the form
\begin{align}\label{ao35}
D^{({\bf 0})}=\sum_{k\ge 0}(k+1)\z_{k+1}\partial_{\z_{k}};
\end{align}
the sum is obviously effectively finite on $\mathbb{R}[\z_k,\z_{\bf n}]$. 
From (\ref{ao35}) we infer the matrix representation with respect to the monomial basis:
\begin{align}\label{ao24}
(D^{({\bf 0})})_\beta^\gamma=\sum_{k\ge 0}\left\{\begin{array}{cl}
(k+1)\gamma(k)&\mbox{if}\;\gamma+e_{k+1}=\beta+e_k\\
0&\mbox{otherwise}\end{array}\right\}.
\end{align}
Note that the finiteness property (\ref{ao23}) is satisfied
so that $D^{({\bf 0})}$ is well defined as a derivation of $\mathbb{R}[[\z_k,\z_{\bf n}]]$. 
We also see that the finiteness property \eqref{ao64} holds, so that \eqref{ao107} defines an endomorphism $(D^{(\0)})^\dagger$ on $\R[[\z_k,\z_\n]]^\dagger$.
 
\medskip

After this representation (\ref{ao24}) of infinitesimal shifts of $u$,
we turn to the shifts of space $x_1$ and time $x_2$, that is, the action (\ref{ao27bis})
of $\mathbb{R}^2$ on $(a,p)$-space. Again, this action extends by pull-back to functions $\pi$ on $(a,p)$,
see (\ref{ao80}). We formally consider its infinitesimal generators\footnote{For arbitrary $x\in\R^2$ consider $\pi,\pi' \in \R[[\z_k,\z_\n]]$ characterized through $\pi[a,p] = p(x)$ and $\pi'[a,p] = \frac{\partial p}{\partial x_1}(x)$; they are related by $\pi' = \partial_1 \pi$.}
\begin{align}\label{ao25}
\begin{split}
\partial_1\pi[a,p]&=\tfrac{d}{dy_1}_{|y_1=0}\pi\big[a\big(\cdot+p(y_1,0)\big),
p\big(\cdot+(y_1,0)\big)-p(y_1,0)\big],\\
\partial_2\pi[a,p]&=\tfrac{d}{dy_2}_{|y_2=0}\pi\big[a\big(\cdot+p(0,y_2)\big),
p\big(\cdot+(0,y_2)\big)-p(0,y_2)\big].
\end{split}
\end{align}
By the chain rule and (\ref{ao26}) for $\z_k$,
and using the same argument (with $p$ playing the role of $a$)
that led to (\ref{ao26}) for $\z_{\bf n}$, we formally derive
\begin{align*}
\partial_1\z_k=\z_{(1,0)}D^{({\bf 0})}\z_k,\quad
\partial_1\z_{\bf n}=(n_1+1)\z_{{\bf n}+(1,0)},
\end{align*}
which we now postulate. Together with the postulate that $\partial_1$ be a derivation,
this implies that on the sub-algebra $\mathbb{R}[\z_k,\z_{\bf n}]$
we have
\begin{align}\label{ao30}
\partial_1=\sum_{{\bf n}}(n_1+1)\z_{{\bf n}+(1,0)}D^{({\bf n})}\quad\mbox{with}
\;D^{({\bf n})}:=\partial_{\z_{\bf n}}\;\mbox{for}\;{\bf n}\not={\bf 0}.
\end{align}
%
The notation $D^{({\bf n})}$ 
$=\partial_{\z_{\bf n}}$ is redundant, but very convenient; we obviously have
the matrix representation
\begin{align}\label{ao41}
(D^{({\bf n})})_\beta^\gamma=\left\{\begin{array}{cl}
\gamma({\bf n})&\mbox{if}\;\gamma=\beta+e_{\bf n}\\
0&\mbox{otherwise}\end{array}\right\}\quad\mbox{for}\;{\bf n}\not={\bf 0}.
\end{align}
Incidentally, still for ${\bf n}\not={\bf 0}$, we have
\begin{align}\label{hk05}
D^{({\bf n})}\pi[a,p]=\frac{d}{dt}_{|t=0}\pi[a,p+tx^{\bf n}],
\end{align}
which can be given a sense as an endomorphism on both 
$\mathbb{R}[[\z_k,\z_{\bf n}]]$ and $\mathbb{R}[\z_k,\z_{\bf n}]$.

\medskip

Inserting (\ref{ao24}) and (\ref{ao41}) into (\ref{ao30}) we obtain the matrix representation
\begin{align} \label{ao42}
\begin{split}
(\partial_1)_\beta^\gamma&=\sum_{k\ge 0}\left\{\begin{array}{cl}
(k+1)\gamma(k)&\mbox{if}\;\gamma+e_{k+1}+e_{(1,0)}=\beta+e_k \\
0&\mbox{otherwise}\end{array}\right\} \\
&+\sum_{{\bf n}\not=\0}\left\{\begin{array}{cl}
(n_1+1)\gamma({\bf n})&\mbox{if}\;\gamma+e_{{\bf n}+(1,0)}=\beta+e_{\bf n}\\
0&\mbox{otherwise}\end{array}\right\};
\end{split}
\end{align}
we again learn from (\ref{ao42}) that
$\partial_1$ satisfies the finiteness properties \eqref{ao64} and (\ref{ao23}). 
Hence $\partial_1$ is also well defined as an element of $ {\rm End}(\mathbb{R}[[\z_k,\z_{\bf n}]])$ and $\partial_1^\dagger$ as an element of ${\rm End}(\mathbb{R}[[\z_k,\z_{\bf n}]]^\dagger)$.
The same applies to $\partial_2$, with $(1,0)$ and $n_1$ replaced by $(0,1)$ and $n_2$ in \eqref{ao42}.

\medskip

We now have defined the building blocks, which are 
derivations on $\mathbb{R}[[\z_k,\z_{\bf n}]]$
\begin{align}\label{ao37}
\{D^{({\bf n})}\}_{{\bf n}} \cup \{\partial_i\}_{i}
\;\subset\;{\rm Der}(\mathbb{R}[[\z_k,\z_{\bf n}]])
\end{align}
satisfying the finiteness properties (\ref{ao64}) and (\ref{ao23}).

\subsection{The polynomial sector $\bar{\T}$}\label{Sect3.3}
\mbox{}

In view of the second item of (\ref{ao20bis}), which identifies the coordinate
$\z_{\bf n}=\z^{e_\n}$
of $\mathbb{R}[[\z_k,\z_{\bf n}]]$ 
with the derivative
$\frac{1}{{\bf n}!}\frac{d^{\bf n}}{dx^{\bf n}}$, it is natural to identify
the element $\z_{e_{\bf n}} \in \mathbb{R}[[\z_k,\z_{\bf n}]]^\dagger$ 
with the polynomial $x^{\bf n}=x_1^{n_1}x_2^{n_2}$.
Hence we identify 
\begin{equation}\label{ao106}
\gls{polynomialsector} := \textnormal{span}  \{\z_{e_{\bf n}}\}_{{\bf n}\not=\0} 
\subset \mathbb{R}[[\z_k,\z_{\bf n}]]^\dagger
\end{equation}
with $\mathbb{R}[x_1,x_2]/\mathbb{R}$,
the space of polynomials in the variables $x_1,x_2$ quotiented by the constants.
Following \cite[Assumption 3.20]{Hairer}, 
we call $\mathsf{\bar T}$ the polynomial sector.
We note that the transposed endomorphisms of (\ref{ao37}) 
preserve this polynomial sector\footnote{which is the Lie algebra version of the corresponding postulate
	on $\Gamma\in\mathsf{G}$ in \cite[Assumption 3.20]{Hairer}}
	 $\mathsf{\bar T}$
\begin{align}\label{ao57}
D^\dagger\mathsf{\bar T}\subset\mathsf{\bar T}\quad\mbox{for}\quad
D\in\{D^{({\bf n})}\}_{{\bf n}} \cup \{\partial_i\}_{i},
\end{align}
which on the level of the matrix representation amounts to
\begin{align}\label{ao56}
D_\beta^\gamma=0\quad\mbox{for}\quad\beta\in\{e_{\bf n}\}_{{\bf n}\not=\0}
\quad\mbox{and}\quad\gamma\not\in\{e_{\bf n}\}_{{\bf n}\not=\0},
\end{align}
and can be inferred from (\ref{ao24}), (\ref{ao41}), and (\ref{ao42}). 
We note that $\partial_1^\dagger,\partial_2^\dagger$ almost act as partial 
derivatives on the polynomial sector\footnote{which is the Lie algebra version of the corresponding 
	Lie group postulate in \cite[Assumption 3.20]{Hairer}} $\mathsf{\bar T}$, which (in case of $\partial_1^\dagger$) means\footnote{where $\m<\n$ means $\m\leq \n$ component-wise and $\m\neq\n$}
\begin{equation}\label{ao51}
\partial_1^\dagger x^\n =
\left\{\begin{array}{cl}
n_1 x^{{\bf n}-(1,0)}&\mbox{if}\;\n > (1,0)\\
0&\mbox{otherwise}\end{array}\right\}\quad
\mbox{for}\;{\bf n}\not=\0.
\end{equation}
In terms of the matrix representation, this means
\begin{align*}
(\partial_1)_{e_{{\bf n}}}^\gamma
=\left\{
\begin{array}{cl}
n_1&\mbox{if}\;\gamma=e_{{\bf n}-(1,0)},~ \n > (1,0)\\
0  &\mbox{otherwise}
\end{array}
\right\},
\end{align*}
which in turn can be read off from (\ref{ao42}). The reason why the case $\n = (1,0)$ (and analogously for $\partial_2$ the case $\n=(0,1)$) is excluded is that we modded out constants in the polynomial $p$, cf. \eqref{as35}; see however the upcoming Subsection \ref{Sect3.5}.

\subsection{Commutators of $\{D^{(\n)}\}_\n$ and $\{\partial_i\}_i$ behave naturally} \label{Sect3.5}
\mbox{}

We now make a connection between $\{D^{(\n)}\}_{\n}\cup \{\partial_i\}_{i}$ and the classical Lie algebra of tilt and shift on polynomials.
We start noting that
\begin{align}\label{ao36}
[D^{({\bf n})},D^{({\bf n}')}]=0,
\end{align}
which is obvious in case of ${\bf n}\not=\0$ and ${\bf n}'\not=\0$,
and can be easily inferred from (\ref{ao35}) for ${\bf n}\not=\0$ and ${\bf n}'=\0$.
We next argue that (\ref{ao25}) implies
\begin{align}\label{ao67}
[\partial_1,\partial_2]=0.
\end{align}
Indeed, by the finiteness property (\ref{ao64}), the monomial
$\z^\gamma$ is mapped by $\partial_1,\partial_2$ onto finite
linear combinations of monomials. Hence we may indeed appeal to
(\ref{ao25}) when computing $(\partial_1\partial_2-\partial_2\partial_1)\z^\gamma$,
which shows that this expression vanishes by the symmetry of second derivatives.
Turning to the commutator between $D^{(\n)}$'s and $\partial_i$'s, we first observe that by the characterization
(\ref{ao30}) and the commutation relation (\ref{ao36}) we have $[D^{(\0)},\partial_1]=0$ 
by the second item in (\ref{ao26}). Likewise,  
for ${\bf n}\not=\0$, we have $[D^{({\bf n})},\partial_1]=n_1D^{({\bf n}-(1,0))}$ 
(with the understanding that this expression vanishes if $n_1=0$) 
by the second item in (\ref{ao30}). We retain that
\begin{equation}\label{ao104}
[D^{({\bf n})},\partial_1] = n_1D^{({\bf n}-(1,0))} \textnormal{ and }
[D^{({\bf n})},\partial_2] = n_2D^{({\bf n}-(0,1))} \textnormal{ for all }{\bf n}.
\end{equation}

\medskip

The identities (\ref{ao36}), (\ref{ao67}), and (\ref{ao104}) mean
that the derivations $\partial_1$, $\partial_2$, $\{D^{({\bf n})}\}_{{\bf n}}$
on $\mathbb{R}[[\z_k,\z_{\bf n}]]$,
when it comes to their commutators, precisely behave like certain endomorphisms on 
${\mathbb R}[x_1,x_2]$.
The important fact here is that this is the full space ${\mathbb R}[x_1,x_2]$, 
not just the space ${\mathbb R}[x_1,x_2]/\mathbb{R}$ with the constants factored out,
which was our starting point in Section \ref{Sect1}. The corresponding endomorphisms on
${\mathbb R}[x_1,x_2]$ are just the infinitesimal generators of shift and tilt.
In particular, the subtle $D^{({\bf 0})}$ has a simple analogue in the infinitesimal generator
of the ``tilt'' by a constant polynomial. 
This shows that the incorporation
of constants into the $a$-part in (\ref{ao27bis}) and (\ref{ao29bis}) did not lead
to a loss of information.

\medskip 

\subsection{Triangular structure}\label{Sectri}
\mbox{}

We now point out that the building blocks
$\{D^{({\bf n})}\}_{{\bf n}}\cup \{\partial_i\}_{i}$
are strictly triangular
with respect to the following two additive functionals on multi-indices $\gamma$
\begin{align}\label{ao50}
\gls{noisehomogeneity} :=\sum_{k\ge 0}k\gamma(k)-\sum_{{\bf n}\not=\0}\gamma({\bf n}) \quad\mbox{and}\quad \sum_{{\bf n}\not=\0}|{\bf n}|\gamma({\bf n}).
\end{align}
Here $\mathbb{N}_0^2\ni{\bf n} \mapsto|{\bf n}|\in\mathbb{N}_0$ denotes a scaled length of $\n$.
In applications to parabolic equations of the form \eqref{as34}, one considers $|\n|=n_1+2n_2$; 
in general, $|\n|$ is an additive and coercive map which is determined by the scaling of the differential operator.
We note that the combination of $\sum_{k\ge 0}k\gamma(k)$ and $\sum_{{\bf n}\not=\0}\gamma({\bf n})$
in \eqref{ao50} is natural: 
Like we identified $\z_{e_{\bf n}}$ with $p(x)=x^{\bf n}$ 
at the beginning
of Subsection \ref{Sect3.3}, we may identify $\z_{e_k}$ with $a(u)=u^k$; 
hence while $\sum_{k\ge 0}k\gamma(k)$ measures the homogeneity in the $u$-variable,
$\sum_{{\bf n}\not=\0}\gamma({\bf n})$ measures the homogeneity in the polynomial $p$;
$u$ and $p$-values have the same ``physical'' dimension. 
Considering the difference
$[\gamma]$ is forced upon us by the following.

\begin{lemma}\label{Lem3.1} It holds for ${\bf n}\not=\0$
\begin{equation}\label{ao40}
\begin{split}
\lefteqn{(D^{({\bf 0})})_\beta^\gamma\not=0}\\
&\Longrightarrow\quad
\left\{\begin{array}{ccc}
[\gamma]+1&=&[\beta]\quad\mbox{and}\\
\sum_{{\bf m}\not=\0}|{\bf m}|\gamma({\bf m})&=&\sum_{{\bf m}\not=\0}|{\bf m}|\beta({\bf m})
\end{array}\right\},
\end{split}
\end{equation}
\begin{equation}\label{ao47}
\begin{split}
\lefteqn{(D^{({\bf n})})_\beta^\gamma\not=0}\\
&\Longrightarrow\quad
\left\{\begin{array}{ccc}
[\gamma]+1&=&[\beta]\quad\mbox{and}\\
\sum_{{\bf m}\not=\0}|{\bf m}|\gamma({\bf m})&
=&\sum_{{\bf m}\not=\0}|{\bf m}|\beta({\bf m})+|{\bf n}|
\end{array}\right\},
\end{split}
\end{equation}
\begin{equation}\label{ao46}
\begin{split}
\lefteqn{(\partial_1)_\beta^\gamma\not=0}\\
&\Longrightarrow\quad
\left\{\begin{array}{ccc}
[\gamma]&=&[\beta]\quad\mbox{and}\\
\sum_{{\bf m}\not=\0}|{\bf m}|\gamma({\bf m})+|(1,0)|&
=&\sum_{{\bf m}\not=\0}|{\bf m}|\beta({\bf m})
\end{array}\right\}.
\end{split}
\end{equation}
\end{lemma}

Note that (\ref{ao40}) and (\ref{ao46}) are of similar character:
If the matrix element does not vanish, both functionals (\ref{ao50}) are ordered
and one of them is strictly ordered.
However, (\ref{ao47}) is of a different character, and we will get back to this in (\ref{ao55}).
\begin{proof}
Here comes the argument for (\ref{ao40}): From
(\ref{ao24}) we read off that $(D^{(\0)})^\gamma_\beta$ $\not=0$ implies $\gamma+e_{k+1}=\beta+e_k$
for some $k\ge0$, which by (\ref{ao50}) yields as desired $[\gamma]+1=[\beta]$ and
$\sum_{{\bf n}\not=\0}|{\bf n}|\gamma({\bf n})$ $=\sum_{{\bf n}\not=\0}|{\bf n}|\beta({\bf n})$.
Turning to (\ref{ao47}) we infer from (\ref{ao41}) that 
$(D^{({\bf n})})_\beta^\gamma\not=0$ if $\gamma=\beta+e_{\bf n}$ which implies $[\gamma]$
$=[\beta]-1$. For (\ref{ao46}) we look at (\ref{ao42}):
$(\partial_1)^\gamma_\beta$ $\not=0$ implies $\gamma+e_{k+1}+e_{(1,0)}$ $=\beta+e_k$
for some $k\ge0$ or $\gamma+e_{{\bf n}+(1,0)}$ $=\beta+e_{\bf n}$
for some ${\bf n}\not=\0$. In the first case and in the second case we have $[\gamma]$ $=[\beta]$
and $\sum_{{\bf n}\not=\0}|{\bf n}|\gamma({\bf n})+|(1,0)|$
$=\sum_{{\bf n}\not=\0}|{\bf n}|\beta({\bf n})$.
\end{proof}

\medskip

\subsection{The abstract model space $\T$}\label{Sect3.4}
\mbox{}

Now is a good moment to introduce the model space\footnote{It differs from a standard model space in regularity structures, see the discussion in Subsection \ref{Sect.ext}.} $\gls{modelspace}$ and its dual $\T^*$. 
We define $\T^*\subset\R[[\z_k,\z_\n]]$ to be the direct product over the multi-indices $\gamma$ with 
\begin{equation}\label{newreference}
[\gamma]\geq0
\quad
\textnormal{or}
\quad
\gamma\in\{e_\n\}_{\n\neq \0}.
\end{equation} 
This restriction of $\mathbb{R}[[\z_k,\z_{\bf n}]]$ to $\mathsf{T}^*$ is motivated by the fact that the model component $\Pi_\gamma$ is only non-vanishing when $[\gamma]$ $\ge 0$ or $\gamma$ $\in\{e_{\bf n}\}_{{\bf n}\not=\0}$. This can be read off from \eqref{extra1}; 
the same holds for \eqref{drivenODE} and \eqref{SHE}, see \eqref{ckh03} and \eqref{gpam22}, respectively\footnote{ Actually, such statements may be shown inductively under some uniqueness assumption for the corresponding equation, which guarantees that the only solution to the homogeneous problem is $0$. This more analytic remark is outside the scope of this paper; we refer to \cite{LO,LOTT} for a full argument.}. We denote by $\mathsf{\tilde T}^*$ the subspace of elements of the dual space $\T^*$ that vanish on the space $\mathsf{\bar T}$ introduced in \eqref{ao106}; in particular $\mathsf{\tilde T}^*$ is the direct product over the multi-indices satisfying $[\gamma]\geq 0$.
Then $\T^*=\mathsf{\bar T} ^* \oplus \mathsf{\tilde T}^*$, where $\mathsf{\bar T} ^*$ is the direct product over the multi-indices $\gamma\in\{e_\n\}_{\n\neq \0}$. 
Thus we can identify the model space $\T$ with the direct
sum of the polynomial sector $\mathsf{\bar T}$ introduced in \eqref{ao106} and the space $\gls{nonpolynomialsector}$
spanned by all monomials $\gls{monomialT}$ with $[\gamma]\ge 0$. 
Since $[\gamma]\ge 0$ 
is closed under addition of multi-indices, and thus under multiplication, $\mathsf{\tilde T}^*$ is a sub-algebra.
We note that the derivations (\ref{ao37}) map $\mathsf{\tilde T}^*$ into
$\mathsf{\tilde T}^*$:
\begin{align}\label{ao48}
D\mathsf{\tilde T}^*\subset\mathsf{\tilde T}^*
\quad\mbox{for}\quad
D\in \{ D^{({\bf n})}\}_{{\bf n}} \cup \{\partial_i\}_i,
\end{align}
which on the level of the coordinate representation means
\begin{align}\label{ao49}
D_\beta^\gamma=0\quad\mbox{for}\quad[\beta]<0\quad\mbox{and}\quad[\gamma]\ge 0,
\end{align}
and follows from (\ref{ao40}), (\ref{ao47}), and (\ref{ao46}).


\subsection{The infinitesimal generators of variable tilt $\{\z^\gamma D^{(\n)}\}_{\gamma,\n}$} \label{Sect3.8}
\mbox{}

After introducing the building blocks (\ref{ao37}), 
we now specify the full collection of derivations on 
$\mathbb{R}[[\z_k,\z_{\bf n}]]$ that will act as the basis of $\L$.
Again, we start with a motivation: The purpose of the structure group 
$\mathsf{G}\subset{\rm End}(\T)$, or
rather its pointwise dual $\mathsf{G}^*\subset{\rm End}(\mathsf{T}^*)$, 
is to provide the transformations of the $\mathsf{T}^*$-valued model $\Pi_x$ when passing
from one base-point $x$ to another, see \cite[Definition 3.3]{Hairer}. 
This re-centering involves
subtracting a Taylor polynomial. Denoting the coefficients of such a polynomial
by $\{\pi^{(\bf n)}\}_{\bf n}$, and treating the constant part (i.~e.~the part
with ${\bf n}=\0$) differently in line with
(\ref{ao28}) and (\ref{ao27bis}), this corresponds to the action (\ref{ao29bis}).

\medskip

In the inductive construction of the $\mathsf{\tilde T}^*$-valued centered model $\Pi_x$, 
the coefficients $\pi^{({\bf n})}$ depend on the $\mathsf{\tilde T}^*$-valued $\Pi_x$
itself, which for us means
$\pi^{({\bf n})}\in\mathsf{\tilde T}^*\subset\mathbb{R}[[\z_k,\z_{\bf n}]]$.
We pass from $\pi^{({\bf n})}\in\mathsf{\tilde T}^*$ to a finite linear combination\footnote{We
will free ourselves from this restriction later.}
of monomials $\z^{\gamma}$ with $[\gamma]\ge 0$.
Hence on an infinitesimal level, in view of the characterization (\ref{ao21}) of $D^{(\0)}$
and the definition (\ref{ao30}) of $D^{(\bf n)}$ for ${\bf n}\not=\0$,
transformations of the type (\ref{ao29bis}) give rise to the derivations
\begin{align}\label{ao32}
\gls{generatorVariableTilt} \quad\mbox{for}\quad[\gamma]\ge 0\;\;\mbox{and}\;\;
{\bf n}.
\end{align}
Since $\mathsf{\tilde T}^*$ is closed under multiplication, it follows that (\ref{ao48}) is preserved:
\begin{align}\label{ao61}
D\mathsf{\tilde T}^*\subset\mathsf{\tilde T}^*
\quad\mbox{for}\quad
D\in \{\z^{\gamma}D^{({\bf n})}\}_{[\gamma]\ge 0,{\bf n}}\cup \{\partial_i\}_{i}.
\end{align}
However, even for $[\gamma]\ge 0$, multiplication with $\z^{\gamma}$ does not map
$\mathsf{\bar T}^*$ into $\mathsf{T}^*$. Luckily, the composition
$\z^{\gamma}D^{({\bf n})}$ does; we have
\begin{align}\label{hk02}
D\mathsf{T}^*\subset\mathsf{T}^*
\quad\mbox{for}\quad
D\in\{\z^{\gamma}D^{({\bf n})}\}_{[\gamma]\ge 0,{\bf n}}\cup \{\partial_i\}_{i}.
\end{align}
Indeed, this is an immediate consequence of \eqref{ao35} and \eqref{ao30} together with \eqref{ao61}.

\medskip

On the level of the matrix representation,
\begin{align}\label{ao58}
(\z^{\gamma} D^{(\bf n)})_{\beta'}^{\gamma'}
=(D^{(\bf n)})_{\beta'-\gamma}^{\gamma'};
\end{align}
this implies that for all these operators, and not just for $\partial_1$ and $\partial_2$,
the finiteness property (\ref{ao23}) holds. 
Moreover, when passing from $D=D^{({\bf n})}$ to $D=\z^{\gamma} D^{(\bf n)}$,
(\ref{ao56}) is preserved so that (\ref{ao57}) can be upgraded to
\begin{align}\label{fw25}
D^\dagger\mathsf{\bar T}\subset\mathsf{\bar T}\quad\mbox{for}\quad
D\in\{\z^{\gamma}D^{({\bf n})}\}_{[\gamma]\ge 0,{\bf n}}\cup \{\partial_i\}_{i}.
\end{align}

\subsection{A pre-Lie structure $\prelie$ and bigrading ${\rm bi}$}\label{Sect3.10}
\mbox{}

The Lie algebra ${\rm Der}(\mathbb{R}[[\z_k,\z_\n]])$ of derivations on the
algebra $\mathbb{R}[[\z_k,\z_\n]]$ can be seen as the space of vector fields
on the linear span of $\{\z_k,\z_\n\}$. Since $\mathbb{R}[[\z_k,\z_\n]]$ as an affine space is flat, 
the Lie bracket $[\cdot,\cdot]$ arises from the
pre-Lie product $\gls{preLieProduct}$ that is
given by the covariant derivative of one vector field along another vector field,
see e.~g.~\cite{Manchon}; the relation between the bracket and the product is given 
by $[D,D']=D\prelie D'-D'\prelie D$.
In case of our derivations we find for arbitrary $D\in{\rm Der}(\mathbb{R}[[\z_k,\z_\n]])$
\begin{align}\label{fs01}
D\prelie \z^\gamma D^{(\n)}=(D\z^\gamma) D^{(\n)}\quad\mbox{and}\quad
\z^\gamma D^{(\n)}\prelie\partial_1= n_1 \z^\gamma D^{(\n-(1,0))}
\end{align}
and an analogous formula with $\partial_1$ replaced by $\partial_2$.
However, $\partial_1\prelie\partial_1$ cannot be expressed in terms
of a linear combination of $\{\partial_i\}_i\cup\{\z^\gamma D^{(\n)}\}_{\gamma,\n}$,
so that the span of the latter is not closed under $\prelie$. 
Note that it is not possible to fix this by postulating $\partial_1\prelie\partial_1=0$, 
since then the (left) pre-Lie identity is not satisfied.\footnote{
A simple counterexample is 
$(\z^\gamma D^{(\n)}\prelie\partial_1)\prelie\partial_1-\z^\gamma D^{(\n)}\prelie(\partial_1\prelie\partial_1) = n_1(n_1-1)\z^\gamma D^{(\n-(2,0))}$, whereas $(\partial_1\prelie\z^\gamma D^{(\n)})\prelie\partial_1-\partial_1\prelie(\z^\gamma D^{(\n)}\prelie\partial_1) =  (\partial_1\z^\gamma)D^{(\n)}\prelie\partial_1-\partial_1\prelie n_1\z^\gamma D^{(\n-(1,0))} = 0$.}
Nevertheless, it follows from (\ref{fs01}) and (\ref{ao67}) that the span of
$\{\partial_i\}_i\cup\{\z^\gamma D^{(\n)}\}_{\gamma,\n}$ is closed under $[\cdot,\cdot]$,
which will be used in Subsection \ref{Sect3.7}.

\medskip

The presence of a pre-Lie structure connects to the pre-Lie algebras in rough paths \cite{BCFP19} and regularity structures \cite{BCCH19}. Indeed, as we shall see in Section \ref{Sect5} in the specific case of driven ODEs, $\prelie$ is related to the grafting pre-Lie product (up to combinatorial factors, see Subsection \ref{Sect5.4} for a detailed discussion).

\medskip

We now come to an important observation:
There is a bigrading\footnote{This is just a compact way of saying that there exist two gradings, which we put together in a two-component vector, cf. \eqref{fs02}.} on the index set 
$\{1,2\}\cup\{(\gamma,\n) \, | \, [\gamma]\ge 0,\n\not=\0\}$ of our (linearly independent)
family of derivations that is 
compatible with the pre-Lie product $\prelie$. 
Indeed, we associate a pair of integers to every index by the following map $\gls{bigrading}$:
\begin{equation}\label{fs02}
\begin{split}
&{\rm bi}(\gamma,\n):=\big(1+[\gamma],\sum_{\m\not=\0}|\m|\gamma(\m)-|\n|\big),\\
&{\rm bi}\, 1 := (0,|(1,0)|),\ \  {\rm bi}\, 2 := (0,|(0,1)|).
\end{split}
\end{equation}
By compatibility we mean that for any two elements $D,D'$ of our family, provided
not both are of the form $\partial_i$, the product $D\prelie D'$ is a linear combination 
of elements of our family that only correspond to indices such that their bigrading 
is the sum of the bigradings of the index for $D$ and for $D'$.
This is obvious for the second item in (\ref{fs01}).
Expanding the first item in (\ref{fs01}) as 
\begin{align*}
\z^{\gamma'} D^{(\n')}\prelie \z^{\gamma} D^{(\n)}
&=\sum_{\beta}(\z^{\gamma'}D^{(\n')})_\beta^\gamma \z^\beta D^{(\n)},\\
\partial_1\prelie \z^{\gamma} D^{(\n)}
&=\sum_{\beta}(\partial_1)_\beta^\gamma \z^\beta D^{(\n)}
\end{align*}
and appealing to definition (\ref{fs02}) we see that our claim amounts to
\begin{align}\label{fs04}
\lefteqn{(\z^{\gamma'}D^{({\bf n}')})_\beta^\gamma\not=0\quad
\underset{\eqref{ao58}}{\Longrightarrow}
\quad(D^{({\bf n}')})_{\beta-\gamma'}^\gamma\not=0}\nonumber\\
&\Longrightarrow\quad
\left\{\begin{array}{ccl}
[\beta]&=&[\gamma]+[\gamma']+1,\\
{\displaystyle\sum_{{\bf m}\not=\0}|{\bf m}|\beta({\bf m})}&=&
{\displaystyle\sum_{{\bf m}\not=\0}|{\bf m}|\gamma({\bf m})
+\sum_{{\bf m}\not=\0}|{\bf m}|\gamma'({\bf m})}-|{\bf n}'|
\end{array}\right\}
\end{align}
and to (\ref{ao46}).
The second part of this implication also follows from Lemma \ref{Lem3.1}.

\medskip

Bigraded spaces appear in the context of regularity structures in \cite{BHZ}. In the tree-based setting, one chooses a bigrading \cite[(2.4)]{BHZ} which encodes the size of the tree, on the one hand, and the decorations, on the other. The same guiding principle is present in \eqref{fs02}: the quantity $1+[\gamma]$ is the number of edges of the trees represented by the multi-index $\gamma$, whereas the second component is, roughly speaking, counting the polynomial decorations. We refer to Sections \ref{Sect5} and \ref{Sect6} for more details.

\subsection{Homogeneities $\lhom\cdot\rhom\in\mathsf{A}$, and gradedness of $\T$}\label{Sect3.6}
\mbox{}

We now return to the strict triangular structure with respect to (\ref{ao50})
and in particular the deficiency of (\ref{ao47}). The choices we make now are guided by the application to the quasi-linear equation \eqref{as34} with a driver $\xi$ of regularity $\alpha - 2$. 
Inspired by \eqref{fs02} we choose an $\alpha>0$ and define the homogeneity of a multi-index $\gamma$ as
\begin{align}\label{ao52}
\gls{homogeneity} =\alpha([\gamma]+1)+\sum_{{\bf n}\not=\0}|{\bf n}|\gamma({\bf n}),
\end{align}
where the normalization $\lhom 0\rhom=\alpha$, which destroys additivity,
is made such that, in line with \cite[Assumption 3.20]{Hairer}, 
\begin{align}\label{ao59}
\lhom e_{\bf n}\rhom=|{\bf n}|\quad\mbox{for}\;{\bf n}\not=\0;
\end{align}
in particular, on the index set $\{e_{\bf n}\}_{{\bf n}\not=\0}$
of $\mathsf{\bar T}$ we have $\lhom e_\n \rhom\in\mathbb{N}$. On the index
set $\{\gamma \, | \, [\gamma]\ge 0\}$ of $\mathsf{\tilde T}$, we have that 
$\lhom\gamma\rhom\in\alpha\mathbb{N}+\mathbb{N}_0$. Hence 
$\gls{homogeneities} := \{\lhom\gamma\rhom \, | \, \z_\gamma\in\T \}$
satisfies the assumptions of \cite[Definition 3.1]{Hairer} of being bounded from below (namely by $\min\{\alpha,1\}$)
and locally finite.
Note that in our setting, the entire index set of $\T$ has positive homogeneity, and only corresponds to the ``integrated" part of the model; for a detailed connection to Hairer's set of homogeneities see Subsection \ref{Sect.ext}.

\medskip

Provided
the monomials that appear as multiplication operators in (\ref{ao32}) are constrained by
\begin{align}\label{ao60}
\z^{\gamma}D^{({\bf n})}\quad\mbox{for}\;[\gamma]\ge 0\;\mbox{and}\;
\lhom\gamma\rhom>|{\bf n}|,
\end{align}
as a consequence of combining \eqref{ao24}, \eqref{ao37} and \eqref{ao58}, the finiteness property \eqref{ao23} is uniform over
the entire collection $\{\z^{\gamma}D^{({\bf n})}\}_{[\gamma]\ge 0,{\bf n}}\cup \{\partial_i\}_{i}$:
\begin{align}\label{yourchoice1}
\{(\gamma',(\gamma,{\bf n}))\,|\,(\z^{\gamma}D^{({\bf n})})_{\beta'}^{\gamma'}\not=0\}
\;\;\mbox{is finite for all}\;\beta'.
\end{align}
This strengthening of \eqref{ao23} will be crucial when constructing $\Delta$.
Moreover, we have the following strict triangular structure:
\begin{lemma}
For $D\in\{\z^{\gamma}D^{({\bf n})}\}_{[\gamma]\ge 0,\lhom\gamma\rhom >|{\bf n}|}\cup \{\partial_i\}_{i}$ it holds
\begin{align}\label{ao55}
D_{\beta'}^{\gamma'}&\not=0\quad\implies\quad\lhom\gamma'\rhom<\lhom\beta'\rhom.
\end{align}
\end{lemma}
\begin{proof}
For $D=\partial_1$ (and $\partial_2$), this follows immediately from (\ref{ao46}) by definition
(\ref{ao52}). 
We turn to $D=\z^{\gamma} D^{(\0)}$. From (\ref{ao40}) and (\ref{ao58}) we read off that
$(\z^{\gamma}D^{(\0)})_{\beta'}^{\gamma'}$ 
$\not=0$ implies $[\gamma'+\gamma]$ $<[\beta']$ and
$\sum_{{\bf n}\not=\0}|{\bf n}|(\gamma'+\gamma)({\bf n})$ 
$\le\sum_{{\bf n}\not=\0}|{\bf n}|\beta'({\bf n})$. The latter, due to $\alpha>0$, implies
$\alpha([\gamma'+\gamma]+1)+\sum_{{\bf n}\not=\0}|{\bf n}|(\gamma'+\gamma)({\bf n})$
$<\alpha([\beta']+1)+\sum_{{\bf n}\not=\0}|{\bf n}|\beta'({\bf n})$,
which because of $[\gamma]\ge 0$ in turn yields the desired
$\alpha([\gamma']+1)+\sum_{{\bf n}\not=\0}|{\bf n}|\gamma'({\bf n})$
$<\alpha([\beta']+1)+\sum_{{\bf n}\not=\0}|{\bf n}|\beta'({\bf n})$. By definition
(\ref{ao52}), this establishes
(\ref{ao55}) for $D=\z^{\gamma}D^{(\0)}$.
We now turn to $D=\z^{\gamma}D^{({\bf n})}$ with ${\bf n}\not=\0$ and note
that by \eqref{ao41} and \eqref{ao58} we have $D^{\gamma'}_{\beta'}\not=0$ only for
$\gamma'+\gamma$ $=\beta'+e_{{\bf n}}$, which by 
(\ref{ao52}) implies $\lhom\gamma'\rhom$ $+\lhom\gamma\rhom$
$=\lhom\beta'\rhom$ $+\lhom e_{{\bf n}}\rhom$. By (\ref{ao59}) and the condition
in (\ref{ao60}), this yields as desired $\lhom\gamma'\rhom$
$<\lhom\beta'\rhom$.
\end{proof}

\medskip

Property (\ref{ao55}) results in the following gradedness: For $\kappa\in\mathsf{A}$ let
$\mathsf{T}_\kappa\subset\mathsf{T}$ denote the subspace corresponding to the indices $\gamma$ with $\lhom\gamma\rhom=\kappa$; we obviously have 
\begin{equation*}
\mathsf{T}=\bigoplus_{\kappa\in\mathsf{A}}\mathsf{T}_\kappa,
\end{equation*} 
in line with \cite[Definition 3.1]{Hairer}. Then (\ref{ao55}) can be reformulated as 
\begin{align}\label{ao68}
D^\dagger\mathsf{T}_\kappa&\subset\bigoplus_{\kappa'<\kappa}\mathsf{T}_{\kappa'}
\nonumber\\
&\mbox{for}\;D\in\{\z^{\gamma}D^{({\bf n})}\}_{[\gamma]\ge 0,\lhom\gamma\rhom >|{\bf n}|}\cup \{\partial_i\}_{i},
\end{align}
with the implicit understanding that $\kappa,\kappa'\in\mathsf{A}$.
We note that because of the presence of the $\z_0$-variable,
and thus the $\gamma(k=0)$-component on which $\lhom\gamma\rhom$ is not coercive,
$\mathsf{T}_\kappa$ is not finite dimensional. 
However, in the practice of \eqref{as34}, 
this is of no concern since the model $\Pi(x) \in\mathbb{R}[[\z_k,\z_{\bf n}]]$,
which a priori is a formal power series, 
actually is analytic in $\z_0$, 
which plays the role of a constant coefficient in $\partial_2 - (1+\z_0)\partial_1^2$.

\subsection{The Lie algebra $\mathsf{L}$}\label{Sect3.7}
\mbox{}

\begin{lemma}
The span of
\begin{equation}\label{ao74}
\{\z^{\gamma}D^{({\bf n})}\}_{[\gamma]\ge 0,\lhom\gamma\rhom >|{\bf n}|}\cup \{\partial_i\}_{i} ,
\end{equation}
as derivations
on $\mathbb{R}[[\z_k,\z_{\bf n}]]$, defines a bigraded Lie algebra $\gls{LieAlgebra}$.
\end{lemma}
\begin{proof}
We need to show that this sub-space of ${\rm Der}(\mathbb{R}[[\z_k,\z_{\bf n}]])$
is closed under taking the commutator $[D,D']$ $=D \prelie D'-D' \prelie D$, cf. Subsection \ref{Sect3.10}. To this purpose,
for any $D,D'$ $\in \{\z^{\gamma}D^{({\bf n})}\}_{[\gamma]\ge 0,\lhom\gamma\rhom >|{\bf n}|}\cup \{\partial_i\}_{i}$, we have to identify
$[D,D']$ as a linear combination of elements of this set. 

\medskip

We first note that by \eqref{ao67} we have $[\partial_1,\partial_2]=0$.
By \eqref{fs01}, written in its component-wise form, we obtain
\begin{align}\label{ao62}
\lefteqn{[\z^\gamma D^{({\bf n})},\z^{\gamma'}D^{({\bf n}')}]}\nonumber\\
&=\sum_{\beta'}(\z^\gamma D^{({\bf n})})_{\beta'}^{\gamma'}
\,\z^{\beta'} D^{({\bf n}')}
-\sum_{\beta}(\z^{\gamma'} D^{({\bf n}')})_{\beta}^{\gamma}
\,\z^{\beta}D^{({\bf n})},
\end{align}
where both sums are finite due to \eqref{ao64}.
By (\ref{ao49}), we learn that because of
$[\gamma],[\gamma']\ge 0$, the sums restrict to $[\beta],[\beta']\ge 0$.
Due to \eqref{ao55} they restrict to $\lhom\beta\rhom>\lhom\gamma\rhom$ and $\lhom\beta'\rhom>\lhom\gamma'\rhom$. 
Moreover, by assumption \eqref{ao60} we have
$\lhom\gamma\rhom>|{\bf n}|$ and $\lhom\gamma'\rhom>|{\bf n}'|$; hence as desired, the sums in \eqref{ao62} involve only multi-indices with $\lhom\beta'\rhom>|{\bf n}'|$ and 
$\lhom\beta\rhom>|{\bf n}|$.

\medskip

Finally, again by \eqref{fs01}, we have
\begin{align}\label{ao66}
[\z^\gamma D^{({\bf n})},\partial_1]=n_1\z^{\gamma}D^{({\bf n}-(1,0))}
-\sum_{\beta}(\partial_1)_{\beta}^\gamma\,\z^\beta D^{({\bf n})},
\end{align}
where (\ref{ao64}) again ensures the effective finiteness of the sum.
We note that (\ref{ao66}) has the desired form: The first r.~h.~s. term, which only is
present for $n_1\ge 1$, is admissible since obviously  
$\lhom\gamma\rhom$ $>|{\bf n}|$ $>|{\bf n}-(1,0)|$.
For the second r.~h.~s. term we note that 
by (\ref{ao49}) the sum is limited to $[\beta]\ge 0$, and by (\ref{ao55}) it is limited
to $\lhom\beta\rhom$ $>\lhom\gamma\rhom$ $>|{\bf n}|$.

\medskip

It remains to show that $\mathsf{L}$ is bigraded; this is clear from \eqref{ao46} and \eqref{fs04}, which show that the pre-Lie product (and thus the Lie bracket) is compatible with \eqref{fs02}, together with the commutation relation $[\partial_1,\partial_2]=0$.
\end{proof}


\section{The Hopf algebra structure}\label{Sect4}

\subsection{The universal enveloping algebra ${\rm U}(\mathsf{L})$, $\T^*$ as a module over ${\rm U}(\mathsf{L})$}\label{Sect4.1}
\mbox{}

We now adopt a more abstract point of view and consider the elements of the Lie algebra $\L$ as 
mere symbols rather than endomorphisms, and we interpret \eqref{ao67}, \eqref{ao62} and \eqref{ao66} as a coordinate representation of the Lie bracket in terms of the basis \eqref{ao74}. %
We denote by $\gls{UniversalEnvelope}$ the corresponding universal enveloping
algebra \cite[p. 28]{Abe}, an algebra which is based on the tensor algebra formed by $\L$
and quotiented through the ideal generated by the relations defining the Lie bracket.
We may think of the tensor algebra as the direct sum indexed by words. 

\medskip

Due to the mapping properties \eqref{hk02}, the canonical Lie algebra morphism $\gls{action}: \L \to \End(\T^*)$, which replaces every abstract symbol $D\in\L$ with its corresponding endomorphism, is well defined, and as a consequence of Subsection \ref{Sect3.7}, $\rho$ is a Lie algebra morphism.
By the universality property \cite[(U), p.29]{Abe}, such $\rho$ extends in a unique way to an algebra morphism $\rho : {\rm U}(\L) \to \End (\T^*)$; 
in particular, concatenation of words turns into composition of endomorphisms. 
However, this representation is not faithful\footnote{
i.~e.~one-to-one: consider $\z^{2e_1}D^{(1,0)}\z^{e_2}D^{(1,0)}$ and $\z^{e_1}D^{(1,0)}\z^{e_1 + e_2}D^{(1,0)}$, which are different words in ${\rm U}(\mathsf{L})$, but the same as endomorphisms.}. 
In a canonical way, we may rewrite $\rho$ as a map ${\rm U}(\L)\otimes \T^* \to \T^*$, so that $\T^*$ as a linear space becomes a left module over ${\rm U}(\L)$.

\medskip

The universal enveloping algebra ${\rm U}(\L)$ is naturally a Hopf algebra, cf. \cite[Examples 2.5, 2.8]{Abe}; the product is given by the concatenation of words, whereas the coproduct is characterized by its action on the elements $D\in\dum{L}$ (which we call primitive elements), namely
\begin{equation}\label{cop0203}
\gls{cop} D = 1\otimes D + D \otimes 1,
\end{equation}
and in general by the compatibility with the product, meaning that for all $U,U'\in {\rm U}(\L)$
\begin{equation}\label{prodcop}
\cop U U'=(\cop U) \, (\cop U').
\end{equation}

\subsection{The derived algebra $\tilde{\mathsf{L}}$ and the pre-Lie structure $\prelie$ revisited}\label{SectAd}
\mbox{}

As mentioned in Subsection \ref{Sect3.10}, the Lie algebra $\mathsf{L}$ is not closed
under the pre-Lie product $\prelie$. However, the only failure, namely
$\partial_i\prelie\partial_{i'}\not\in\mathsf{L}$, turns out to be peripheral.
This follows from the fact that 
$[D,D']$ does not have a $\partial_i$-component, see (\ref{ao62}) and (\ref{ao66}).
In other words we have for the derived algebra $[\mathsf{L},\mathsf{L}]\subset\tilde{\mathsf{L}}$, 
where the Lie sub-algebra $\tilde{\mathsf{L}}\subset\mathsf{L}$ is defined as
\begin{equation}\label{p26}
\gls{derivedAlgebra}:={\rm span}\{\z^\gamma D^{(\n)}\}_{[\gamma]\ge 0,\lhom\gamma\rhom>|\n|}.
\end{equation}
Since $\tilde{\mathsf{L}}$ is also an ideal, the quotient Lie algebra $\mathsf{L}/\tilde{\mathsf{L}}$
is Abelian, see \cite[Lemma 1.2.5]{HGK10}, and thus is isomorphic to 
$\{\partial_1,\partial_2\}$.
Moreover, the
Lie algebra morphism $\mathsf{L}\rightarrow \mathsf{L}/\tilde{\mathsf{L}}\cong\{\partial_1,\partial_2\}$ 
induces
an algebra morphism ${\rm U}(\mathsf{L})\rightarrow {\rm U}(\mathsf{L}/\tilde{\mathsf{L}})$
$\cong\{\partial^{\m}\}_{\m\in\mathbb{N}_0^2}$,
see \cite[p. 29]{Abe}. This algebra morphism in turn induces the decomposition
\begin{align}\label{dec}
{\rm U}(\mathsf{L})=\bigoplus_{\m\in\mathbb{N}_0^2} {\rm U}_{\m}.
\end{align}
By definition, ${\rm U}_{\0}$ is canonically isomorphic to ${\rm U}(\tilde{\mathsf{L}})$.
Since $\tilde{\mathsf{L}}$ is closed under $\prelie$, the pre-Lie structure
provides a canonical isomorphism, as cocommutative coalgebras, 
between ${\rm U}(\tilde{\mathsf{L}})$ and the
symmetric tensor algebra ${\rm S}(\tilde{\mathsf{L}})$, see \cite[Theorem 2.12]{GuinOudom}.
Via the definition (\ref{fs07}), the pre-Lie structure 
$\prelie\colon\mathsf{L}\times\tilde{\mathsf{L}}\rightarrow\mathsf{L}$
provides a natural isomorphism between the linear spaces ${\rm U}_{\m}$ 
and ${\rm U}(\tilde{\mathsf{L}})$, as will become apparent in Subsection \ref{Sect4.2}.
These natural isomorphisms, of which we will make no explicit use, will guide our construction of a basis in Subsection \ref{Sect4.2}.

\medskip

We now will be more precise on how we salvage the pre-Lie structure $\prelie$.
We use $\prelie$ in terms of the product
\begin{align}\label{fs07}
\mathsf{L}\times\tilde{\mathsf{L}}\ni\;(D,\tilde{D})\mapsto D\tilde{D}-D\prelie \tilde{D}\;\in{\rm U}(\mathsf{L}).
\end{align}
Fixing the second factor $\tilde{D}=\z^\gamma D^{(\n)}$, we extend this product from $D\in\mathsf{L}$
to $U\in{\rm U}(\mathsf{L})$ 
\begin{align}\label{yc11}
{\rm U}(\mathsf{L})\ni\; U\mapsto \z^\gamma U D^{(\n)}\;\in{\rm U}(\mathsf{L}),
\end{align}
The map (\ref{yc11}) is inductively defined in the length of $U$ by anchoring
through $\z^{\gamma} 1 D^{({\bf n})}$ 
$=\z^{\gamma}D^{({\bf n})}$ and postulating
for any $D\in\mathsf{L}\subset{\rm U}(\mathsf{L})$
\begin{align}\label{yc12}
\z^{\gamma} D U D^{({\bf n})}
=D \z^{\gamma} U D^{({\bf n})}
-\sum_\beta D_\beta^{\gamma}\z^\beta U D^{({\bf n})}.
\end{align}
Let us comment on (\ref{yc12}): First of all, the identity (\ref{yc12}) is consistent with
the map $\rho\colon{\rm U}(\mathsf{L})\rightarrow{\rm End}(\mathsf{T}^*)$
in the sense of
\begin{align}\label{fw24}
\rho \mathsf{z}^\gamma U D^{({\bf n})}=\mathsf{z}^\gamma (\rho U) D^{({\bf n})},
\end{align}
since $D$ as an element of $\End(\T^*)$ is a derivation.
As an identity in ${\rm U}(\mathsf{L})$, it is to be read as follows:
On the l.~h.~s., we first multiply $U$ by $D$ via concatenation,
and then apply (\ref{yc11}). For the first r.~h.~s. term, we reverse this order.
The second r.~h.~s. term is a linear combination of several versions of (\ref{yc11}) 
(with $\gamma$ replaced by $\beta$); the coefficients are given by identifying $D\in\mathsf{L}$
with $D\in{\rm End}(\mathsf{T}^*)$, and \eqref{ao55} shows that $(\beta,{\bf n})$ is in the index set of $\mathsf{L}$.
Hence (\ref{yc12}) indeed provides an inductive definition of (\ref{yc11}).

\medskip

A first crucial observation is that the maps (\ref{yc11}) commute\footnote{which is at the basis of the canonical identification of ${\rm S}(\mathsf{L})$ with ${\rm U}(\mathsf{L})$ in \cite[Theorem 3.14]{GuinOudom}}:
\begin{lemma}\label{Lem4.1} It holds 
\begin{align}\label{com}
\z^{\gamma'}\z^{\gamma} U D^{({\bf n})}D^{({\bf n}')}
=\z^{\gamma}\z^{\gamma'} U D^{({\bf n}')}D^{({\bf n})}.
\end{align}
\end{lemma}
\begin{proof}
We argue by induction. 
The base case of $U=1$ follows from using \eqref{yc12} twice (once for $D=\z^{\gamma'} D^{(\n')}$ and $U=1$) and connecting the outcomes via \eqref{ao62}. 
We now assume that \eqref{com} is satisfied for some $U$ and give ourselves an element $D\in \mathsf{L}$. Applying \eqref{yc12} twice, we obtain
\begin{align*}
\z^{\gamma'}\z^{\gamma}D U D^{(\n)}&D^{(\n')} = D\z^{\gamma'}\z^{\gamma} U D^{(\n)}D^{(\n')} \\
&\quad- \sum_{\beta'} D_{\beta'}^{\gamma'} \z^{\beta'}\z^{\gamma}U D^{(\n)}D^{(\n')}- \sum_{\beta} D_{\beta}^{\gamma}\z^{\gamma'}\z^{\beta}U D^{(\n)}D^{(\n')},
\end{align*}
and the analogous expression in case of $\z^{\gamma}\z^{\gamma'}D U D^{(\n')}D^{(\n)}$; by the induction hypothesis, both are equal.
\end{proof}

\medskip

A second crucial observation is that the maps (\ref{yc11}) commute with 
the coproduct on ${\rm U}(\mathsf{L})$ in the following sense: 

\begin{lemma}\label{lem4.2}
If \footnote{ Here and in the sequel we use Sweedler's notation, see e.~g.~\cite[p. 56]{Abe}.}
\begin{align}\label{copU01}
{\rm cop}\,U=\sum_{(U)}U_{(1)}\otimes U_{(2)}
\end{align}
then
\begin{align}\label{yc14}
{\rm cop}\,\z^\gamma U D^{({\bf n})}=\sum_{(U)}\big(\z^\gamma U_{(1)}D^{({\bf n})}
\otimes U_{(2)}+U_{(1)}\otimes \z^\gamma U_{(2)}D^{({\bf n})}\big).
\end{align}
\end{lemma}
\begin{proof}
Once more we argue by induction. The base case of $U=1$ is included in \eqref{cop0203} and our definition $\z^\gamma 1 D^{(\n)} = \z^\gamma D^{(\n)}$. 
We now assume that \eqref{yc14} is satisfied for some $U$ and give ourselves an element $D\in\mathsf{L}$. 
Using the inductive definition (\ref{yc12})
\begin{align*}
\lefteqn{\cop\mathsf{z}^\gamma DUD^{({\bf n})}
=\cop D\mathsf{z}^\gamma UD^{({\bf n})}-\sum_{\beta}D_\beta^\gamma \,
\cop \mathsf{z}^\beta UD^{({\bf n})}}\\
&\underset{\eqref{cop0203},\eqref{prodcop}}{=}
(1\otimes D+D\otimes 1)\cop \mathsf{z}^\gamma UD^{({\bf n})}
-\sum_{\beta}D_\beta^\gamma \, \cop \mathsf{z}^\beta UD^{({\bf n})},
\end{align*}
we may feed in the induction hypothesis, leading to
\begin{align*}
\lefteqn{\cop \mathsf{z}^\gamma DUD^{({\bf n})}}\nonumber\\
&=\sum_{(U)}\Big(\mathsf{z}^\gamma U_{(1)}D^{({\bf n})}\otimes DU_{(2)}
+U_{(1)}\otimes D\mathsf{z}^\gamma U_{(2)}D^{({\bf n})}\nonumber\\
&\mbox{}\qquad+D \mathsf{z}^\gamma U_{(1)}D^{({\bf n})}\otimes U_{(2)}
+D U_{(1)}\otimes \mathsf{z}^\gamma U_{(2)}D^{({\bf n})}\nonumber\\
&\mbox{}\qquad-\sum_{\beta}D_\beta^\gamma\big(\mathsf{z}^\beta U_{(1)}D^{({\bf n})}\otimes U_{(2)}
+U_{(1)}\otimes \mathsf{z}^\beta U_{(2)}D^{({\bf n})}\big)\Big),
\end{align*}
which by the inductive definition (\ref{yc12}) compactifies to

\begin{align*}
\cop \mathsf{z}^\gamma DUD^{({\bf n})}
&=\sum_{(U)}\Big(\mathsf{z}^\gamma U_{(1)}D^{({\bf n})}\otimes DU_{(2)}
+U_{(1)}\otimes \mathsf{z}^\gamma D U_{(2)}D^{({\bf n})}\nonumber\\
&+\mathsf{z}^\gamma D U_{(1)}D^{({\bf n})}\otimes U_{(2)}
+D U_{(1)}\otimes \mathsf{z}^\gamma U_{(2)}D^{({\bf n})}\Big).
\end{align*}
Since by \eqref{cop0203} and \eqref{prodcop}
\begin{equation}\label{copU02}
\cop DU = \sum_{(U)} (DU_{(1)}\otimes U_{(2)} + U_{(1)}\otimes D U_{(2)}),
\end{equation}
the proof is complete.
\end{proof}

\medskip

A third crucial observation is that the maps \eqref{yc11} connect product and coproduct in the following sense: 
\begin{lemma}
	Under the assumption \eqref{copU01},
\begin{equation}\label{yc15}
U \z^\gamma D^{(\n)}
=
\sum_{(U),\beta} (U_{(1)})_\beta^\gamma \, \z^\beta U_{(2)} D^{(\n)}.
\end{equation}
\end{lemma}
\begin{proof}
Again we argue by induction. For $U=1$ the identity follows by noting that $\cop 1  = 1\otimes 1$. Assume now that \eqref{yc15} is satisfied for some $U$, then for $D\in \L$ by the induction hypothesis
\begin{align*}
D U \z^\gamma D^{(\n)} =
D \sum_{(U),\beta} (U_{(1)})_\beta^\gamma \, \z^\beta U_{(2)} D^{(\n)} ,
\end{align*}
which by \eqref{yc12} leads to
\begin{align*}
D U \z^\gamma D^{(\n)}=
\sum_{(U),\beta} (U_{(1)})_\beta^\gamma \Big( 
\z^\beta D U_{(2)} D^{(\n)} + \sum_{\beta'} D_{\beta'}^\beta \z^{\beta'} U_{(2)} D^{(\n)} \Big) .
\end{align*}
Using $\sum_{\beta} D_{\beta'}^\beta \, (U_{(1)})_\beta^\gamma = (D U_{(1)})_{\beta'}^\gamma$ together with \eqref{copU02} finishes the proof.
\end{proof}

\medskip

A final observation is an intertwining of the maps \eqref{yc11} with $\partial_1$:
\begin{lemma}
It holds
\begin{equation}\label{yc16}
\z^\gamma U D^{(\n)}\partial_1 = \z^\gamma U\partial_1 D^{(\n)} + n_1 \z^\gamma U D^{(\n-(1,0))},
\end{equation}
and an analogous statement holds for $\partial_1$ replaced by $\partial_2$.
\end{lemma}
\begin{proof}
We argue by induction. If $U=1$, 
\begin{align*}
\z^\gamma  D^{(\n)} \partial_1 &\underset{\eqref{ao66}}{=}\partial_1 \z^\gamma D^{(\n )} + n_1 \z^\gamma D^{(\n-(1,0))} - \sum_\beta (\partial_1)_\beta^\gamma \z^\beta D^{(\n)}\\
&\underset{\eqref{yc12}}{=} \z^\gamma \partial_1 D^{(\n)} + n_1 \z^\gamma D^{(\n-(1,0))}.\\
\end{align*}
We now assume that \eqref{yc16} is true for a given $U$ and aim to prove it for $DU$, where $D\in\mathsf{L}$. Then by \eqref{yc12}
\begin{align*}
\z^\gamma DU D^{(\n)}\partial_1 = D\z^\gamma U D^{(\n)}\partial_1 - \sum_\beta D_\beta^\gamma \z^\beta U D^{(\n)}\partial_1,
\end{align*}
and appealing to the induction hypothesis yields
\begin{align*}
\z^\gamma DU D^{(\n)}\partial_1 = D\z^\gamma U \partial_1 D^{(\n)} + n_1 D\z^\gamma U D^{(\n-(1,0))} \\
- \sum_\beta D_\beta^\gamma \z^\beta U \partial_1 D^{(\n)} - n_1\sum_\beta D_\beta^\gamma \z^\beta U D^{(\n-(1,0))},
\end{align*}
which by a second application of \eqref{yc12} takes the desired form.
\end{proof}

\medskip

\subsection{The choice of basis $\{\D_{(J,\m)}\}_{(J,\m)}$}\label{Sect4.2}
\mbox{}

We now define a basis in ${\rm U}(\mathsf{L})$ based on the structure derived from the pre-Lie structure in the previous Subsection \ref{SectAd}.
Applying iteratively the map \eqref{yc11} to an element $\partial^\m$, and making use of Lemma \ref{Lem4.1} we may define
\begin{align}\label{ao75}
\gls{basisU}:=
\frac{1}{J!{\bf m}!}
\prod_{(\gamma,{\bf n})} (\z^{\gamma})^{J(\gamma,{\bf n})} \,
\partial_1^{m_1}
\partial_2^{m_2}
\prod_{(\gamma,{\bf n})} (D^{({\bf n})})^{J(\gamma,{\bf n})}\in {\rm U}(\mathsf{L}).
\end{align}
Here $J$ denotes a multi-index on tuples $(\gamma,\n)$ with $[\gamma]\geq 0$ and $\lhom\gamma\rhom > |\n|$; we set $J!:= \prod_{(\gamma,\n)} J(\gamma,\n)!$, so that the normalization constant $J!\m!$ may be seen as the multi-index factorial $(J,\m)!$.
This normalization is chosen such that the basis representation
of the coproduct is standard, see \eqref{cop01} below. 
The collateral damage of this normalization is that the
basis representation of (\ref{yc11}) acquires a combinatorial factor:
\begin{align}\label{combfac}
\z^\gamma \D_{(J,\m)}D^{(\n)}=(J(\gamma,\n)+1)\D_{(J+e_{(\gamma,\n)},\m)}.
\end{align}
For $J\equiv0$,
(\ref{ao75}) reduces to the standard basis $\{\frac{1}{\m!}\partial^\m\}_{\m\in\mathbb{N}_0^2}$
for the coalgebra of differential operators, characterized as dual to the standard basis
$\{x^\m\}_{\m\in\mathbb{N}_0^2}$ of the algebra $\mathbb{R}[x_1,x_2]$ under the pairing of
\cite[Example 2.2]{HairerKelly}.
The concatenation of these elements leads as well to a combinatorial factor:
\begin{equation}\label{proddel}
D_{(0,\m')}D_{(0,\m'')} = \tbinom{\m'+\m''}{\m'}D_{(0,\m'+\m'')}.
\end{equation}
For later purpose, let us define the length of $(J,\m)$ by \[\gls{lengthJm} := \sum_{(\gamma,\n)} J(\gamma,\n) + m_1 + m_2.\]
\begin{lemma}
The set $\{\D_{(J,{\bf m})}\}_{(J,{\bf m})}$ is a basis of ${\rm U}(\L)$.
\end{lemma}
\begin{proof}
 As a consequence of the Poincaré-Birkhoff-Witt Theorem, cf. \cite[Theorem 1.9.6]{HGK10}, after a choice of an order $\prec$ on the set of pairs $(\gamma,\n)$, the set of elements of the form
\begin{align}\label{bas01}
\B_{(J,\m)}&:=\frac{1}{J!\m!}\partial_1^{m_1}\partial_2^{m_2} \z^{\gamma_1}D^{(\n_1)}\cdots \z^{\gamma_k}D^{(\n_k)},\\
&\quad\textnormal{where } J=e_{(\gamma_1,\n_1)}+...+e_{(\gamma_k,\n_k)} \nonumber\\
&\quad\textnormal{and }(\gamma_1,\n_1)\preceq...\preceq (\gamma_k,\n_k)\nonumber
\end{align}
is a basis of ${\rm U}(\mathsf{L})$. Applying \eqref{yc12} iteratively, one can show the representation 
\begin{equation*}
\D_{(J,\m)} = \B_{(J,\m)} + \mbox{span} \{B_{(J',\m')} \, \big| \, |(J',\m')|<|(J,\m)| \}.
\end{equation*}
Since $\{B_{(J,\m)}\}_{(J,\m)}$ is a basis, it is easy to deduce from this identity that also $\{D_{(J,\m)}\}_{(J,\m)}$ is a basis.
\end{proof}

\medskip

The advantage of the basis \eqref{ao75} over a Poincaré-Birkhoff-Witt basis of the form \eqref{bas01} is that the former does not rely on the choice of an order in $\mathsf{L}$, cf. \eqref{com}, whereas the latter crucially does. 
The only choice to be made is the order of the three symbols: having first the $\z^\gamma$'s, then the $\partial$'s and last the $D^{(\n)}$'s generates the only basis for which the analogue of \cite[(4.14)]{Hairer}, namely \eqref{Hai4.14}, is true.
In addition, with the basis \eqref{ao75} we obtain the most direct identification of our group elements as exponentials of shift and tilt parameters, cf. Proposition \ref{exp01}.

\medskip

The coproduct has the following simple structure in the basis \eqref{ao75}, which is reminiscent of the Hopf algebra of constant-coefficient differential operators over the algebra of smooth functions, cf. \cite{Brouder}.
\begin{lemma}\label{lemcop}
	It holds
	\begin{equation}\label{cop01}
	{\rm cop}\,\D_{(J,{\bf m})}=\sum_{(J',{\bf m'})+(J'',{\bf m''})
		=( J,{\bf m})} \hspace{-6ex} \D_{(J',{\bf m'})}\otimes\D_{(J'',{\bf m''})}.
	\end{equation}
\end{lemma}
\begin{proof}
We proceed by induction in the length $|(J,\m)|$. 
The base case $|(J,\m)|$ $=0$ reduces to the trivial $\cop 1 = 1\otimes 1$. 
For the induction step, we give ourselves an element $D_{(J,\m)}$ and distinguish two cases. 
If $J=0$, we assume without loss of generality $m_1\neq 0$ and with help of \eqref{proddel} we write $m_1 D_{(0,\m)} = D_{(0,\m-(1,0))}D_{(0,(1,0))}$ in order to access the induction hypothesis. 
Then by \eqref{cop0203}, \eqref{prodcop} and the induction hypothesis,
\begin{equation*}
\begin{split}
&\cop D_{(0,\m-(1,0))}D_{(0,(1,0))} \\
&\quad \,\,= \sum_{\m' + \m'' = \m - (1,0)} \hspace{-5ex} D_{(0,\m')}D_{(0,(1,0))}\otimes D_{(0,\m'')} + D_{(0,\m')}\otimes D_{(0,\m'')}D_{(0,(1,0))}\\
&\quad\underset{\eqref{proddel}}{=}  \sum_{\m' + \m'' = \m - (1,0)} \hspace{-5ex} m_1' D_{(0,\m'+(1,0))}\otimes D_{(0,\m'')} + m_1'' D_{(0,\m')}\otimes D_{(0,\m''+(1,0))}.
\end{split}
\end{equation*}
By \eqref{ao75} and Lemma \ref{lemsum01}, this equals $m_1 \sum_{\m' + \m'' = \m} D_{(0,\m')}\otimes D_{(0,\m'')}$, which proves \eqref{cop01} in the case $J=0$. 

\medskip

We now address the case $J\neq 0$. We take a pair $(\gamma,\n)$ such that $J(\gamma,\n)\neq 0$ and use \eqref{combfac} to write \[(J(\gamma,\n) + 1)D_{(J,\m)} = \z^\gamma D_{(J-e_{(\gamma,\n)},\m)}D^{(\n)}.\] Then by \eqref{yc14} and the induction hypothesis,
\begin{equation*}
\begin{split}
&\cop \z^\gamma D_{(J-e_{(\gamma,\n)},\m)}D^{(\n)} \\
&\quad\,\, =\hspace{-5ex} \sum_{(J',\m') + (J'',\m'') = (J-e_{(\gamma,\n)},\m)}\hspace{-8.5ex} \z^\gamma D_{(J',\m')}D^{(\n)}\otimes D_{(J'',\m'')} + D_{(J',\m')}\otimes \z^\gamma D_{(J'',\m'')}D^{(\n)}\\
&\quad\underset{\eqref{combfac}}{=} \hspace{-4ex} \sum_{(J',\m') + (J'',\m'') = (J-e_{(\gamma,\n)},\m)}\hspace{-8.5ex} (J'(\gamma,\n) + 1) D_{(J'+e_{(\gamma,\n)},\m')}\otimes D_{(J'',\m'')} \\
&\quad\quad\quad\quad\quad\quad\quad+ (J''(\gamma,\n) + 1) D_{(J',\m')}\otimes D_{(J''+e_{(\gamma,\n)},\m'')}.
\end{split}
\end{equation*}
By \eqref{ao75} and Lemma \ref{lemsum01}, this as desired reduces to 
\begin{equation*}
(J(\gamma,\n)+1)\hspace{-3ex} \sum_{(J',\m')+(J'',\m'') = (J,\m)}\hspace{-6ex} D_{(J'\m')}\otimes D_{(J'',\m'')}.\qedhere
\end{equation*}
\end{proof}

\medskip

Recall that by construction, the concatenation product on ${\rm U}(\mathsf{L})$
is an abstract lifting of the composition product on ${\rm End}(\mathsf{T}^*)$,
as can be seen by applying $\rho$.
The upcoming lemma is a projection of
this fact onto $\mathsf{\tilde L}$, see (\ref{p26}). More precisely, we resolve
$\mathsf{\tilde L}$ in terms of ${\bf n}$ by introducing the family of maps
$\gls{iotan}\colon{\rm U}(\mathsf{L})\rightarrow\mathsf{T}^*$,
determined by how it acts on elements of the basis (\ref{ao75}):
\begin{align}\label{io01}
\iota_{\bf n}D_{(J,{\bf m})}=\left\{\begin{array}{cl}
\mathsf{z}^\beta&\mbox{if}\;(J,{\bf m})=(e_{(\beta,{\bf n})},\0)\\
0&\mbox{otherwise}
\end{array}\right\} .
\end{align}

\medskip

As it turns out, the next lemma involves generalized counits\footnote{
$\varepsilon_\0$ is the plain counit of ${\rm U}(\mathsf{L})$}
$\gls{counitn} \colon
{\rm U}(\mathsf{L})\rightarrow\mathbb{R}$ defined on the basis through
\begin{align}\label{seemynotesonp38}
\varepsilon_{\bf n} D_{(J,{\bf m})}=\left\{\begin{array}{cl}
1&\mbox{if}\;(J,{\bf m})=(0,{\bf n})\\
0&\mbox{otherwise}
\end{array}\right\} .
\end{align}

\begin{lemma}
It holds as an identity in $\T^*$
	\begin{equation}\label{io02}
	\iota_\n U_1 U_2 = (\rho U_1)\, \iota_\n U_2 + \sum_\m \tbinom{\n+\m}{\m} (\varepsilon_\m  U_2)\iota_{\n+\m}U_1.
	\end{equation}
\end{lemma}
This identity should be seen as the dual of the forthcoming intertwining relation of $\Delta^+$ and $\Delta$ via $\mathcal{J}_\n$, cf. \eqref{Hai4.14}.

\begin{proof}
	By linearity of $\iota_\n$, it is enough to show \eqref{io02} for $U_1$ and $U_2$ in the set of basis elements \eqref{ao75}, namely $U_1 = D_{(J_1,\m_1)}$ and $U_2 = D_{(J_2,\m_2)}$. Moreover, we assume $(J_1,\m_1) \neq (0,\0)$; otherwise the statement is trivial, since $\rho 1 = {\rm id}$ and $\iota_\n 1 = 0$, cf. \eqref{io01}.
	
	\medskip
	
	We argue by induction in the length $|(J_2,\m_2)|$; the base case amounts to $U_2 = 1$, which is immediate because of $\iota_{\bf n} 1 =0$ and $\varepsilon_\m  1 = \delta_\0^\m$. For the induction step, we distinguish the cases $J_2=0$ and $J_2\neq 0$. The former implies $U_2=D_{(0,\m_2)}$, so that, by $\iota_\n D_{(0,\m_2)} = 0$ and $\varepsilon_{\m} D_{(0,\m_2)} = \delta_\m^{\m_2}$, \eqref{io02} assumes the form
	\begin{equation}\label{*p32}
	\iota_\n D_{(J_1,\m_1)} D_{(0,\m_2)} = \tbinom{\n+\m_2}{\m_2}\iota_{\n+\m_2}D_{(J_1,\m_1)}.
	\end{equation}
	We assume without loss of generality $\m_2 \geq (1,0) $. Recalling that $D_{(0,\m)} = \frac{1}{\m!}\partial^\m$, cf. \eqref{ao75}, we rewrite the l.~h.~s.~of \eqref{*p32} as
	\begin{align}\label{io08}
	\iota_\n D_{(J_1,\m_1)} D_{(0,\m_2)} = \tfrac{1}{(\m_2)_1} \iota_\n D_{(J_1,\m_1)} D_{(0,\m_2-(1,0))} \partial_1.
	\end{align}
	We will now use the following identity, which holds for all $U\in {\rm U}(\mathsf{L})$:
	\begin{equation}\label{io07}
	\iota_\n U\partial_1 = (n_1 + 1)\iota_{\n+(1,0)} U.
	\end{equation}
	To prove \eqref{io07}, by linearity it is enough to take $U$ of the form \eqref{ao75}. If $J=0$, then \eqref{io07} is clear since $\iota_\n D_{(0,\m)}\partial_1 = (m_1 +1)\iota_\n D_{(0,\m + (1,0))}$, cf. \eqref{proddel}, and that vanishes by construction \eqref{io01}. If $J\neq 0$, we choose a pair $(\gamma',\n')$ such that $U = \frac{1}{J(\gamma',\n')+1} \z^{\gamma'} D_{(J,\m)} D^{(\n')}$ and write, with help of \eqref{yc16},
	\begin{equation}\label{det01}
	\begin{split}
	\iota_\n \z^{\gamma'} D_{(J,\m)} D^{(\n')} \partial_1=\; &\iota_\n \z^{\gamma'} D_{(J,\m)} \partial_1 D^{(\n')} \\
	&+ n'_1 \iota_\n \z^{\gamma'} D_{(J,\m)} D^{(\n' - (1,0))}. 
	\end{split} 
	\end{equation}
	We now note that \eqref{io01} implies the following:
	\begin{equation}\label{det02}
	\iota_\n \z^\gamma D_{(J,\m)} D^{(\n)} \neq 0\; \implies (J,\m)=0.
	\end{equation}
	Then by \eqref{det02} the first r.~h.~s.~contribution of \eqref{det01} is always vanishing, since terms coming from $D_{(J,\m)} \partial_1$ have strictly positive length, cf. \eqref{yc16}. The second r.~h.~s.~contribution of \eqref{det01} yields the output of \eqref{io07}.
	
	\medskip
	
	We now combine \eqref{io08}, \eqref{io07} and the induction hypothesis, yielding \eqref{*p32} in form of
	\begin{equation*}
	\begin{split}
	\iota_\n D_{(J_1,\m_1)} D_{(0,\m_2)} &= \tfrac{n_1+1}{(\m_2)_1}\iota_{\n+(1,0)} D_{(J_1,\m_1)} D_{(0,\m_2-(1,0))}\\ 
	&= \tfrac{n_1+1}{(\m_2)_1} \tbinom{\n+\m_2}{\m_2-(1,0)}\iota_{\n+\m_2}D_{(J_1,\m_1)};
	\end{split}
	\end{equation*} 
	this concludes the proof for $J_2 = 0$.
	
	\medskip
	
	In the case of $J_2\neq 0$, and thus $\varepsilon_\m D_{(J_2,\m_2)} =0$, once more we choose a pair $(\gamma',\n')$ such that $J_2(\gamma',\n')\neq 0$ and use \eqref{ao75} to write $U_2 = \tfrac{1}{J'(\gamma',\n')+1} \z^{\gamma'} D_{(J'_2,\m'_2)} D^{(\n')}$. By \eqref{yc15} applied to $U=D_{(J'_2,\m'_2)}$, combined with \eqref{cop01}, we have
	\begin{align*}
	&\iota_\n D_{(J_1,\m_1)} \z^{\gamma'}D_{(J'_2,\m'_2)} D^{(\n')} \\
	&\, = \iota_\n D_{(J_1,\m_1)} \bigg(D_{(J'_2,\m'_2)} \z^{\gamma'} D^{(\n')} -\hspace{-9ex} \sum_{\substack{\beta' \\ (J''_2,\m''_2)+ (J'''_2,\m'''_2) = (J'_2,\m'_2) \\ (J''_2,\m''_2)\neq (0,\0)}}\hspace{-8ex} (D_{(J''_2,\m''_2)})_{\beta'}^{\gamma'}\,\z^{\beta'}D_{(J'''_2,\m'''_2)} D^{(\n')} \bigg).
	\end{align*}
	On the first r.~h.~s. term we apply once more \eqref{yc15} to $U=D_{(J_1,\m_1)} D_{(J'_2,\m'_2)}$. 
	Since by assumption $(J_1,\m_1)\neq (0,\0)$, the product $D_{(J_1,\m_1)} D_{(J'_2,\m'_2)}$ is a linear combination of basis elements \eqref{ao75} with strictly positive length, as may be seen by an iterative application of \eqref{yc15}. We thus appeal to \eqref{det02} to the effect of
	\begin{equation}\label{io03}
	\iota_\n D_{(J_1,\m_1)} D_{(J'_2,\m'_2)} \z^{\gamma'}D^{(\n')} = (\rho D_{(J_1,\m_1)} D_{(J'_2,\m'_2)})\iota_\n \z^{\gamma'}D^{(\n')}.
	\end{equation}
	For the second r.~h.~s. term, we note that if $(J'_2,\m'_2)= (0,\0)$  the sum is empty, so that \eqref{io02} follows from \eqref{io03}.  If $(J'_2,\m'_2)= (0,\0)$ , the length of $\z^{\beta'} D_{(J'''_2,\m'''_2)} D^{(\n')}$ is strictly smaller than that of $D_{(J_2,\m_2)}$, so that by the induction hypothesis the second r.~h.~s. term is given by
	\begin{equation}\label{io04}
	\sum_{\substack{\beta' \\ (J''_2,\m''_2)+ (J'''_2,\m'''_2) = (J'_2,\m'_2) \\ (J''_2,\m''_2)\neq (0,\0)}}\hspace{-6ex} (D_{(J''_2,\m''_2)})_{\beta'}^{\gamma'}\;(\rho D_{(J_1,\m_1)})\,\iota_\n\z^{\beta'}D_{(J'''_2,\m'''_2)} D^{(\n')}.
	\end{equation}
	Note that by \eqref{det02}, $\iota_\n \z^{\beta'} D_{(J'''_2,\m'''_2)} D^{(\n')}$ vanishes for $(J'''_2,\m'''_2)\neq (0,\0)$, so \eqref{io04} further reduces to
	\begin{equation*}
	\sum_{\beta'} (D_{(J'_2,\m'_2)})_{\beta'}^{\gamma'}\; (\rho D_{(J_1,\m_1)}) \iota_\n \z^{\beta'} D^{(\n')}.
	\end{equation*}
	By definition of the algebra morphism $\rho$, this equals \[(\rho D_{(J_1,\m_1)}D_{(J'_2,\m'_2)})\;\iota_\n \z^{\gamma'}D^{(\n')},\] which cancels with \eqref{io03}. Since $(J'_2,\m'_2)\neq (0,\0)$ implies by \eqref{det02} that $\iota_\n \z^{\gamma'} D_{(J'_2,\m'_2)} D^{(\n')} = 0$, this shows that \eqref{io02} holds.
\end{proof}

\medskip

\subsection{The bigrading revisited and finiteness properties}\label{Sect4.5}
\mbox{}

Recall that $\mathsf{L}$ is a bigraded Lie algebra with respect to \eqref{fs02}, and thus ${\rm U}(\mathsf{L})$ becomes a bigraded Hopf algebra. This means that there exists a decomposition ${\rm U}(\mathsf{L}) = \bigoplus_{\mathbf{b} \in \N_0 \times \mathbb{Z}} {\rm U}_{\mathbf{b}}$ such that the concatenation product maps ${\rm U}_{\mathbf{b}'}\otimes {\rm U}_{\mathbf{b}''}$ to ${\rm U}_{\mathbf{b}' + \mathbf{b}''}$ and the coproduct $\cop$ maps ${\rm U}_{\mathbf{b}}$ to $\bigoplus_{\mathbf{b}' + \mathbf{b}'' = \mathbf{b}} {\rm U}_{\mathbf{b}'}\otimes {\rm U}_{\mathbf{b}''}$. Note that this decomposition is different from \eqref{dec}. It turns out that our basis elements \eqref{ao75} are homogeneous:
\begin{lemma}\label{lemhom}
	$D_{(J,\m)}\in {\rm U}_{{\rm bi}(J,\m)}$, where 
	\begin{equation}\label{bigrad}
	\gls{bigrading}(J,\m) := \sum_{(\gamma,\n)} J(\gamma,\n) {\rm bi}(\gamma,\n) + (0,|\m|).
	\end{equation}
\end{lemma}
We adopt the same notation ${\rm bi}$ as in \eqref{fs02}, without any risk of confusion. 
\begin{proof}
	Let us fix a pair $(\gamma,\n)$. We will show that
	\begin{equation}\label{bigrad02}
	U\in {\rm U}_\mathbf{b}\; \implies \z^\gamma U D^{(\n)} \in {\rm U}_\mathbf{b + {\rm bi}(\gamma,\n)};
	\end{equation}
	this clearly proves the lemma, since $D_{(J,\m)}$ is built starting from $\frac{1}{\m!}\partial^\m \in {\rm U}_{(0,|\m|)}$ and applying \eqref{yc11} iteratively. We argue in favor of \eqref{bigrad02} by induction. The case $U=1$ holds since by construction $1\in{\rm U}_{(0,0)}$. We now assume \eqref{bigrad02} to be true for a given $U$ and give ourselves a $D\in \mathsf{L}$  such that $D\in {\rm U}_{\mathbf{b}'}$. We express $\z^\gamma DU D^{(\n)}$ using \eqref{yc12}. Since $DU\in {\rm U}_{\mathbf{b} + \mathbf{b'}}$, our goal is to prove $\z^\gamma DU D^{(\n)}\in {\rm U}_{\mathbf{b} + \mathbf{b'} + {\rm bi}(\gamma,\n)}$. Indeed, for the first r.~h.~s. term of \eqref{yc12}, by the induction hypothesis,
	\begin{equation*}
	D\z^\gamma U D^{(\n)} \in {\rm U}_{\mathbf{b}'+\mathbf{b} + {\rm bi}(\gamma,\n)}.
	\end{equation*}
	For the second r.~h.~s. term, we note that by the compatibility of ${\rm bi}$ and $\prelie$, see Subsection \ref{Sect3.10}, $D_\beta^\gamma \neq 0$ implies
	\begin{equation*}
	{\rm bi}(\beta,\n) = \mathbf{b}' + {\rm bi}(\gamma,\n),
	\end{equation*}
	which combined with the induction hypothesis yields
	\begin{equation*}
	D_\beta^\gamma \z^\beta U D^{(\n)}\in {\rm U}_{\mathbf{b}'+\mathbf{b} + {\rm bi}(\gamma,\n)}.\qedhere
	\end{equation*}
	\end{proof}
	
	\medskip
	
Any linear combination of the bigrading defines a new grading compatible with the Hopf algebra structure of ${\rm U}(\mathsf{L})$. In view of the definition \eqref{ao52} of the homogeneity, a natural choice is to consider the first component weighted by $\alpha$ and the second weighted by $1$; by \eqref{ao52}, this defines
\begin{equation}\label{grad01}
\gls{grading} := \sum_{(\gamma,\n)} J(\gamma,\n)(\lhom\gamma\rhom -|\n|) + |\m|.
\end{equation} 
Thanks to our restriction $\lhom\gamma\rhom > |\n|$, cf. \eqref{ao60}, it holds that 
\begin{equation}\label{posit}
|(J,\m)|_{\rm gr}\geq 0,\text{ and } |(J,\m)|_{\rm gr}= 0 \iff (J,\m) =(0,\0). 
\end{equation}

\medskip

With help of the bigrading \eqref{bigrad} and the grading \eqref{grad01} we will now establish finiteness properties of the action and the product. For this, we first write the basis representations of both maps. It is tautological that the basis representation of the action ${\rm U}(\L)\otimes \T^* \to \T^*$ with respect to \eqref{ao75} and the monomial ``basis''\footnote{
with a slight abuse of language, since the elements $\z^\gamma$ do not constitute a basis of $\T^*$}, i.~e.~
\[ D_{(J,\m)} \otimes \z^\gamma \mapsto \sum_\beta \Delta_{\beta\;(J,{\bf m})}^{\gamma} \z^\beta, \]
is given by
\begin{align}\label{ao77}
\Delta_{\beta\;(J,{\bf m})}^{\gamma}=(\D_{(J,{\bf m})})_\beta^\gamma,
\end{align}
where $(\D_{(J,{\bf m})})_\beta^\gamma$ is the matrix representation
of $\rho \D_{(J,{\bf m})}$ $\in{\rm End}(\T^*)$. 
We choose the notation $\Delta$ since it will give rise to a coaction, cf. \eqref{SG01}, that plays the role of the one in \cite[Subsection 4.2]{Hairer}.

\begin{lemma}\label{Lem4.9}
\begin{equation}\label{SG06}
\begin{split}
&\{ \big((J,\m),\gamma\big) \, \mbox{with }[\gamma]\geq 0\mbox{ or }\gamma =e_\n\; | \, \Delta_{\beta \, (J,\m)}^\gamma \neq 0 \}\\& \quad \mbox{is finite for all}\;\beta.
\end{split}
\end{equation}
Moreover, for $(J,\m)\neq (0,\0)$ we have the triangular structure
\begin{equation}\label{fin10}
\Delta_{\beta \, (J,\m)}^\gamma \neq 0 \implies \lhom\gamma\rhom < \lhom\beta\rhom.
\end{equation}
\end{lemma}
We stress that the restriction of $\gamma$ is crucial in our approach, as will become apparent in the proof. Incidentally, by \eqref{hk02}, it implies the same restriction for $\beta$.

\begin{proof}
We first show by induction in the length $|(J,\m)|$ that 
\begin{equation}\label{bigrad01}
\begin{split}
(\D_{(J,\m)})_\beta^\gamma\neq 0
\implies
&\big(1+[\beta],\sum_{\n'\neq\0}|\n'|\beta(\n')\big) 
\\& \quad = \big( 1+[\gamma] , \sum_{\n'\neq\0} |\n'|\gamma(\n') \big) + {\rm bi}(J,\m).
\end{split}
\end{equation}	
The base case $|(J,\m)|=0$ is trivial, since this implies $\beta = \gamma$. In the induction step, we fix $(J,\m)$ and distinguish two cases: if $J=0$, then the claim follows from $(\partial^{\m})_\beta^\gamma=\sum_{\gamma'} (\partial^{\m-(1,0)})_\beta^{\gamma'}(\partial_1)_{\gamma'}^\gamma$ via \eqref{ao46} and the induction hypothesis in form of
\begin{equation*}
\begin{split}
(\partial^{\m-(1,0)})_\beta^{\gamma'}\neq 0 \implies 
&\big(1+[\beta],\sum_{\n'\neq\0}|\n'|\beta(\n')\big) 
\\& \quad = \big( 1+[\gamma'] , \sum_{\n'\neq\0} |\n'|\gamma'(\n') \big) + (0, |\m - (1,0)|).
\end{split}
\end{equation*}
If $J\neq 0$, the claim likewise follows via \eqref{yc15}, which we may use thanks to \eqref{cop01}, into which we insert \eqref{fs04}. 

\medskip

We now claim that \eqref{SG06} holds true when restricting $(J,\m)$ to be of fixed length. 
Indeed, this is again established by induction in the length $|(J,\m)|$: the argument is the very same as for \eqref{bigrad01}, just starting from \eqref{ao42} and \eqref{yourchoice1} instead of \eqref{ao46} and \eqref{fs04}. The next step is to show that the length $|(J,\m)|$ is bounded. To this end, we first note that for $(\gamma',\n')$ with $J(\gamma',\n')\neq0$ we have $[\gamma']\geq0$, so that the first component of ${\rm bi}(J,\m)$ dominates $|J|$, cf. \eqref{bigrad}. Therefore, assuming that $(D_{(J,\m)})_\beta^\gamma\neq0$, the first component in \eqref{bigrad01} yields $|J| \leq 1+[\beta]$, where we used $[\gamma]\geq -1$. Moreover, taking in \eqref{bigrad01} a linear combination of the first and the second component weighted by $\alpha$ and $1$, respectively, yields by definitions \eqref{ao52} and \eqref{grad01}
\begin{equation}\label{bigrad05}
\lhom\beta\rhom
= \lhom \gamma \rhom
+ |(J,\m)|_{\rm gr}.
\end{equation}
Since $\lhom\gamma'\rhom>|\n'|$ for pairs with $J(\gamma',\n')\neq0$, and since $|\gamma| > 0$, we now obtain $|\m| < |\beta|$. Summing up, we find $|(J,\m)| \leq 1 + [\beta] + \lhom\beta\rhom$, which finishes the proof of \eqref{SG06}. 
Finally, \eqref{fin10} is a straightforward consequence of \eqref{bigrad05} and the positivity of $|\cdot|_{\rm gr}$, cf. \eqref{posit}.
\end{proof}	

\medskip

We shall now give a characterization of the basis representation of the concatenation product, i.~e.~ 
	\begin{equation}\label{prod01}
	\D_{(J',{\bf m}')}\D_{(J'',{\bf m}'')}
	= \sum_{( J,{\bf  m})}(\Delta^+)_{(J',{\bf m}')\,(J'',{\bf m}'')}^{( J,{\bf m})}
	\D_{( J,{\bf m})}.
	\end{equation}
	Writing \eqref{prodcop} in coordinates and using \eqref{cop01}, we see that the numbers
	$(\Delta^+)_{(J',{\bf m}')\,(J'',{\bf m}'')}^{( J,{\bf m})}$
	are determined by the special case where the multi-index
	$( J,{\bf m})$ is of length one. This means either ${\bf m}\in\{(1,0),(0,1)\}$ and $J=0$ or ${\bf m}={\bf 0}$ and the multi-index $ J$ having just one non-trivial entry -- equal to one --
	at $(\gamma,{\bf n})$; for this, we write $ J=e_{(\gamma,{\bf n})}$. The former case is easy; indeed, by \eqref{yc12} and \eqref{yc16}, we see that $J=0$ implies $J'=J''=0$, which reduces all possible situations to formula \eqref{proddel}. In particular, this yields
	\begin{equation}\label{prod05}
	(\Delta^+)_{(J',\m')(J'',\m'')}^{(0,(1,0))} = \delta_{(J',\m')+(J'',\m'')}^{(0,(1,0))}, 
	\end{equation}
	and a similar statement for $\m=(0,1)$. The case $J=e_{(\gamma,\n)}$ is characterized by the application of \eqref{io02}; indeed, $(\Delta^+)_{(J',{\bf m}')\,(J'',{\bf m}'')}^{( e_{(\gamma,\n)}, \0)}$ is the coefficient of $\z^\gamma$ in $\iota_\n D_{(J',\m')} D_{(J'',\m'')}$.

\medskip

\begin{lemma}\label{Lem4.10}
	\begin{equation}\label{SG03}
	\begin{split}
	&\{ (J',\m'),(J'',\m'') \, | \, (\Delta^+)^{(J,\m)}_{(J',\m')(J'',\m'')} \neq 0 \}\\&\quad\mbox{is finite for all}\;(J,\m).
	\end{split}
	\end{equation}
\end{lemma}
\begin{proof}
Once more by \eqref{prodcop} and \eqref{cop01} it is enough to show \eqref{SG03} for $|(J,\m)|=1$, see the discussion after \eqref{prod01}. The case $J=0$ and $|\m|=1$ is trivial from \eqref{prod05}. For $|J|=1$ and $\m=\0$, we write $J=e_{(\beta,\n)}$ and claim that \eqref{SG03} follows from \eqref{io02}. For this, we apply \eqref{io02} with $U_1=D_{(J',\m')}$ and $U_2=D_{(J'',\m'')}$, and consider the coefficient of the $\z^\beta$-term, which is non-vanishing by assumption. By \eqref{det02}, the first r.~h.~s. term in \eqref{io02} is non-vanishing only if $U_2 = \z^\gamma D^{(\n)}$ for some $\gamma$; applying \eqref{SG06} to the first factor $\rho U_1$, we see that there are only finitely many $\gamma$'s and $(J',\m')$'s which give a non-vanishing contribution to the $\z^\beta$-coefficient. Turning to the second r.~h.~s. term of \eqref{io02}, its $\z^\beta$-coefficient is non-zero unless $U_1 = \z^\beta D^{(\n+\m)}$ and $U_2 = D_{(0,\m)}$ for some $\m$. The constraint $|\n+\m|<\lhom\beta\rhom$ only allows for finitely many $\m$'s.
\end{proof}


\subsection{Dualization leading to $\T^+$, $\Delta^+$ and $\Delta$}\label{Sect4.3}
\mbox{}

We consider the\footnote{
unique up to linear isomorphisms}
linear space $\gls{hopfSpace}$, with basis $\{ \gls{basisTplus} \}_{(J,{\bf m})}$ indexed by $(J,\m)$ and the canonical non-degenerate pairing between ${\rm U}(\L)$ and $\T^+$ given by
\begin{equation}\label{dual01}
\langle\D_{(J',{\bf m'})},\Z^{(J,{\bf m})}\rangle=\delta^{(J,{\bf m})}_{(J',{\bf m'})}.
\end{equation}
As a consequence of \eqref{dual01}, ${\rm U}(\L)$ canonically is a subspace of $(\T^+)^*$; note that the latter is much larger, since it is the direct product over the index set of all $(J,\m)$'s, whereas ${\rm U}(\L)$ is just the direct sum.

\medskip

Our next goal is to provide a structure for $\T^+$ by dualization of the Hopf algebra and the module structures of ${\rm U}(\mathsf{L})$. First, we note that the basis representation of a coproduct has the algebraic properties of a product, and thus \eqref{cop01} defines a product\footnote{
We will omit the dot in the notation.} $\cdot$
 in $\T^+$ given by
\begin{equation}\label{prodT1}
\Z^{(J',{\bf m'})} \Z^{(J'',{\bf m''})}=\Z^{(J',{\bf m'})+(J'',{\bf m''})}.
\end{equation}
This way $(\T^+,\cdot)$ becomes the (commutative) polynomial algebra over variables indexed by the index set of $\L$. 

\medskip

In a similar way, we want to transpose the action and the coproduct mentioned in the previous subsection. The transposition in these two cases is possible thanks to the finiteness properties which were stated in Lemmas \ref{Lem4.9} and \ref{Lem4.10}.
Starting with the action, analogously to \eqref{ao107}, from the basis representation \eqref{ao77} we define a map $\Delta:\T \to \T^+\otimes \T$ by 
\begin{equation}\label{SG01}
\gls{coaction} \z_\beta 
= \sum_{\gamma, \,(J,\m)} \Delta_{\beta \; (J,\m)}^{\gamma} \, \Z^{(J,\m)} \otimes \z_\gamma.
\end{equation}
The sum is finite due to \eqref{SG06}, and hence $\Delta$ is well-defined. We stress that the restriction to $\T$ is crucial for our argument; 
it does not seem possible to extend $\Delta$ to
$\mathbb{R}[[\z_k,\z_{\bf n}]]^\dagger
\rightarrow \mathsf{T}^+ \otimes \mathbb{R}[[\z_k,\z_{\bf n}]]^\dagger$.

\medskip

We now turn to the product; by the basis representation \eqref{prod01}, we define a map $\Delta^+:\T^+\to\T^+\otimes\T^+$ via
\begin{equation}\label{copT1}
\gls{coproduct} \Z^{(J,\m)}=\sum_{(J',\m'),\,(J'',\m'')}(\Delta^+)_{(J',\m')\,(J'',\m'')}^{(J,\m)}\,\Z^{(J',\m')}\otimes\Z^{(J'',\m'')}.
\end{equation}
Such a map has the algebraic properties of a coproduct in $\T^+$. The fact that this map is well-defined is a consequence of the finiteness property \eqref{SG03}. 

\medskip

The only missing ingredient to make $\T^+$ a Hopf algebra is an antipode\footnote{ The existence of unit and counit maps easily follows from transposing the counit and unit, respectively, of ${\rm U}(\mathsf{L})$. No finiteness properties are required.}  $\mathcal{S}$: since \eqref{grad01} and \eqref{posit} make $\T^+$ a connected\footnote{ i.~e.~the zero-degree subspace is $\R$, cf. \cite[Definition 2.10.6]{HGK10}} graded bialgebra, this is guaranteed by general theory, see \cite[Proposition 3.8.8]{HGK10}. 

\medskip

These observations are collected in the following result.
\begin{proposition}\label{cor02}
	Let $\Delta^+ : \T^+ \to \T^+\otimes \T^+$ 
	be given by \eqref{copT1}.	Then there exists a map $\mathcal{S}$ such that $(\T^+,\cdot,\Delta^+,\mathcal{S})$ is a Hopf algebra with antipode $\mathcal{S}$. Moreover, let $\Delta : \T \to \T^+ \otimes \T $
	be given by \eqref{SG01}. Then $(\T,\Delta)$ is a (left-) comodule over $\T^+$, i.~e.
	\begin{equation}\label{SG04}
	(\mathrm{id}\otimes \Delta) \Delta = (\Delta^+ \otimes \mathrm{id})\Delta.
	\end{equation}
\end{proposition}
Note that \eqref{SG04} is the dualization of the morphism property of $\rho$.

\medskip

We now have introduced all the objects required to construct a structure group $\mathsf{G}\subset \textnormal{End}(\T)$ according to \cite[Section 4.2]{Hairer}. A minor difference is that, in our case, $(\T,\Delta)$ is a left comodule while in \cite[(4.15)]{Hairer} it is a right comodule, a fact that transfers to \eqref{Hai4.14} and more upcoming identities. This does not affect the construction. In fact, with a similar (though more cumbersome) definition of the Lie algebra $\dum{L}$, working at the level of the transposed endomorphisms from the beginning, we would have been able to recover the same structure, but paying the price of blurring the connection to the actions on $(a,p)$-space that served as a motivation in Section \ref{Sect1}.

\subsection{Intertwining of $\Delta$ and $\Delta^+$ through $\mathcal{J}_\n$}\label{Sect5.6}
\mbox{}

Let us define for every $\n$ a map $\gls{embeddingJ} : \T\to \T^+$ in coordinates by
\begin{equation}\label{hai01}
\J_{\n}\z_{\gamma}=\left\{\begin{array}{ll}
\n! \, \Z^{(e_{(\gamma,\n)},\0)}  &  \textnormal{ if }[\gamma]\geq 0,\ \lhom\gamma\rhom>|\n|\\
0  &  \textnormal{ otherwise} 
\end{array}
\right\}.
\end{equation}
Note that in view of \eqref{io01} $\mathcal{J}_\n$ is the transposition of $\iota_\n$ up to a combinatorial factor: 
\begin{align}\label{fw26}
\mathcal{J}_{\n} = \n! \iota_\n^\dagger. 
\end{align}
The normalization with $\n!$ is made such that the dualization of \eqref{io02} takes the form of the following intertwining relation between the coaction $\Delta$ and the coproduct $\Delta^+$, which is an identity in $\T^+ \otimes \T^+$:
\begin{equation}\label{Hai4.14}
\Delta^+ \J_{\n}\z_\gamma = (\mathrm{id} \otimes \J_{\n}) \Delta \z_\gamma + \sum_{\m} \J_{\m+\n}\z_\gamma \otimes \frac{\Z^{(0,\m)}}{\m!}.
\end{equation}
Indeed, \eqref{Hai4.14} amounts to \eqref{io02} once tested with $U_1\otimes U_2$, where we use the pairings \eqref{pl95} and \eqref{dual01}, the definitions \eqref{prod01} and \eqref{ao77}, and the fact that $\langle U,Z^{(0,\m)}\rangle = \varepsilon_\m U$ in view of \eqref{seemynotesonp38}.  
Combined with 
\begin{equation}\label{Hai4.14b}
\Delta^+ \Z^{(0,(1,0))} = \Z^{(0,(1,0))}\otimes 1 + 1\otimes\Z^{(0,(1,0))},
\end{equation}
which follows from \eqref{prod05}, we see that $\Delta^+$ is determined by $\Delta$ through $\mathcal{J}_\n$ in agreement with regularity structures, cf. \cite[(4.14)]{Hairer}.
Let us also mention that the coaction applied to the polynomial sector $\bar{\mathsf{T}}$ is in agreement\footnote{
up to the constant in $\bar{\mathsf{T}}$ which we modded out, cf. \eqref{as35} and \eqref{ao106}}
with \cite[p. 23]{Hairer}, 
\begin{align}\label{comodulePoly}
\Delta\mathsf{z}_{e_{\bf n}}
=\sum_{\substack{{\bf n}'+{\bf n}''={\bf n}\\{\bf n''}\not=\0}}
\tbinom{\n}{\n'} Z^{(0,{\bf n}')}\otimes\mathsf{z}_{e_{{\bf n}''}}.
\end{align}
This may be seen by \eqref{SG01} and \eqref{ao77}. 
First note that $\rho D_{(J,\n')}$ preserves $\mathsf{\tilde T}^*$, as a consequence of the same property of $\mathsf{L}$, see (\ref{ao61}), 
and that $\rho D_{(J,\n')}$ maps to $\bar{\mathsf{T}}^*$ only if $J=0$ as can be read off \eqref{ao75}.
Therefore, \eqref{comodulePoly} follows from $D_{(0,\n')} \z_{\n''} = \binom{\n'+\n''}{\n'} \z_{\n'+\n''}$, which is a consequence of \eqref{ao75} and \eqref{ao30}.


\section{The group structure}\label{Sect5old}

\subsection{The structure group $\dum{G}$}\label{Sect5.1old}
\mbox{}

With all the algebraic objects defined in Section \ref{Sect4}, we follow \cite[Subsection 4.2]{Hairer} in the construction of the structure group. Let us consider the space of multiplicative linear functionals on $\T^+$, which we denote by  $\Alg (\T^+,\R)$. Writing \begin{equation}\label{**p38}
f^{(J,\m)}:=f \Z^{(J,\m)}
\end{equation}
for $\gls{multFunc} \in (\T^+)^*$, by \eqref{prodT1} the space $\Alg (\T^+,\R)$ is characterized by
\begin{equation}\label{ope03}
f^{(J',\m')+(J'',\m'')}=f^{(J',\m')}f^{(J'',\m'')}\,\text{ and }\,f^{(0,\0)}=1.
\end{equation}
Due to this property, the elements $f\in \text{Alg}(\T^+,\R)$ are parameterized by\footnote{By \eqref{hai01} and \eqref{SG05}, the coefficients $\pi_\gamma^{(\n)}$ may be identified with \cite[(4.8)]{Hairer}.}
\begin{equation}\label{SG05}
f^{(J,\m)}=h^{\m}\prod_{(\gamma,\n)}(\pi_{\gamma}^{(\n)})^{J(\gamma,\n)},
\end{equation}
where $h\in\R^{2}$ and $\{\pi^{(\n)}\}_{\n}\subset \mathsf{\tilde T}^*$ is constrained by 
\begin{equation}\label{constraintPi}
\pi_\gamma^{({\bf n})}\neq 0 \quad\implies\quad |\gamma|>|{\bf n}|.
\end{equation}
From the Hopf algebra structure of $\T^+$, the space $\Alg (\T^+,\R)$ inherits a natural group structure, namely the convolution product of functionals:
\begin{equation}\label{conv01}
fg := (f\otimes g)\Delta^+.
\end{equation}
The neutral element $e$ of this group, which is the counit of $\T^+$, maps $Z^{(0,\0)}$ to $1$ and every other basis element of $\T^+$ to $0$, and the inverse elements are given by $f^{-1} = f  \mathcal{S}$, cf. \cite[Theorem 2.1.5]{Abe}.

\medskip

Following \cite[Subsection 4.2]{Hairer}, we now define a map $\Gamma: (\T^+)^* \to \End(\T)$ by
\begin{equation}\label{defgamma}
\gls{gammaF} := (f \otimes \mathrm{id})\Delta.
\end{equation}
Then the set 
\begin{equation*}
\gls{structuregroup}:= \{ \Gamma_f \, | \, f\in\Alg (\T^+,\R) \}\subset\End (\T)
\end{equation*}
inherits the group structure of $\Alg (\T^+,\R)$, where 
\begin{equation}\label{groupstruct01}
\Gamma_e=\mathrm{id},\ \Gamma_{fg} = \Gamma_f \Gamma_g\ \text{and }\Gamma_{f^{-1}} = \Gamma_f^{-1};
\end{equation}
we call $\mathsf{G}$ \textit{structure group}. Applying definition \eqref{defgamma} to $\z_\beta$, plugging in \eqref{SG01}, \eqref{**p38} and \eqref{ao77}, we obtain from \eqref{ao107} the representation\footnote{Here and in the sequel we identify $\D_{(J,\m)}$ with the corresponding endomorphism $\rho \D_{(J,\m)}$.}
\begin{equation}\label{SG07}
\Gamma_f = \sum_{(J,\m)} f^{(J,\m)} \, \D_{(J,\m)}^\dagger,
\end{equation}
and note that this sum is effectively finite because of \eqref{SG06}. Moreover, as a consequence of \eqref{fin10},
\begin{equation}\label{tri04}
 (\Gamma_f-\mathrm{id})_\beta^\gamma \neq 0 \implies \lhom\gamma\rhom < \lhom \beta \rhom;
\end{equation} 
this may be rewritten more in line with the corresponding requirement in \cite[Definition 3.1]{Hairer}:
\begin{equation}\label{tri01}
(\Gamma_f-\mathrm{id})\T_\kappa \subset \bigoplus_{\kappa'<\kappa} \T_{\kappa'}.
\end{equation}
The elements of the structure group behave nicely with the polynomial sector $\bar{\T}$, see Subsection \ref{Sect3.3}, in the sense that for $f\in\Alg (\T^+,\R)$ with $h\in\R^2$ being its parameter according to \eqref{SG05}, and for all $\n\neq \0$,
\begin{equation}\label{pol01}
\Gamma_f x^\n = \sum_{\m < \n} \tbinom{\n}{\m}h^\m x^{\n-\m},
\end{equation}
which follows from \eqref{defgamma}, \eqref{comodulePoly} and \eqref{SG05}.
In Subsection \ref{Sect.ext}, we argue that this is in line with
\cite[Assumption 3.20]{Hairer}.
%

\subsection{Consistency of $\dum{G}^*$ with our goals}\label{Sect5.2old}
\mbox{}

Our final goal is to identify the construction of Subsection \ref{Sect5.1old} with the group of transformations on $(a,p)$-space heuristically described in Section \ref{Sect1}. We define $\mathsf{G}^*:=\{\Gamma^*\,|\,\Gamma\in\mathsf{G}\}$, 
which due to \eqref{SG07} is formed by the maps
\begin{equation}\label{gstar}
\Gamma_f^*=\sum_{(J,\m)} f^{(J,\m)} \, \D_{(J,\m)}\in\textnormal{End}(\T^*)\ \textnormal{ where }\ f\in \textnormal{Alg}(\T^+,\R).
\end{equation}
From \eqref{groupstruct01}, $\mathsf{G}^*$ inherits a group structure given by
\begin{equation}\label{groupstruct02}
\Gamma_e^*=\mathrm{id},\ \Gamma_{fg}^* = \Gamma_g^* \Gamma_f^*\ \text{and }\Gamma_{f^{-1}}^* = (\Gamma_f^*)^{-1};
\end{equation}
note that the order in the composition rule is reversed as a consequence of transposition.

\medskip

We gather all our results in the following proposition.
\begin{proposition}\label{exp01}
Let $h\in\R^2$ and $\{\pi_{\gamma}^{(\n)}\}_{(\gamma,\n)}\subset \R$ generate $f$ through the characterization \eqref{SG05}. Let $\pi^{(\n)}\in\T^*$ for every $\n\in\N_0^2$ be given by	
	\begin{equation}\label{ext1}
	\pi^{(\n)}:=\sum_{[\gamma]\geq 0}\pi_{\gamma}^{(\n)}\z^{\gamma}+\sum_{\m\neq \0}\pi_{e_{\m}}^{(\n)}\z_{\m},
	\end{equation}
	where  
	\begin{equation}\label{poly_pi}
	\pi_{e_{\m}}^{(\n)}:=\left\{\begin{array}{cl}
	\tbinom{\m}{\n}h^{\m-\n} & \textnormal{if }\n<\m\\
	0 & \textnormal{otherwise}
	\end{array}\right\}.
	\end{equation}
		i) The following formula holds
		\begin{equation}\label{LOsg3}
		\Gamma_f^*=\sum_{k\geq 0}\frac{1}{k!}\sum_{\n_1,...,\n_k}\pi^{(\n_{1})}\cdots\pi^{(\n_k)}D^{(\n_k)}\cdots D^{(\n_1)}.
		\end{equation}
		In particular,
		\begin{align}
		\Gamma_f^* \z_k &= \sum_{l\geq 0} \tbinom{k+l}{k} (\pi^{(\0)})^l \, \z_{k+l}\text{ for all }k\geq 0, \label{act1}\\
		\Gamma_f^* \z_\n &= \z_\n +\pi^{(\n)}\text{ for all }\n\neq\0. \label{act2}
		\end{align}
		
		ii)  For all  $\pi_1,...,\pi_k \in \T^*$ such that $\pi_1\cdots\pi_k$ $\in\T^*$, 
		\begin{equation}\label{eqmult}
		\Gamma_f^*\pi_1 \cdots \pi_k=(\Gamma_f^*\pi_1)\cdots(\Gamma_f^*\pi_k).
		\end{equation}
		iii) The composition rule \eqref{as33} holds.
		
		iv) If $\{\pi^{(\n)}\}_\n \subset \T^* \cap \R[\z_k,\z_\n]$, then for all $(a,p)$ and $\pi\in\T^*\cap \R[\z_k,\z_\n]$,
		\begin{align}\label{act01}
		\Gamma_f^*\pi[a,p]=\pi\Big[a\big(\cdot+ \pi^{(\0)}[a,p] \big),
		p+\sum_{{\bf n}\not={\bf 0}}\pi^{({\bf n})}[a,p]x^{\bf n}\Big].	
		\end{align}
		v) The subset $\bar{\dum{G}}^*\subset \dum{G}^*$ generated via \eqref{gstar} by $f\in\textnormal{Alg}(\T^+,\R)$ of the form \eqref{SG05} with $\pi_{\gamma}^{(\n)}=0$ for all $(\gamma,\n)$ is a subgroup isomorphic to $(\R^2,+)$. Moreover, if $\Gamma_f^*\in\bar{\dum{G}}^*$ with $h\in\R^2$ as in \eqref{SG05}, then for all $(a,p)$ and $\pi\in\T^*\cap \R[\z_k,\z_\n]$,
		\begin{equation}\label{act03}
		\Gamma_f^*\pi[a,p]=\pi\Big[a\big(\cdot+ p(h) \big),	p(\cdot+h)-p(h)\Big].
		\end{equation}
		\\
		vi) The subset $\widetilde{\dum{G}}^*\subset\dum{G}^*$ generated via \eqref{SG07} by $f\in\textnormal{Alg}(\T^+,\R)$ of the form \eqref{SG05} with $h=(0,0)$ is a subgroup.

\end{proposition}
The reader should see \eqref{act01} as a variant of \eqref{ao81} for the functions on $(a,p)$-space given by polynomials $\pi\in\T^*\cap \R[\z_k,\z_\n]$, where we interpret $\pi^{(\n)}$'s as in \eqref{ext1}. The subgroups $\bar{\dum{G}}^*$ and $\widetilde{\dum{G}}^*$ correspond to shifts and ($(a,p)$-dependent) tilts, respectively: the shift \eqref{ao80} is recovered by \eqref{act03}, whereas \eqref{act01} translates into \eqref{ao81} (since also \eqref{ao81} includes \eqref{ao80}). It is however not possible to recover the tilt by an $(a,p)$-independent polynomial, namely \eqref{ao79}, because the $\pi^{(\n)}$'s are restricted by \eqref{ao60} which does not allow $(a,p)$-independent expressions for large $|\n|$. Finally, although \eqref{act01} and \eqref{act03} hold for all possible  $\pi\in \T^*\cap \R[\z_k,\z_\n]$, there is no hope to extend them to $\T^*$ because generic elements of $\R[[\z_k,\z_\n]]$ cannot be identified with functions of $(a,p)$; the same applies for $\{\pi^{(\n)}\}_\n$ in \eqref{act01}. 

\begin{proof}
We first show \eqref{eqmult}; indeed, it is a direct consequence of \eqref{ope03} and the following generalized Leibniz rule: For all $\pi_1,...,\pi_l\in\R[[\z_k,\z_\n]]$
\begin{equation}\label{leib01}
\Sig{\m}{J}\pi_1 \cdots \pi_l
=\sum (\Sig{\m_1}{J_1}\pi_1)\cdots(\Sig{\m_l}{J_l}\pi_l),
\end{equation}
where the sum runs over all $(J_1,\m_1),...,(J_l,\m_l)$ with $(J_1,\m_1)+\ldots+(J_l,\m_l)=(J,\m)$. It is easy to see that by induction, (\ref{leib01}) for general $l\in\mathbb{N}$ follows from 
the case $l=2$. In view of (\ref{cop01}), this case reduces to
\begin{align}\label{leib02}
U\pi_1\pi_2=\sum_{(U)}(U_{(1)}\pi_1)(U_{(2)}\pi_2)
\end{align}
for $U\in{\rm U}(\mathsf{L})$, where we do not distinguish between 
$\rho U \in{\rm End}(\mathsf{T}^*)$ and $U$. Formula (\ref{leib02}) is trivial for 
$U=1$; it is obvious for $U\in\mathsf{L}\subset{\rm Der}(\mathsf{T}^*)$.
It remains to pass from $U$ to $UD$ for some $D\in\mathsf{L}$. 
For the l.~h.~s.~of (\ref{leib02}) we note that by induction hypothesis
(and base case) we have
\begin{align*}
UD\pi_1\pi_2=\sum_{(U)}\big((U_{(1)} D \pi_1)(U_{(2)}\pi_2)+(U_{(1)}\pi_1)(U_{(2)} D\pi_2)\big).
\end{align*}
For the r.~h.~s.~of (\ref{leib02}) we have by the compatibility of the coproduct
with concatenation (composition) and (\ref{cop0203})
\begin{align*}
{\rm cop}UD=\sum_{(U)}\big(U_{(1)}D\otimes U_{(2)}+U_{(1)}\otimes U_{(2)}D\big).
\end{align*}

\medskip

We now turn to the proof of \eqref{LOsg3}. We first argue that the r.~h.~s. of \eqref{LOsg3}, when interpreted as an endomorphism of $\T^*$, is effectively finite (note that we already know that the l.~h.~s. is effectively finite from \eqref{SG06}). For this, we note that the r.~h.~s.~of \eqref{LOsg3} is an infinite sum of terms of the form
\begin{equation*}
\z^{\gamma_1}\cdots \z^{\gamma_k} D^{(\n_k)}\cdots D^{(\n_1)},
\end{equation*}
where either $[\gamma_i]\geq 0$ or $\gamma_i \in \{e_\n\}_{\n\neq \0}$ for $i=1,...,k$. We extend the family of derivations $\tilde{\mathsf{L}}$ by incorporating purely polynomial multi-indices, so that we consider the set $\{\z^{\gamma'} D^{(\n)}\}_{[\gamma']\geq 0,|\gamma'|>|\n|}$ $\cup \{\z_\m D^{(\n)}\}_{\m>\n}$. It can be easily checked that this family is closed under the standard pre-Lie product $\prelie$ given by the first item in \eqref{fs01}: For the mixed terms this follows from $(\z_\m D^{(\0)})\z^\gamma$ $=$ $\sum_{k\geq 0}(k+1)\gamma(k) \z^{\gamma -e_k + e_{k+1} + e_\m}$ with $|\gamma - e_k + e_{k+1} + e_\m| = |\gamma| + \alpha + |\m| > |\gamma|$, cf. \eqref{ao24}; from $(\z_\m D^{(\n)}) \z^{\gamma'} = \gamma'(\n) \z^{e_\m + \gamma - e_\n}$ with $|e_\m + \gamma - e_\n| = |\gamma'| + |\m| - |\n| > |\gamma| > |\n'|$, cf. \eqref{ao30};  and from $(\z^{\gamma'} D^{(\n')}) \z_\m = \delta_\m^{\n'} \z^{\gamma'}$.

\medskip

With this extension at hand, we follow the strategy of Lemma \ref{Lem4.9}. For this we need the two following properties:
\begin{itemize}
	\item extension of \eqref{fs04} to purely polynomial multi-indices, i.~e.~
	\begin{align*}
	&\lefteqn{(\z_\n D^{({\bf n}')})_\beta^\gamma\not=0}\\
	&\Longrightarrow\quad
	\left\{\begin{array}{ccl}
	[\beta]&=&[\gamma],\\
	{\displaystyle\sum_{{\bf m}\not=\0}|{\bf m}|\beta({\bf m})}&=&
	{\displaystyle\sum_{{\bf m}\not=\0}|{\bf m}|\gamma({\bf m})
		+|\n| - |\n'|}
	\end{array}\right\};
	\end{align*}
	
	\item extension of \eqref{yourchoice1} to purely polynomial multi-indices, i.~e.
	\begin{equation*}
	\{(\gamma', (e_\n,\n'))\;|\; (\z_\n D^{(\n')})_{\beta'}^{\gamma'} \neq 0\}\mbox{ is finite for all }\beta'.
	\end{equation*} 
\end{itemize}
Indeed, both follow from \eqref{ao24} and \eqref{ao30}. Under these hypotheses, the inductive argument in the proof of Lemma \ref{Lem4.9} may be used here to show the analogue of \eqref{SG06}, and thus effective finiteness of the r.~h.~s.~ of \eqref{LOsg3}.

\medskip

Hence, to show \eqref{LOsg3} it is enough to apply both sides to a monomial $\z^\gamma$ with $[\gamma]\geq 0$ or $\gamma\in\{e_\n\}_{\n\neq \0}$. Note that both sides are multiplicative\footnote{ Indeed, this holds for any effectively finite expression of the form of the r.~h.~s. of \eqref{LOsg3} with $D^{(\n)}$'s being commuting derivations and $\pi^{(\n)}$'s being multiplication operators.}, so it is enough to consider $\gamma$'s of length one. Thus, in view of \eqref{ao35} and \eqref{ao30}, showing \eqref{LOsg3} amounts to showing \eqref{act1} and \eqref{act2}.
We start with \eqref{act1}. 
Using \eqref{gstar} and applying \eqref{fw13} (with the roles of $l$ and $k$ flipped), we see that it is enough to show that 
$$
\sum_{(J,\m)} f^{(J,\m)} (\iota_\0 +\sum_{\n\neq\0} \varepsilon_\n\otimes\z_\n )^k D_{(J,\m)} = (\pi^{(\0)})^k .
$$
Here we interpret
$\iota_\0+\sum_{{\bf n}\not=\0}\varepsilon_{\bf n}\otimes\mathsf{z}_{\bf n}$ as a linear
map from ${\rm U}(\mathsf{L})$, a space endowed with coproduct, into
the algebra $\mathbb{R}[[\mathsf{z}_k,\mathsf{z}_{\bf n}]]$, 
so that powers make sense. 
By the binomial formula\footnote{
Note that the coproduct in $\mathrm{U}(\mathsf{L})$ is co-commutative, see \eqref{cop01}, therefore the product of linear maps from $\mathrm{U}(\mathsf{L})$ to any commutative algebra is Abelian.} 
applied to $(\iota_\0 +\sum_{\n\neq\0} \varepsilon_\n\otimes\z_\n )^k$ and \eqref{cop01}, the l.~h.~s. equals
\begin{align*}
\sum_{\substack{(J_1,\m_1),\dots,(J_k,\m_k) \\ 0\leq k'\leq k}} 
\tbinom{k}{k'} f^{(J_1,\m_1)+\dots+(J_k,\m_k)} 
(\iota_\0 D_{(J_1,\m_1)}) \cdots (\iota_\0 D_{(J_{k'},\m_{k'})}) \\
\times \sum_{\n\neq\0} (\varepsilon_\n D_{(J_{k'+1},\m_{k'+1})}) \z_\n \cdots
\sum_{\n\neq\0} (\varepsilon_\n D_{(J_{k},\m_{k})}) \z_\n .
\end{align*}
By \eqref{SG05}, \eqref{io01} and \eqref{seemynotesonp38}, this equals
\begin{align*}
\sum_{0\leq k' \leq k} \tbinom{k}{k'} \Big( \sum_\gamma \pi_\gamma^{(\0)} \z^\gamma \Big)^{k'} \Big( \sum_{\m\neq\0} h^\m \z_\m \Big)^{k-k'},
\end{align*}
which by the binomial formula and \eqref{ext1} is seen to coincide with $(\pi^{(\0)})^k$.
To show \eqref{act2}, we appeal to \eqref{gstar} to see
\begin{align*}
\Gamma_f^* \z_\n = \z_\n + \sum_{J\neq0} f^{(J,\0)} D_{(J,\0)} \z_\n + \sum_{\m\neq\0} f^{(0,\m)} D_{(0,\m)} \z_\n
\end{align*}
since $D_{(J,\m)}$ with $J\neq0$ and $\m\neq\0$ would annihilate $\z_\n$, cf. \eqref{ao75}.
The first sum has contributions only from $J=e_{(\gamma,\n)}$, therefore by \eqref{ao30} and \eqref{SG05} we have
\begin{align*}
\Gamma_f^* \z_\n = \z_\n + \sum_{\gamma} \pi_\gamma^{(\n)} \z^\gamma + \sum_{\m\neq\0} \tbinom{\n+\m}{\n} h^\m \z_{\n+\m},
\end{align*}
which together with \eqref{ext1} yields \eqref{act2}.

\medskip

We now turn to the proof of the composition rule \eqref{as33}.
Given $\pi^{(\n)},\pi'^{(\n)}$ with corresponding $\Gamma^*,\Gamma'^*$, we define $\bar\pi^{(\n)}:=\pi^{(\n)}+\Gamma^*\pi'^{(\n)}$ and have to show that the corresponding $\bar\Gamma^*$ satisfies $\bar\Gamma^*=\Gamma^*\Gamma'^*$.
By multiplicativity \eqref{eqmult}, it is enough to show that $\bar\Gamma^*$ coincides with $\Gamma^*\Gamma'^*$ on the coordinates $\{\z_k,\z_{\n}\}_{k\geq 0,\n\neq \0}$.
For $\n\neq\0$, we obtain by applying \eqref{act2} three times
\[
\Gamma^*\Gamma'^* \z_\n
= \Gamma^* (\z_\n+\pi'^{(\n)}) 
= \z_\n + \pi^{(\n)} + \Gamma^*\pi'^{(\n)}
= \bar\Gamma^* \z_\n.
\]
For $k\geq0$ we apply \eqref{act1} twice and use multiplicativity \eqref{eqmult} to obtain
\begin{align*}
\Gamma^*\Gamma'^*\z_k 
&= \Gamma^* \sum_{l'\geq0} \tbinom{k+l'}{k} (\pi'^{(\0)})^{l'} \, \z_{k+l'} \\
&= \sum_{l'\geq0} \tbinom{k+l'}{k} (\Gamma^*\pi'^{(\0)})^{l'} \sum_{l\geq0} \tbinom{k+l'+l}{k+l'} (\pi^{(\0)})^l \, \z_{k+l'+l}.
\end{align*}
A re-summation together with the binomial formula and applying once more \eqref{act1} yield 
\[
\Gamma^*\Gamma'^*\z_k 
= \sum_{\bar{l}\geq0} \tbinom{k+\bar{l}}{k} (\pi^{(\0)} +\Gamma^*\pi'^{(\0)} )^{\bar{l}} \, \z_{k+\bar{l}}
= \bar\Gamma^* \z_k,
\]
which finishes the proof of \eqref{as33}.

\medskip

To prove \eqref{act01}, by \eqref{eqmult} it is enough to show it for $\pi\in\{\z_k,\z_{\n}\}_{k\geq 0,\n\neq \0}$. These special cases are a consequence of \eqref{act1}, \eqref{act2} and Taylor's formula.

\medskip

We now argue in favor of \textit{v)}. By \eqref{groupstruct02}, we have to show that for $f$ of the form
\begin{equation}\label{square}
f^{(J,\m)} = \left\{\begin{array}{cl}
h^\m & \mbox{if}\ J=0\\
0 & \mbox{otherwise}
\end{array}
\right\}
\end{equation}
the product \eqref{conv01} amounts to addition of $h\in\R^2$. Indeed, 
\begin{align*}
(f'f'')^{(J,\m)} 
&\,\underset{\eqref{conv01},\eqref{square}}{=} 
\sum_{\m',\m''} (\Delta^+)^{(J,\m)}_{(0,\m')(0,\m'')} (h')^{\m'} (h'')^{\m''}\\
&\underset{\eqref{proddel},\eqref{prod01}}{=}
\delta_{0}^{J} \sum_{\m'+\m''=\m} \tbinom{\m'+\m''}{\m'} (h')^{\m'} (h'')^{\m''},
\end{align*}
and we conclude by the binomial formula. 
For $f$ of the form \eqref{square}, and using \eqref{ao20bis}, \eqref{ext1} assumes the form $\sum_\n \pi^{(\n)}[a,p] x^\n = p(x+h)$, so that \eqref{act01} turns into \eqref{act03}.

\medskip

We finally turn to the proof of \textit{vi)}. By \eqref{groupstruct02}, it suffices to show that the set of $f\in\Alg (\T^+,\R)$ such that 
\[
f^{(J,\m)} = 0 \quad \mbox{for } \m\neq\0
\] is a subgroup; this is a direct consequence of ${\rm U}(\tilde{\mathsf{L}})$, cf. \eqref{p26}, being a sub-Hopf algebra of ${\rm U}(\mathsf{L})$.
\end{proof}


\medskip


\subsection{Relation to Hairer's regularity structure}\label{Sect.ext}
\mbox{}

In this subsection, while keeping $\mathsf{G}$ as an abstract group,
we enlarge the abstract model space $\mathsf{T}$ on which it acts. We do so in order to
draw a closer connection to \cite{Hairer}. We will proceed in two steps,
first enlarging the abstract model space by a placeholder for the omitted constants,
and then by a placeholder for the right-hand side of the equation.

\medskip

While thinking of $p$ only modulo constants was an important guiding principle 
in uncovering the algebraic structure, see Section \ref{Sect1},
we will now re-introduce constants into the polynomial sector $\mathsf{\bar T}$, 
cf.~Subsection \ref{Sect3.3}, by augmenting its basis 
$\{x^{\bf n}\}_{{\bf n}\not={\bf 0}}$ by the element $x^{\bf 0}$.
As a consequence, we pass from $\mathsf{T}$ to $\mathbb{R}\oplus\mathsf{T}$. 
We now argue that the action of $\mathsf{G}$ naturally extends to $\mathbb{R}\oplus\mathsf{T}$,
where we first adopt the point of view of Subsection \ref{Sect5.1old}:
Indeed, given $h\in\mathbb{R}^d$ and $\{\pi^{({\bf n})}\}_{{\bf n}}\subset\mathsf{\tilde T}^*$,
which gives rise to $\Gamma\in{\rm End}(\mathsf{T})$, the extension to an
endomorphism of $\mathbb{R}\oplus\mathsf{T}$ is visualized by the block structure\footnote{the
sum in the upper right entry is effectively finite and thus
defines an element of $\mathsf{T}^*$, in line with the meaning of this block}
\begin{align}\label{ud01}
\left(\begin{array}{cc}
{\rm id}&\pi^{({\bf 0})}+\sum_{{\bf m}\not={\bf 0}}h^{\bf m}\mathsf{z}_{\bf m}\\
0&\Gamma
\end{array}\right).
\end{align}
This form of extension completes the action (\ref{pol01}) on the (extended)
polynomial sector $\mathbb{R}\oplus\mathsf{\bar T}$ in the sense of \cite[Assumption 3.20]{Hairer}: 
Since an element of $\mathsf{\tilde T}^*$,
like $\pi^{({\bf 0})}$, is characterized by vanishing on $\mathsf{\bar T}$, 
(\ref{ud01}) maps the basis element $x^{\bf n}$ 
onto $\sum_{{\bf m}}{{\bf n}\choose {\bf n}}h^{\bf m}x^{{\bf n}-{\bf m}}$,
which formally can be written as $(x+h)^{\bf n}$. 
We will motivate the presence of $\pi^{({\bf 0})}$ in
(\ref{ud01}) below.

\medskip

For (\ref{ud01}) to define an action, we need to check that the composition of
two endomorphisms of the form of (\ref{ud01}) preserves this form. This is more
easily seen for the
induced dual action on $(\mathbb{R}\oplus\mathsf{T})^*\cong\mathbb{R}\oplus\mathsf{T}^*$,
which is of the block form 
\begin{align}\label{ud02}
\left(\begin{array}{cc}
{\rm id}&0\\
\pi^{({\bf 0})}+\sum_{{\bf m}\not={\bf 0}}h^{\bf m}\mathsf{z}_{\bf m}&\Gamma^*
\end{array}\right).
\end{align}
It is now convenient to adopt the point of view of Subsection \ref{Sect5.2old},
which amounts to viewing the lower left entry of (\ref{ud02}) as a single object,
as done in (\ref{ext1}) and (\ref{poly_pi}), still labelled by $\pi^{({\bf 0})}$.
The desired statement then follows from (\ref{as33}), which was rigorously established
in part \textit{iii)} of Proposition \ref{exp01}. 

\medskip

It is also on the level of this extended
definition (\ref{ext1}) of $\pi^{({\bf 0})}$ that we may motivate (\ref{ud02}): 
The purpose of $\mathsf{G}^*$ is to contain elements $\Gamma_{xy}^*$ that ``algebrize''
the re-centering of the model, which is a $\mathsf{T}^*$-valued function\footnote{ or distribution, depending on the application}
of space-time, from its version $\Pi_y$ centered at one base point $y$ 
to its version $\Pi_x$ centered at another base point $x$, see \cite[Definition 3.3]{Hairer}. 
In the application \cite[(2.41)]{LOTT} of our setting, 
this holds only up to a space-time constant $\pi^{({\bf 0})}_{xy}\in\mathsf{T}^*$, that is,
\begin{align}\label{ud09}
\Pi_x=\Gamma^*_{xy}\Pi_y+\pi^{({\bf 0})}_{xy},
\end{align}
which however is tied to $\Gamma^*_{xy}$ via (\ref{act1}). 
Now (\ref{ud02}), with $\Gamma^*$ and $\pi^{({\bf 0})}$ specified to
$\Gamma^*_{xy}$ and $\pi^{({\bf 0})}_{xy}$, 
is made such that (\ref{ud09}) exactly assumes the form of \cite[Definition 3.1]{Hairer} 
provided we augment the model by the constant space-time function of value 1, 
which lives in the $\mathbb{R}$-component of $\mathbb{R}\oplus\mathsf{T}^*$.

\medskip

We now come to the second extension. Hairer's construction of a regularity structure
is bottom-up and combinatorial, in the sense that the index
set of the abstract model space encodes all combinations of 
integration\footnote{which in case of (\ref{as34}) involves the two kernels
$(\frac{\partial}{\partial z_2} - \frac{\partial^2}{\partial z_1^2})^{-1}$ and $(\frac{\partial}{\partial z_2} - \frac{\partial^2}{\partial z_1^2})^{-1}\frac{\partial^2}{\partial z_1^2}$; here we denote by $z$ the active variable of $\Pi_x$} 
and multiplication which are relevant for the equation. In particular, this implies that for every model component $\Pi_{x\beta}$,
also $\Pi^{-}_{x\beta}:=(\frac{\partial}{\partial z_2} - \frac{\partial^2}{\partial z_1^2})\Pi_{x\beta}$ is a component of the model.
Since in regularity structures one only cares for the equation up to polynomials,
it is natural to think of $\Pi^{-}_x$ as a $\mathsf{\tilde T}^*$-valued Schwartz distribution;
recall that $\mathsf{\tilde T}^*$ is canonically characterized as the space
of all linear functionals $\pi\in\mathsf{T}^*$ that vanish on the $x^{\bf n}$'s.
In order to capture this on the level of our abstract model space, or rather its dual,
we pass from $(\mathbb{R}\oplus\mathsf{T})^*$ to  
$(\mathbb{R}\oplus\mathsf{T})^*\oplus\mathsf{\tilde T}^*$. 

\medskip

We extend (\ref{ud02}) on this larger space $(\mathbb{R}\oplus\mathsf{T})^*
\oplus\mathsf{\tilde T}^*$ as follows
\begin{align}\label{ud03}
\left(\begin{array}{cc}
\mbox{(\ref{ud02})}&0\\
\sum_{{\bf m}}\pi^{({\bf m})}\otimes(\partial_2-\partial_1^2)^\dagger x^{\bf m}&\Gamma^*
\end{array}\right). 
\end{align}
In formulating (\ref{ud03}), 
we return to the perspective of Subsection \ref{Sect5.1old} in the sense that
$\pi^{({\bf m})}\in\mathsf{\tilde T}^*$, so that together with 
$(\partial_2-\partial_1^2)^\dagger x^{\bf m}\in\mathbb{R}\oplus\mathsf{T}$,
which canonically embeds into the bi-dual $(\mathbb{R}\oplus\mathsf{T})^{**}$,
the bottom left entry indeed defines
a linear map from $(\mathbb{R}\oplus\mathsf{T})^*$ into $\mathsf{\tilde T}^*$. Here, we extended the
definition of $\partial_i^\dagger$, see (\ref{ao51}), to $x^{\bf 0}$ in the obvious way, 
so that ${\bf m}=(0,0),(1,0)$ do not contribute, 
and the contribution of ${\bf m}=(0,1),(2,0)$ renders $x^{\bf 0}$ and $-2 x^{\bf 0}$,
respectively. The sum
is effectively finite according to the population condition \eqref{constraintPi}.
Since as a consequence of (\ref{ao61}), $\Gamma^*$ preserves $\mathsf{\tilde T}^*$,
it is legitimate as a bottom right entry.

\medskip

The motivation for the extension (\ref{ud03}) is again given by the application \cite[(2.40)]{LOTT} of our abstract structure. Indeed, the new model components transform 
according to $\Gamma_{xy}^*$ up to a $\mathsf{\tilde T}^*$-valued 
(formal) power series in space-time
\begin{align}\label{ud10}
\Pi_{x}^{-}=\Gamma_{xy}^*\Pi_y^{-}+\sum_{{\bf m}}
\big((\frac{\partial}{\partial z_2} - \frac{\partial^2}{\partial z_1^2})(\cdot-y)^{\bf m}\big)\pi^{({\bf m})}_{xy},
\end{align}
where again the coefficients are tied to $\Gamma_{xy}^*$ via (\ref{act2});
note that the terms ${\bf m}=(0,0),(1,0)$ do not contribute. We note that by
\cite[(2.23)]{LOTT}, which is in line with the axioms 
\cite[first item in (3.11)]{Hairer}, and (\ref{ao51}), we may write
$(\frac{\partial}{\partial z_2} - \frac{\partial^2}{\partial z_1^2})(\cdot-y)^{\bf m}$ $=\langle\Pi_y,(\partial_2-\partial_1^2)^\dagger x^{\bf m}\rangle$. 
Hence we see that (\ref{ud10}) assumes the axiomatic form \cite[Definition 3.3]{Hairer},
which is free of polynomial corrections,
under the extension (\ref{ud03}). 

\medskip

In order to establish that (\ref{ud03}) provides indeed a representation of $\mathsf{G}$,
we now argue that it is compatible with composition. 
By the compatibility of (\ref{ud02}), and of the bottom right block of (\ref{ud03}),
we are left with the bottom left block of the product, which consists of the summands
\begin{align}\label{ud07}
\pi^{({\bf m})}\otimes
\big(\mbox{$(\ref{ud01})'$ applied to}\;(\partial_2-\partial_1^2)^\dagger x^{\bf m}\big)
+\Gamma^*{\pi'}^{({\bf m})}\otimes (\partial_2-\partial_1^2)^\dagger x^{\bf m},
\end{align}
where $(\ref{ud01})'$ stands for (\ref{ud01}) with $\Gamma$ and $\pi^{({\bf 0})}$ replaced by
$\Gamma'$ and ${\pi'}^{({\bf 0})}$, respectively. We note that 
$p_{\bf m}:=(\partial_2-\partial_1^2)^\dagger x^{\bf m}$ is an element of the (extended) polynomial
sector $\mathbb{R}\oplus\mathsf{\bar T}$; hence applying $(\ref{ud01})'$ to it,
we obtain\footnote{to be interpreted in a formal sense} $p_{\bf m}(\cdot+h')$, 
as we have shown above.
Hence (\ref{ud07}) assumes the form
\begin{align}\label{ud05}
\pi^{({\bf m})}\otimes p_{\bf m}(\cdot+h')
+\Gamma^*{\pi'}^{({\bf m})}\otimes p_{\bf m}\quad\mbox{with}
\quad p_{\bf m}=(\partial_2-\partial_1^2)^\dagger x^{\bf m}.
\end{align}
We need to re-express (\ref{ud05}) in terms of the extended definition (\ref{ext1}) 
in order to (eventually) apply (\ref{as33}). To this purpose we introduce
\begin{align*}
{\rm id}-P=\sum_{{\bf n}\not={\bf 0}}\mathsf{z}_{\bf n}\otimes x^{\bf n},
\end{align*}
which since $\{\mathsf{z}_{\bf n}\}_{{\bf n}\not={\bf 0}}$
is dual to $\{x^{\bf n}\}_{{\bf n}\not={\bf 0}}$ 
defines a projection $\gls{projectionP}$ from $\mathsf{T}^*$ onto $\mathsf{\tilde T}^*$
that allows to pass from the extended definition (\ref{ext1}) of $\pi^{({\bf m})}$
to the original one used in (\ref{ud05}).
Clearly, the second summand in (\ref{ud05}) calls for the commutator of
$\Gamma^*$ and $P$, of which we need a representation:
Using once more that $\Gamma^*$ preserves $\mathsf{\tilde T}^*$, which can be written
as $\Gamma^*P=P(\Gamma^*-\Gamma^*({\rm id}-P))$, we obtain by (\ref{act2})
\begin{align*}
\Gamma^*P=P\big(\Gamma^*-\sum_{{\bf n}\not={\bf 0}}\pi^{({\bf n})}\otimes x^{\bf n}\big).
\end{align*}
In view of (\ref{ud05}), we apply this to ${\pi'}^{({\bf m})}$, which by 
(\ref{poly_pi}) yields
\begin{align*}
\Gamma^*P{\pi'}^{({\bf m})}=P\big(\Gamma^*{\pi'}^{({\bf m})}
-\sum_{{\bf n}>{\bf m}}\tbinom{\bf n}{\bf m}{h'}^{{\bf n}-{\bf m}}\pi^{({\bf n})}\big).
\end{align*}
Hence in terms of the extended definition (\ref{ext1}), (\ref{ud05}) assumes the form
\begin{align}
&P\Big(\pi^{({\bf m})}\otimes p_{\bf m}(\cdot+h')
+\big(\Gamma^*{\pi'}^{({\bf m})}
-\sum_{{\bf n}>{\bf m}}\tbinom{\bf n}{\bf m}{h'}^{{\bf n}-{\bf m}}
\pi^{({\bf n})}\big)\otimes p_{\bf m}\Big)\nonumber\\
&=P\Big(\big(\pi^{({\bf m})}+\Gamma^*{\pi'}^{({\bf m})}\big)\otimes p_{\bf m}\label{ud06}\\
&+\pi^{({\bf m})}\otimes(p_{\bf m}(\cdot+h')-p_{\bf m})
-\sum_{{\bf n}>{\bf m}}\tbinom{\bf n}{\bf m}{h'}^{{\bf n}-{\bf m}}
\pi^{({\bf n})}\otimes p_{\bf m}\Big).\label{ud08}
\end{align}
According to (\ref{as33}), the term in line (\ref{ud06}) is the desired output.
Hence it remains to argue that the term in (\ref{ud08}) vanishes when summed over ${\bf m}$,
which by resummation and relabelling amounts to
\begin{align*}
p_{\bf m}(\cdot+h')-p_{\bf m}
=\sum_{{\bf n}<{\bf m}}\tbinom{\bf m}{\bf n}{h'}^{{\bf m}-{\bf n}}p_{\bf n}.
\end{align*}
Recalling that $p_{\bf n}=(\partial_2-\partial_1^2)^\dagger x^{\bf n}$,
and by (\ref{ao51}), this follows from the same formula with
$p_{\bf n}$ replaced by $x^{\bf n}$, which amounts to Leibniz' rule.

\medskip

We note that $\mathsf{\tilde T}\subset\mathsf{T}$ is canonically defined as
consisting of those elements that are annihilated by 
$\{\mathsf{z}_{\bf n}\}_{{\bf n}\not={\bf 0}}\subset\mathsf{T}^*$;
it is a natural complement of $\mathsf{\bar T}$ in $\mathsf{T}$.
Moreover, $(\mathbb{R}\oplus\mathsf{T})^*\oplus\mathsf{\tilde T}^*$ is
the dual of $(\mathbb{R}\oplus\mathsf{T})\oplus\mathsf{\tilde T}$.
Passing to this primal side
$(\mathbb{R}\oplus\mathsf{T})\oplus\mathsf{\tilde T}$, (\ref{ud03}) turns into
\begin{align}\label{ud11}
{\bf \Gamma}=\left(\begin{array}{cc}
\mbox{(\ref{ud01})}
&\sum_{{\bf m}}\pi^{({\bf m})}\otimes(\partial_2-\partial_1^2)^\dagger x^{\bf m}\\
0&P^\dagger\Gamma
\end{array}\right),
\end{align}
since the bottom r.~h.~s.~of (\ref{ud03}) can be rewritten as
$\Gamma^*P$, and since $\gls{projectionPdagger}$ defined through
\begin{align}\label{ud14}
{\rm id}-P^\dagger=\sum_{{\bf n}\not={\bf 0}}x^{\bf n}\otimes \mathsf{z}_{\bf n}
\end{align}
is the dual of $P$, and a projection from $\mathsf{T}$ onto $\mathsf{\tilde T}$.
Hairer's integration map \cite[Assumption 3.21]{Hairer} has the block form
\begin{align}\label{ud12}
{\bf I}=\left(\begin{array}{cc}
0&\iota\\
0&0
\end{array}\right),
\end{align}
where $\iota$ is the injection $\mathsf{\tilde T}\subset\mathsf{T}$.
By definition (\ref{ud01}), the commutator ${\bf \Gamma I-I\,\Gamma}$ has the block form
\begin{align}\label{ud13}
{\bf \Gamma I-I\,\Gamma}=\left(\begin{array}{cc}
0&({\rm id}-P^\dagger)\Gamma\\
0&0
\end{array}\right),
\end{align}
so that by (\ref{ud14}),
the image of ${\bf \Gamma I-I\,\Gamma}$ is contained in $\mathsf{\bar T}$,
in line with the axiom \cite[(3.9)]{Hairer}.

\medskip

We close this subsection by arguing, in line with the axiom \cite[Definition 3.1]{Hairer}, 
that the matrix representation of
${\bf \Gamma}-{\rm id}$ is strictly triangular, provided we extend the basis
$\{\mathsf{z}_\beta\}_{\{[\beta]\ge 0\}\cup\{\beta=e_{\bf n}\}}$ of
$\mathsf{T}$ to a basis of $\mathbb{R}\oplus\mathsf{T}\oplus\mathsf{\tilde T}$, 
and extend the homogeneity of its index set as follows: 
On the $\mathbb{R}$-component, we take the dual basis vector to $x^{\bf 0}$, 
and endow this single index, which we (momentarily) denote by ${\bf 0}$, with homogeneity $0$;
on the $\mathsf{\tilde T}$-component, we take the basis $\{\mathsf{z}_\beta\}_{[\beta]\ge 0}$,
endowing the index $\beta$ with the homogeneity $|\beta|-2$. In particular, 
the integration map (\ref{ud12}) increases the homogeneity by $2$, 
which is the order of $\frac{\partial}{\partial z_2} - \frac{\partial^2}{\partial z_1^2}$.  
We first turn to (\ref{ud01}), where it remains to consider the upper right entry,
which is given by
\begin{align*}
{\bf \Gamma}^{\bf 0}_{\beta}=\left\{\begin{array}{cc}
\pi^{({\bf 0})}_{\beta}&\mbox{for}\;[\beta]\ge 0\\
h^{\bf m}&\mbox{for}\;\beta=e_{\bf m}\;\mbox{for some}\;{\bf m}\not={\bf 0}\end{array}\right\},
\end{align*}
and satisfies triangularity because of $|\beta|>0=|{\bf 0}|$.
We now turn to (\ref{ud11}), where once more it suffices to consider the upper
right entry, which is given by
\begin{align}\label{ud18}
{\bf \Gamma}_{\beta}^{\gamma}=\sum_{{\bf m}}
\pi^{({\bf m})}_{\beta}\langle\mathsf{z}^{\gamma},(\partial_2-\partial_1^2)^\dagger x^{\bf m}\rangle.
\end{align}
This expression vanishes unless $\gamma\in\{e_{{\bf m}-(0,1)},e_{{\bf m}-(2,0)}\}$
and thus $|\gamma|+2=|{\bf m}|$, and unless $|{\bf m}|<|\beta|$, by the population 
condition (\ref{constraintPi}). Hence (\ref{ud18}) vanishes unless $|\gamma|<|\beta|-2$.


\section{Homomorphism to the Connes-Kreimer Hopf algebra}\label{Sect5}
In this logically independent and rather combinatorial 
section, we specify to the particularly simple
case of scalar\footnote{i.~e.~in a one-dimensional state space} branched rough paths.
We will argue that our coproduct $\Delta^+$ on $\mathsf{T}^+$, when suitable restricted, 
arises from
the Connes-Kreimer coproduct. More precisely, there is a linear subspace
$\mathsf{T}_{RP}\subset\mathsf{T}$ and a pre-Lie subalgebra 
$\mathsf{L}_{RP}\subset\mathsf{L}$ (which as a linear space is isomorphic to
$\mathsf{T}_{RP}$) of derivations $D$ (which are such that $D^\dagger$ preserves $\mathsf{T}_{RP}$)
such that the corresponding restriction of $\Delta^+$ intertwines with the 
Connes-Kreimer coproduct on forests, see Subsection \ref{Sect5.1}.
The intertwining is provided by the linear one-to-one map $\phi$ that relates our
model, which is indexed by multi-indices, to branched rough paths indexed by trees, 
see Subsection \ref{Sect5.2}. 

\subsection{Relating the model $\Pi$ to branched rough paths}
\mbox{}

Since it does not affect the algebraic insight of this section,
we consider a qualitatively smooth driver $\xi$ to avoid renormalization.
Following our initial discussion in Subsection \ref{Sect1.11}, 
we consider the solution of the initial value problem
\begin{align}\label{ckh01}
\frac{du}{dx_2}=a(u)\xi,\quad u(x_2=0)=0
\end{align}
as a functional $u=u[a](x_2)$ of the (polynomial) nonlinearity $a$. It lifts
to a function of the coordinates $\{\z_k\}_{k\geq 0}$ 
introduced in (\ref{ao20bis}).
Hence we may take derivatives with respect to these coordinates evaluated at $\z_k=0$;
these partial derivatives are indexed by multi-indices\footnote{
which are maps $\mathbb{N}_0\ni k\mapsto\beta(k)\in\mathbb{N}_0$ with finitely many non-zero values} $\beta$.
It is easy to (formally) verify that the resulting partial derivatives $\Pi_\beta$ satisfy
\begin{align}\label{ckh02}
\frac{d\Pi_\beta}{dx_2}
=\sum_{k\ge 0}\sum_{e_k+\beta_1+\cdots+\beta_k=\beta}\Pi_{\beta_1}\cdots\Pi_{\beta_k}\xi,
\quad \Pi_\beta(x_2=0)=0,
\end{align}
with the understanding that $\Pi_{e_0}(x_2)=\int_0^{x_2}\xi$.
These components combine to the centered\footnote{Centered at time $x_2=0$, 
which however we suppress in our notation.} model $\gls{model}=\{\Pi_\beta\}_{\beta}$.
Incidentally, interpreting $x_2\mapsto\Pi(x_2)\in\mathbb{R}[[\z_k]]$,
(\ref{ckh02}) can be compactly written as $\frac{d\Pi}{dx_2}$
$=\sum_{k\ge 0}\z_k\Pi^k\xi$.
While this derivation of (\ref{ckh02}) is formal, $\Pi=\{\Pi_\beta\}_{\beta}$ can be, 
inductively in the length of $\beta$, constructed rigorously for sufficiently
regular $\xi$. 

\medskip

Based on (\ref{ckh02}) we may read off that not all the multi-indices are populated. 
More precisely, we claim that $\Pi_\beta\not=0$ implies 
\begin{align}\label{ckh03}
\sum_{k\ge 0}(k-1)\beta(k)=-1.
\end{align}
We establish (\ref{ckh03}) in its negated form by induction in $\sum_{k\ge 0}k\beta(k)$.
In the base case $\sum_{k\ge 0}k\beta(k)=0$, which is equivalent to $\beta\in\mathbb{N}_0e_0$, 
in which case
the r.~h.~s. of (\ref{ckh02}) reduces to $k=0$ and thus $\beta=e_0$, which satisfies (\ref{ckh03}).
Turning to the induction step, we note that the r.~h.~s. of (\ref{ckh01}) restricts to $k\ge 1$
so that the induction hypothesis can be applied to $\beta_1,\dots,\beta_k$. Hence
the induction step follows from the fact that (\ref{ckh03}) is preserved when
passing from $\beta_1,\dots,\beta_k$ to $\beta=e_k+\beta_1+\cdots+\beta_k$.

\medskip

We now compare (\ref{ckh02}) to the standard definition of branched rough paths,
which is based on (if not otherwise stated: rooted and thus non-empty and undecorated) 
trees $\gls{genericTree}$ instead of multi-indices $\beta$.
We recall that for a collection $\tau_1,\dots,\tau_k,\tau$ of such trees, the notation
\begin{align}\label{fs16}
\tau={\mathcal B}_+(\tau_1\cdots \tau_k)
\end{align}
means that $\tau$ is the tree that is obtained from 
attaching an edge to each of the trees $\tau_1,\dots,\tau_k$ and merging them in a common root,
with the understanding that ${\mathcal B}_+(\emptyset)$ gives the
tree with a single node\footnote{which is the root}, denoted by $\treeZero$. 
We recall from \cite[Section 4]{Gub10} that the branched rough path\footnote{the canonical
lift of $\xi$}
$\{\mathbb{X}_\tau\}_{\tau}$ is, inductively in the number of edges, defined through
\begin{align}\label{fs17}
\frac{d\mathbb{X}_\tau}{dx_2}=\mathbb{X}_{\tau_1}\cdots\mathbb{X}_{\tau_k}\xi,
\quad\mathbb{X}_\tau(x_2=0)=0
\quad\mbox{provided (\ref{fs16}) holds},
\end{align}
which includes $\mathbb{X}\treeZeroo (x_2)=\int_0^{x_2}\xi$. 

\medskip

It is clear from
(\ref{ckh02}) and (\ref{fs17}) that every $\Pi_\beta$ is a linear combination
of the $\mathbb{X}_\tau$'s.
\begin{lemma}\label{lemrp} For every multi-index $\beta$,
\begin{align}\label{fs18}
\Pi_\beta=\sum_{\tau\in{\mathcal T}_\beta}\frac{\sigma(\beta)}{\sigma(\tau)}\mathbb{X}_\tau.
\end{align}
\end{lemma}
Here ${\mathcal T}_\beta$ is the set of trees that have $\beta(k)$ nodes with $k$ 
children\footnote{Note that thanks to the restriction (\ref{ckh03}), this set is not empty.},
and where $\sigma(\beta)$ and $\sigma(\tau)$ are symmetry factors defined as follows:
\begin{align}\label{fs24}
\sigma(\beta):=\prod_{k\ge 0}(k!)^{\beta(k)}
\end{align}
is the size of the group of all transformations
of a tree  $\tau\in{\mathcal T}_\beta$ that are obtained by permuting the children (with their
descendants attached) at every node; $\sigma(\tau)$ is the size of the subgroup that leaves
a particular tree $\tau\in{\mathcal T}_\beta$ invariant\footnote{For a more explicit definition, see \cite[p. 430]{Brouder}.}; hence 
$\frac{\sigma(\beta)}{\sigma(\tau)}$ is the size of the orbit of $\tau$ under all
above transformations\footnote{Thus it is an integer.}.

\begin{proof}
We proceed by induction in the number of edges 
$\sum_{k\ge 0}k\beta(k)$.
In the base case when $\sum_{k\ge 0}k\beta(k)=0$,
(\ref{ckh03}) implies that $\beta=e_0$, so that ${\mathcal T}_\beta$ $=\{\treeZero\}$
and $\sigma(\beta)=\sigma(\treeZero)=1$, and hence the claim follows. For the induction
we give ourselves a $\beta$ and consider 
$\tilde\Pi_\beta:=\sum_{\tau\in{\mathcal T}_\beta}\frac{\sigma(\beta)}{\sigma(\tau)}\mathbb{X}_\tau$
so that by (\ref{fs17}) 
\begin{align}\label{fs25}
\frac{d\tilde\Pi_\beta}{dx_2}=
\sum_{\tau\in{\mathcal T}_\beta}
\frac{\sigma(\beta)}{\sigma(\tau)}\mathbb{X}_{\tau_1}\cdots\mathbb{X}_{\tau_k}\xi,
\end{align}
where $k$ and the $k$-tuple $(\tau_1,\dots,\tau_k)$ of trees
is, uniquely up to permutation, determined by (\ref{fs16}). 
More precisely, it is the multi-index $J=e_{\tau_1}+\cdots+e_{\tau_k}$ of trees. 
Note that $\tau={\mathcal B}_+(\tau_1\cdots\tau_k)$ implies 
$\sigma(\tau)=J!\sigma(\tau_1)\cdots\sigma(\tau_k)$, cf. \cite[p. 430]{Brouder}; it also yields
\begin{equation}\label{fs101}
\beta=e_k+\beta_1+\cdots+\beta_k,
\end{equation}
where $\tau_j\in{\mathcal T}_{\beta_j}$,
which in turn implies by definition \eqref{fs24}
\begin{align}\label{fs32}
\sigma(\beta)=k!\sigma(\beta_1)\cdots\sigma(\beta_k).
\end{align}
We now appeal to the re-summation in Lemma \ref{lemsum02}, which yields
\begin{align*}
\lefteqn{\sum_{\tau\in{\mathcal T}_\beta}
	\frac{\sigma(\beta)}{\sigma(\tau)}\mathbb{X}_{\tau_1}\cdots\mathbb{X}_{\tau_k}\xi}\nonumber\\
&=\sum_{k\ge 0}\sum_{e_k+\beta_1+\cdots+\beta_k=\beta}
\big(\sum_{\tau_1\in{\mathcal T}_{\beta_1}}\frac{\sigma(\beta_1)}{\sigma(\tau_1)}\mathbb{X}_{\tau_1}\big)
\cdots
\big(\sum_{\tau_k\in{\mathcal T}_{\beta_k}}\frac{\sigma(\beta_k)}{\sigma(\tau_k)}\mathbb{X}_{\tau_k}\big)
\xi.
\end{align*}
By induction hypothesis, the r.~h.~s. assumes the desired form of (\ref{ckh02}).
\end{proof}


\subsection{Relating the abstract model space $\mathsf{T}_{RP}$ to ${\mathcal B}$}\label{Sect5.2}
\mbox{}

According to the previous subsection, in our setting, 
the abstract model space $\gls{ModelSpaceRP}$ relevant for branched rough paths
is the linear sub-space of $\mathsf{T}$, see Subsection \ref{Sect3.4}, 
corresponding to the multi-indices $\beta$
that satisfy (\ref{ckh03}) and only depend on $k$ (and thus trivially satisfies $[\beta]\ge 0$).
In the standard setting, the abstract model space $\gls{modelSpaceBranched}$ relevant for branched rough
paths is the direct sum indexed by all $\tau$'s. 
Following \cite[Definition 3.3]{Hairer}, we think of the model\footnote{centered in
one point, here $x_2=0$} as a linear map
from the abstract model space into the space of distributions ${\mathcal S}'(\mathbb{R})$.
Then the relation (\ref{fs18}) between the components of the two models
defines a linear map 
$\gls{phi} \colon\mathsf{T}_{RP}\rightarrow{\mathcal B}$ such that
\begin{align*}
\Pi=\gls{modelBranched}\phi.
\end{align*}
Applying this identity to the (dual) basis vector $\z_\beta$, we read off 
from (\ref{fs18}) that $\phi$ acts as\footnote{Where we denote the basis elements of $\mathcal{B}$ by $\z_\tau$.}
\begin{align}\label{fs30}
\phi\z_\beta=\sum_{\tau\in{\mathcal T}_\beta}
\frac{\sigma(\beta)}{\sigma(\tau)}\z_\tau,
\end{align}
from which we learn that $\phi$ is one-to-one (but not onto).
Hence the matrix representation of $\phi$ is given by 
\begin{align}\label{fs18bis}
\phi^\beta_\tau=
\left\{\begin{array}{cl}\frac{\sigma(\beta)}{\sigma(\tau)}&\mbox{if}\;\tau\in{\mathcal T}_\beta\\
0&\mbox{otherwise}\end{array}\right\}.
\end{align}


\subsection{Relating the pre-Lie algebra $\mathsf{L}_{RP}$ to $\mathcal{L}^1$}
\label{Sect5.4}\mbox{}

In our setting, the subspace $\gls{LieAlgebraRP}$ of $\mathsf{L}$ relevant for branched rough paths
is spanned by $\z^{\bar\gamma} D^{(\0)}\subset{\rm Der}(\mathbb{R}[[\z_k]])$, with
multi-indices ${\bar\gamma}$ only depending on $k$ and satisfying (\ref{ckh03}). 
As opposed to $\mathsf{L}$, $\mathsf{L}_{RP}$ is closed under the pre-Lie product $\prelie$
on ${\rm Der}(\mathbb{R}[[\z_k]])$. Indeed, this follows from fact that
by (\ref{ao24}), 
$\z^{\bar\gamma} (D^{(\0)}\z^\gamma)$ 
is a linear combination of $\z^{\beta}$'s
with $\beta+e_k=\bar\gamma+\gamma+e_{k+1}$ for some $k\ge 0$, so that (\ref{ckh03}) is preserved.
Moreover, for $D\in\mathsf{L}_{RP}$ we have that $D^\dagger$ preserves $\mathsf{T}_{RP}$.
Indeed, this follows from \eqref{ao24} and \eqref{ao58} by the same reasoning.
Finally, there is a canonical (non-degenerate) pairing between $\mathsf{T}_{RP}$ and
$\mathsf{L}_{RP}$, which are isomorphic as linear spaces, defined through
\begin{equation}\label{fs102}
\langle\z_\gamma,\z^{\bar\gamma}D^{(\0)}\rangle
=\delta_{\gamma}^{\bar\gamma}.
\end{equation}

\medskip

On the classical side, we consider the (pre-)Lie algebra $\gls{GLpreLie}$ introduced by Connes and Kreimer \cite{ConnesKreimer},
which as a linear space is the direct sum indexed by trees $\tau$ and thus isomorphic
to ${\mathcal B}$; we denote by $\gls{basisGL}$ the standard basis. It comes
with the pre-Lie product 
\begin{align}\label{fs27}
\Z_{\tau_1} \gls{preLieGL} \Z_{\tau_2}=\sum_{\tau}n(\tau_1,\tau_2;\tau)\Z_{\tau},
\end{align}
where $n(\tau_1,\tau_2;\tau)$ is the number of single cuts of $\tau$ such that the branch
is $\tau_1$ and the trunk is $\tau_2$. In other words, $n(\tau_1,\tau_2;\tau)$
is the number of edges of $\tau$ that when removed yields
the trees $\tau_1$ and $\tau_2$, where the second one is defined to be the 
one containing the root of $\tau$. The fact that (\ref{fs27}) satisfies the axiom of
a pre-Lie product is established in \cite[(103)]{ConnesKreimer}. We note that
the pre-Lie product can also be recovered from
\begin{align}\label{fs20}
(\sigma(\tau_1)\Z_{\tau_1})\leadsto (\sigma(\tau_2)\Z_{\tau_2})
=\sum_{\tau}m(\tau_1,\tau_2;\tau)\sigma(\tau)\Z_{\tau},
\end{align}
where $m(\tau_1,\tau_2;\tau)$ is the number of nodes of $\tau_2$ to which
$\tau_1$ may be attached to yield $\tau$. Passing from (\ref{fs27}) to (\ref{fs20})
relies on the combinatorial identity\footnote{The notation in \cite[Proposition 4.3]{Hoffman}
is opposite to ours, which is the one of \cite{ConnesKreimer}.}
%
$n(\tau_1,\tau_2;\tau)\sigma(\tau_1)\sigma(\tau_2)$ $=m(\tau_1,\tau_2;\tau)\sigma(\tau)$,
%
which is established in \cite[Proposition 4.3]{Hoffman}. We note that it is the product
\begin{equation}\label{fs103}
\tau_1\gls{grafting}\tau_2:=\sum_{\tau}m(\tau_1,\tau_2;\tau)\tau
\end{equation}
that comes with the intuition of grafting the tree $\tau_1$ onto the tree $\tau_2$; 
it can be extended
to a pre-Lie product on ${\mathcal B}$ by linearity. While (\ref{fs20}) shows that
the two pre-Lie structures on ${\mathcal L}^1$ and on ${\mathcal B}$ are isomorphic
it is helpful to distinguish them here. This is related to the fact that
we consider the standard pairing between ${\mathcal B}$ and ${\mathcal L}^1$,
i.~e.~we think of $\z_\tau\in{\mathcal B}$ and 
$\Z_\tau\in{\mathcal L}^1$ as dual bases\footnote{Alternatively,
one could work with $\curvearrowright$ but impose the pairing $\Z_\tau.\z_{\tau'}$
$=\sigma(\tau)\delta_{\tau}^{\tau'}$, 
see more in \cite[Subsection 3.3]{BCCH19} on the choice of pairings
viz.~inner products.}. These pre-Lie products have been evoked in branched rough paths
\cite[Subsection 3.2.2]{BCFP19} and in regularity structures \cite[Remark 4.1]{BCCH19}.

\medskip

In view of the obvious finiteness properties of $\phi$, see (\ref{fs18bis}),
the two above-mentioned non-degenerate pairings define a linear map
$\phi^\dagger\colon{\mathcal L}^1\rightarrow\mathsf{L}_{RP}$ by duality. From (\ref{fs18bis})
we learn that it acts as
\begin{align}\label{fs21}
\phi^\dagger\sigma(\tau)\Z_\tau=\sigma(\beta)\z^\beta D^{(\0)}
\quad\mbox{provided}\;\tau\in{\mathcal T}_\beta.
\end{align}
\begin{lemma}\label{lembrp}
	$\phi^\dagger$ is a pre-Lie algebra morphism:
\begin{align}\label{fs19}
\phi^\dagger(\Z_{\tau_1}\leadsto\Z_{\tau_2})
=(\phi^\dagger\Z_{\tau_1})\prelie(\phi^\dagger\Z_{\tau_2}).
\end{align}
\end{lemma}
\begin{proof}
	Multiplying (\ref{fs19}) by $\sigma(\tau_1)\sigma(\tau_2)$, appealing to
(\ref{fs20}) for the l.~h.~s. and then to (\ref{fs21}), we see that (\ref{fs19}) follows from
\begin{align}\label{fs22}
\sum_\tau m(\tau_1,\tau_2;\tau)\sigma(\beta)\z^\beta
=\sigma(\beta_1)\sigma(\beta_2)\z^{\beta_1}(D^{(\0)}\z^{\beta_2}),
\end{align}
with the understanding that $\beta$, $\beta_1$ and $\beta_2$ are determined by $\tau\in{\mathcal T}_{\beta}$, $\tau_1\in{\mathcal T}_{\beta_1}$
and $\tau_2\in{\mathcal T}_{\beta_2}$. We note that $m(\tau_1,\tau_2;\tau)\not =0$ only if
the fertilities are related by
$\beta_1+\beta_2+e_{k+1}=\beta+e_{k}$ for some $k\ge 0$, which amounts to taking
a node of $\tau_2$ with $k$ children and attaching $\tau_1$ to it (via a new edge).
Hence the l.~h.~s. naturally decomposes into a sum over $k\ge 0$ -- and
so does the r.~h.~s. in view of (\ref{ao35}). Hence (\ref{fs22}) follows from
\begin{align*}
\sum_{\tau:\beta_1+\beta_2+e_{k+1}=\beta+e_{k}}m(\tau_1,\tau_2;\tau)\sigma(\beta)\z^\beta
=(k+1)\sigma(\beta_1)\sigma(\beta_2)\z^{\beta_1}\z_{k+1}
\partial_{\z_k}\z^{\beta_2},
\end{align*}
which, upon multiplication by $\z_k$, reduces to the purely combinatorial
\begin{align*}
\sum_{\tau:\beta_1+\beta_2+e_k=\beta+e_{k+1}}m(\tau_1,\tau_2;\tau)\sigma(\beta)
=(k+1)\sigma(\beta_1)\sigma(\beta_2)\beta_2(k).
\end{align*}
According to the definition (\ref{fs24}) of $\sigma(\beta)$ this reduces to
\begin{align*}
\sum_{\tau:\beta_1+\beta_2+e_k=\beta+e_{k+1}}m(\tau_1,\tau_2;\tau)
=\beta_2(k).
\end{align*}
This last identity holds because also the l.~h.~s. is the number of nodes of $\tau_2$
with $k$ children on which $\tau_1$ can be attached (via a new edge).
\end{proof}

The fact that $\phi^\dagger$ is not one-to-one reflects that our (pre-)Lie algebra is not free, as opposed to $\mathcal{L}^1$ (cf. \cite[Theorem 1.9]{Chapoton}). Moreover, $\mathsf{L}_{RP}$ is isomorphic to $\mathcal{L}^1$ quotiented by an ideal.

\begin{corollary} 
As pre-Lie algebras, $(\mathsf{L}_{RP},  \prelie)$ and $(\mathcal{L}^1 / \mathcal{R}, \leadsto)$ are isomorphic, where
	\begin{equation}\label{fs104}
	\mathcal{R} := \mbox{\rm span} \{ \sigma(\tau_1) Z_{\tau_1} - \sigma(\tau_2)Z_{\tau_2}\,|\, \tau_1,\tau_2\in \mathcal{T}_\beta \}_\beta.
	\end{equation}
\end{corollary} 
\begin{proof}
	By Lemma \ref{lembrp}, $\phi^\dagger:\mathcal{L}^1\to\mathsf{L}_{RP}$ is a morphism. Thanks to the restriction \eqref{ckh03}, it is onto. It only remains to show that $\mathrm{ker}\phi^\dagger = \mathcal{R}$. By \eqref{fs21}, $\phi^\dagger$ vanishes on the generating set \eqref{fs104}, hence $\mathcal{R}\subset \mathrm{ker}\phi^\dagger$.
	To show the opposite inclusion, we fix $\sum_\tau c_\tau Z_\tau \in \mathrm{ker}\phi^\dagger$, and by linearity it is enough to show that $\sum_{\tau\in\mathcal{T}_\beta} c_\tau Z_\tau \in \mathcal{R} $ for every $\beta$. 
	By the representation \eqref{fs21}, $\sum_\tau c_\tau Z_\tau$ is in the kernel of $\phi^\dagger$ if and only if $\sum_{\tau\in\mathcal{T}_\beta} \frac{c_\tau}{\sigma(\tau)}=0$ for all $\beta$.	
\end{proof}
\begin{remark}\label{rem01}
	Note that, although freeness is lost, the generation property is preserved; in this setting, $\z_0 D^{(\0)}$ is the generator of $\mathsf{L}_{RP}$.
\end{remark}

\medskip

\subsection{Relating the coproduct $\Delta^+_{RP}$ to Butcher's}\label{Sect5.1}
\mbox{}

The pre-Lie algebra morphism property (\ref{fs19}) of $\phi^\dagger$ obviously implies
that it is also a Lie-algebra morphism between ${\mathcal L}^1$
and $\mathsf{L}_{RP}$. By the characterizing property of universal
envelopes, $\phi^\dagger$ lifts to a morphism between the Hopf algebras ${\rm U}({\mathcal L}^1)$ 
and ${\rm U}(\mathsf{L}_{RP})$. According to \cite[Theorem 3 b)]{ConnesKreimer} the
standard pairing \cite[(105)]{ConnesKreimer} between the 
Hopf algebra ${\rm U}({\mathcal L}^1)$ and the
Connes-Kreimer Hopf algebra ${\mathcal H}$ respects the Hopf algebra structures.
We recall that as an algebra, ${\mathcal H}$ is the polynomial
algebra $\mathbb{R}[\tau]$ over trees $\tau$, and the coproduct $\gls{coproductButcher}$ 
is defined according to Butcher via cutting-off sub-trees (``pruning''),
see e.~g.~\cite[Section 3]{Brouder}.
Defining $\gls{TplusRP}$ and $\gls{coproductRP}$, based on the Lie algebra
$\mathsf{L}_{RP}$, in analogy to
$\mathsf{T}^+$ and $\Delta^+$, see Subsection \ref{Sect4.3}, 
we thus obtain that $\phi:\mathsf{T}^+_{RP}\to\mathcal{H}$ is a Hopf algebra morphism, in particular
\begin{align*}
(\phi\otimes\phi)\Delta^+_{RP}=\Delta_B \phi.
\end{align*}
Here, we used that on $\T^+_{RP}\otimes\T^+_{RP}$, which is naturally a subspace of $(\mathrm{U}(\mathsf{L}_{RP})\otimes\mathrm{U}(\mathsf{L}_{RP}))^*$, we have that $(\phi^\dagger\otimes\phi^\dagger)^*=\phi\otimes\phi$.

\medskip

\subsection{Relating $\phi^\dagger$ to $\Upsilon$}\label{Sect5.5}
\mbox{}

The morphism property (\ref{fs19}) is closely related to the ones that appear
in regularity structures \cite[Corollary 4.15]{BCCH19} and branched rough paths
\cite[Lemma 3.7]{BonnefoiCMW},
as we shall explain in this subsection for the latter: 
In view of its canonical pairing \eqref{fs102} with $\mathsf{T}_{RP}$, $\mathsf{L}_{RP}$ can be canonically identified with a subspace\footnote{
	Both have the same
	index set, but while $\mathsf{L}_{RP}$ is a direct sum, $\mathsf{T}_{RP}^*$
	is a direct product.} 
of $\mathsf{T}_{RP}^*$, so that
we may think of $\phi^\dagger$ as mapping into $\mathsf{T}_{RP}^*$
and then interpret (\ref{fs19}) as the following identity in 
$\mathsf{T}_{RP}^*\subset\mathbb{R}[[\z_k]]$:
\begin{align}\label{fs28}
\phi^\dagger(\Z_{\tau_1}\leadsto\Z_{\tau_2})
=(\phi^\dagger\Z_{\tau_1})(D^{(\0)}\phi^\dagger\Z_{\tau_2}).
\end{align}
We note that the image of $\phi^\dagger$ is actually contained in the polynomial
subspace $\mathbb{R}[\z_k]$, and thus in view of (\ref{ao20bis}) 
in the space of functions on $a$-space. Hence we
may apply $\phi^\dagger\Z_{\tau}$ to a polynomial\footnote{even to a formal
	power series} $a$, and thus also to $a(\cdot+u)$ for some shift $u\in\mathbb{R}$.
We also note that $D^{(\0)}$ preserves $\mathbb{R}[\z_k]$,
see (\ref{ao35}). Hence we may ``test'' (\ref{fs28}) with $a(\cdot+u)$ and obtain
by definition (\ref{ao21}) of $D^{(\0)}$
\begin{align}\label{fs29}
\phi^\dagger(\Z_{\tau_1}\leadsto\Z_{\tau_2})[a(\cdot+u)]
=(\phi^\dagger\Z_{\tau_1})[a(\cdot+u)]
\big(\frac{d}{du}\phi^\dagger\Z_{\tau_2}[a(\cdot+u)]\big).
\end{align}
With the abbreviation
\begin{align}\label{fs31}
\gls{upsilon}^a[\tau](u):=\phi^\dagger\sigma(\tau)\Z_\tau[a(\cdot+u)]
\end{align}
and the help of (\ref{fs20}) and \eqref{fs103},
(\ref{fs29}) turns into the following simple\footnote{Simple because our scalar setting
	does not require node decorations, and since we did not extend to forests.} 
version of \cite[Lemma 3.4]{BonnefoiCMW}
\begin{align*}
\Upsilon^a[\tau_1\curvearrowright\tau_2]=\Upsilon^a[\tau_1](\frac{d}{du}\Upsilon^a[\tau_2]),
\end{align*}
which states that for fixed $a$, 
$\Upsilon^a$ is a pre-Lie algebra morphism from $({\mathcal B},\curvearrowright)$ into
the pre-Lie algebra of functions of $u\in\mathbb{R}$.

\medskip

We remark that the object (\ref{fs31}) coincides with the one recursively defined in
\cite[Definition 2.13]{BonnefoiCMW}. This follows from the fact that under the
assumption (\ref{fs16}), we learn that by (\ref{fs21}) and \eqref{fs101}, (\ref{fs32}) translate into
the following identity in $\mathsf{T}_{RP}^*\subset\mathbb{R}[[\z_k]]$
\begin{align*}
\phi^\dagger\sigma(\tau)\Z_\tau=k!\z_k
(\phi^\dagger\sigma(\tau_1)\Z_{\tau_1})\cdots
(\phi^\dagger\sigma(\tau_k)\Z_{\tau_k}).
\end{align*}
This identity, when tested with $a(\cdot+u)$, by definitions (\ref{ao20bis}) 
and (\ref{fs31}), turns into
\begin{align*}
\Upsilon^a[\tau]=(\frac{d^ka}{du^k})
\Upsilon^a[\tau_1]\cdots\Upsilon^a[\tau_k],
\end{align*}
which coincides with the induction \cite[(2.11)]{BonnefoiCMW}. 
The base case is obvious: For $\tau=\treeZero$
we learn from (\ref{fs21}) that $\phi^\dagger\sigma(\tau)\Z_\tau$ $=\z_0$,
and from (\ref{ao20bis}) that $\z_0[a(\cdot+u)]$ $=a(u)$. 

\medskip

\subsection{Renormalization of rough paths via multi-indices}\label{Sect6.6}
\mbox{}

We now give some details on future directions of our research, namely that renormalization can be carried out within the multi-index description without passing via trees. From the analytic and stochastic viewpoint, this is carried out in the case of quasi-linear SPDEs in the work \cite{LOTT}. In this section, we reveal the algebraic structure that guides renormalization in the simple case of branched rough paths, in line with \cite{BCFP19}.

\medskip

Let us consider the following generalized version of \eqref{ckh01}: 
\begin{equation}\label{tra04}
\frac{du}{dx_2} = a_0(u) + a_1(u)\xi, \quad u(x_2 = 0)=0.
\end{equation}
Following Subsection \ref{Sect1.11}, we see the solution $u$ as a function of both nonlinearities, $u = u[a_0,a_1]$. For simplicity, we will only study transformations in $a_0$ by addition of a function of $a_1$, i.~e.\footnote{ 
	This class includes the It\^o-Stratonovich conversion in SDEs: Assume that $\xi$ is the time derivative of a standard Brownian motion and take $c = \frac{1}{2} \z_0^1 \z_1^1$, so that $c[a_1](u) = \frac{1}{2} a_1(u) a'_1(u)$. Then \eqref{tra04} transforms into
	\begin{equation*}
	\frac{du}{dx_2} = a_0(u) + a_1(u)\xi + \frac{1}{2} a_1(u) a'_1(u).
	\end{equation*}}
\begin{equation}\label{tra01}
(a_0,a_1) \mapsto (a_0 + c[a_1], a_1);
\end{equation}
here,
\begin{equation*}
c[a_1] (u) := c[a_1(\cdot + u)],
\end{equation*}
so that shift-covariance is built in. As in Section \ref{Sect1}, we lift \eqref{tra01} to the space of functions of $(a_0,a_1)$ by means of an algebra morphism $\gls{rmap}$:
\begin{equation}\label{tra02}
M_c \pi [a_0,a_1] = \pi [a_0 + c[a_1],a_1].
\end{equation}
In analogy with \eqref{ao20bis}, we introduce coordinates on $(a_0,a_1)$-space
\begin{equation}\label{tra03}
\z_k^0 [a_0,a_1] = \frac{1}{k!} \frac{d^k a_0}{dv^k}(0),\;\;\; \z_k^1 [a_0,a_1] = \frac{1}{k!} \frac{d^k a_1}{dv^k}(0),
\end{equation}
and multi-indices which now involve both families of variables, i.~e.
\begin{equation*}
\z^\beta := \prod_{k\ge 0} (\z_k^0)^{\beta(0,k)}\, (\z_k^1)^{\beta(1,k)}.
\end{equation*}
We still denote by $\T^*_{RP}$ the dual of the model space, now characterized by the population condition 
\begin{equation}\label{ckh03bis}
\sum_{k \geq 0}(k+1)\big(\beta(0,k) + \beta(1,k)\big) = -1;
\end{equation}
the argument is the same as that of \eqref{ckh03}. 

\medskip

We moreover define the infinitesimal generator of tilt by a constant in line with \eqref{ao21}:
\begin{equation}
D^{(\0)}\pi [a_0,a_1] = \frac{d}{dv}_{|v=0}\pi[a_0(\cdot + v), a_1(\cdot + v)].
\end{equation}
Then \eqref{tra02} acts on the coordinate functionals \eqref{tra03} as
\begin{equation}\label{tra07}
M_c \z_k^0 = \z_k^0 + \frac{1}{k!}(D^{(\0)})^kc,\;\;\; M_c \z_k^1 = \z_k^1,
\end{equation}
which thanks to the algebra morphism property defines the action of $M_c$ on the polynomial algebra $\R[\z_k^0,\z_k^1]$. Since $M_c$ amounts to plugging a power series into a power series, cf. \eqref{tra02}, this action extends to the full power series space $\R[[\z_k,\z_\n]]$ as long as we impose
\begin{equation}\label{tra11}
c_{\beta = 0} = 0.
\end{equation}
The following result gathers the properties of the map $M_c$.
\begin{lemma}\label{lem6.5}
	Let $c\in \T_{RP}^*\, \cap \, \R[[\z_k^1]]$. Then $M_c \in \textnormal{End}(\T_{RP}^*)$. In addition, for all $\pi_1,\pi_2 \in \T_{RP}^*$,
	\begin{equation}\label{tra15}
	M_c ((\pi_1 D^{(\0)}) \pi_2) = (M_c \pi_1) D^{(\0)} (M_c \pi_2),
	\end{equation}
	which implies that $M_c$ is a pre-Lie morphism in $\mathsf{L}_{RP}\subset \T_{RP}^*$. Moreover, the following composition rule holds:
	\begin{equation}\label{tra16}
	M_{c_1}M_{c_2} = M_{c_1 + c_2},
	\end{equation}
	which yields an Abelian group structure.
\end{lemma}
\begin{proof}
	Since the condition \eqref{tra11} is contained in \eqref{ckh03bis}, $M_c$ is well-defined. We first argue that $M_c$ commutes with $D^{(\0)}$; indeed, this follows from \eqref{ao35} and \eqref{tra07} via
	\begin{equation*}
	D^{(\0)}M_c \z_k^0 = (k+1)\z_{k+1}^0 + \frac{1}{k!}(D^{(\0)})^{k+1} c = (k+1) M_c \z_{k+1}^0 = M_c D^{(\0)} \z_k^0,
	\end{equation*}
	and is tautological for $\z_k^1$. It then extends to the general case by the algebra morphism property of $M_c$ and Leibniz rule for $D^{(\0)}$. As a consequence, \eqref{tra15} is satisfied:
	\begin{equation*}
	M_c ((\pi_1 D^{(\0)}) \pi_2) = (M_c \pi_1) M_c(D^{(\0)}\pi_2) = (M_c \pi_1) D^{(\0)} (M_c \pi_2).
	\end{equation*}
	We now argue that $M_c\T_{RP}^*\subset \T_{RP}^*$. It is enough to show it for the space of finite sums, which is isomorphic to $\mathsf{L}_{RP}$\footnote{ See the discussion at the beginning of Subsection \ref{Sect5.5}.}. Since \eqref{tra15} implies that $M_c$ is a pre-Lie morphism, and since the elements $\z_0^0,\z_0^1$ generate $\mathsf{L}_{RP}$, cf. Remark \ref{rem01}, it suffices to show $M_c \z_0^0, M_c \z_0^1 \in \T_{RP}^*$; this follows from \eqref{tra07} under the assumption $c\in \T_{RP}^*$. Finally, the composition rule \eqref{tra16} may be read off from \eqref{tra02}. 
\end{proof}

Combining \eqref{tra07} and \eqref{tra15}, the map $M_c$ may be seen as a shift of the form $\z_0^0 \mapsto \z_0^0 + c$ which, in addition, is a pre-Lie morphism. This connects the approach to the translation of (branched) rough paths as described in \cite[Definition 14]{BCFP19}. In the specific setting of this section, given an element\footnote{
The general setting of \cite{BCFP19} allows $v$ to be a Lie series in the Grossman-Larson Hopf algebra generated by two nodes distinguished by decorations corresponding to each nonlinearity (in the specific setting of transformations of the form \eqref{tra01} only one decoration matters, so it is legitimate to think of $v$ as non-decorated). Although we can also work with (infinite) Lie series, for notational convenience we restrict to finite sums and write $v \in \mathcal{L}^1$.} $v\in \mathcal{L}^1$, its associated translation map, which we denote by $\gls{rmapBCFP}$, is defined as the unique pre-Lie morphism that extends
\begin{equation}\label{tra09}
M_v^{BCFP} Z \treeThirty{0} = Z {\treeThirty{0}} +v,\;\;\; M_v^{BCFP} Z {\treeThirty{1}} = Z {\treeThirty{1}}.
\end{equation}

We follow the construction of previous subsections and build a dictionary $\phi$ from our model space $\T_{RP}$ to the linear space of trees decorated by $0$ and $1$, so that \eqref{fs30} still holds with $\mathcal{T}_\beta$ given by the set of trees which contain $\beta(i,k)$ nodes decorated by $i$ and with $k$ children. This dictionary\footnote{ This time we consider $\phi^\dagger: \mathcal{L}^1 \to \T_{RP}^*$ instead of $\phi^\dagger: \mathcal{L}^1 \to \L_{RP}$, in line with Subsection \ref{Sect5.5}.} intertwines with the translation maps:

\begin{lemma}
	\begin{equation}\label{tra08}
	\phi^\dagger M_v^{BCFP} = M_{\phi^\dagger v} \phi^\dagger.
	\end{equation}
\end{lemma}
\begin{proof}
	Since both $\phi^\dagger M_v^{BCFP}$ and $M_{\phi^\dagger v} \phi^\dagger$ are pre-Lie morphisms, cf. \eqref{fs28} and \eqref{tra15}, it is enough to show  \eqref{tra08} for the generators $Z {\treeThirty{0}}$ and $Z {\treeThirty{1}}$. The case $Z {\treeThirty{0}}$ follows from
	\begin{equation*}
	\phi^\dagger M_v^{BCFP} Z {\treeThirty{0}} = \phi^\dagger (Z {\treeThirty{0}} +v) = \z_0^0 + \phi^\dagger v = M_{\phi^\dagger v} \z_0^0 =M_{\phi^\dagger v} \phi^\dagger Z {\treeThirty{0}}.
	\end{equation*}
	The case $Z {\treeThirty{1}}$ is trivial from \eqref{tra09}.
\end{proof}

\medskip

The map $M_c$ induces a (purely algebraic) transformation of the model. More precisely, if $\Pi$ is the model constructed inductively from the ODE in $\R[[\z_k^0,\z_k^1]]$
\begin{equation*}
\frac{d}{dx_2} \Pi = \sum_{k\geq 0} \z_k^0 \Pi^k + \sum_{k \geq 0} \z_k^1 \Pi^k \xi,
\end{equation*}
which arises from \eqref{tra04}, then thanks to the morphism property and \eqref{tra07}, $\tilde{\Pi} := M_c \Pi$ solves
\begin{equation}\label{tra05}
\frac{d}{dx_2} \tilde{\Pi} = \sum_{k\geq 0} \z_k^0 \tilde{\Pi}^k + \sum_{k \geq 0} \z_k^1 \tilde{\Pi}^k \xi + \sum_{k \geq 0} \frac{1}{k!} \tilde{\Pi}^k(D^{(\0)})^k c.
\end{equation}
The form of the last r.~h.~s. term in \eqref{tra05} connects to the form of the counter-term in the model equations for quasi-linear SPDEs which was postulated in \cite[Subsection 1.1]{OSSW} and systematically constructed in \cite{LOTT}.

\medskip

Finally, we address the transformation of the equation \eqref{tra04}; under the action \eqref{tra01}, it tautologically turns into
\begin{equation}\label{tra17}
\frac{du}{dx_2} = a_0(u) + a_1(u)\xi + c[a_1(\cdot + u)].
\end{equation}
We now argue that any translation in the sense of \cite{BCFP19}, which is expressed in terms of trees, may be expressed in terms of multi-indices in the form of \eqref{tra17}. For this purpose, we note that, in the notation of \cite[p. 37]{BCFP19},
\begin{equation*}
a_{v} (u) = \phi^\dagger v [a_0(\cdot + u), a_1(\cdot + u)].
\end{equation*}
Indeed, this follows from \eqref{tra03} for $v= Z \treeThirty{0}, Z \treeThirty{1}$, and is extended by the morphism property \eqref{fs29}. Therefore by \cite[Theorem 38 (ii)]{BCFP19} equation \eqref{tra04} assumes the form
\begin{equation}\label{tra18}
\frac{du}{dx_2} = a_0(u) + a_1(u)\xi + \phi^\dagger v[a_1(\cdot + u)].
\end{equation}

\medskip

Comparing \eqref{tra17} and \eqref{tra18}, we see that the greedier setting of multi-indices loses no information with respect to the tree-based approach; on the contrary, it reduces the complexity by grouping trees which give rise to the same renormalization procedure into a single multi-index. We expect this to extend to SPDEs, reducing the size of the renormalization group. 
\section{Homomorphism to the SHE structure}\label{Sect6}

In the spirit of the previous section we show that our algebraic structure is compatible with the ones in regularity structures when it comes to semi-linear SPDEs. More precisely, in Subsections \ref{Sect7.1} to \ref{Sect7.3} we connect our model space and pre-Lie structure to \cite[Subsection 4.1]{BCCH19}; in Subsection \ref{Sect7.4}, we show compatibility of our Hopf algebra $\T^+$ with the one in \cite[Subsection 4.2]{Hairer}. We will establish this in the specific case of the stochastic heat equation \eqref{SHE}.
In this specific case, the model is defined through the hierarchy of linear PDEs\footnote{ 
Since in our setting we work not with kernels but directly with the PDE, in order to guarantee uniqueness of the model we need to impose some extra conditions. Such conditions, which are irrelevant in the algebraic context of this article, amount to growth bounds that in turn allow for the application of Liouville principles; cf. e.~g. \cite[Proposition 5.3]{LOTT} or \cite[Lemma 4.9]{LO}.}
\begin{align}\label{gpam23}
\big(\frac{\partial }{\partial x_2}
- \frac{\partial^2 }{\partial x_1^2}\big)\Pi_\beta
=\sum_{k\ge 0}\sum_{e_k+\beta_1+\cdots+\beta_k=\beta}\Pi_{\beta_1}\cdots\Pi_{\beta_k}\xi,
\end{align}
for $[\beta]\geq0$,
together with $\Pi_{e_\n} (x) = x^\n$. 
This recursive definition leads to the following population condition: $\Pi_{\beta}\neq 0$ implies
\begin{equation}\label{gpam22}
\sum_{k\geq 0} (k-1)\beta (k) - \sum_{{\bf n\neq \0}} \beta(\n) = -1.
\end{equation}

\medskip

\subsection{Relating the abstract model space $\T$ to $\mathscr{B}$}\label{Sect7.1}
\mbox{}

We first fix the dictionary between our model space\footnote{
The reason why we do not restrict the model space $\T$ to the populated subspace, as we did in the rough path case, is that even then, our dictionary $\phi$ will no longer be one-to-one, see Subsection \ref{Sect7.4}.}
$\T$ and the space $\gls{BspaceHairer}$ of linear combinations of trees $\tau$ with expanded polynomial decorations\footnote{
As opposed to the space $\mathscr{V}$ of combinatorial decorated trees \cite[(3.5)]{BCCH19}, which have simpler polynomial decorations, and contain the actual model space in the tree-based framework; see Subsection \ref{Sect7.4}.}
defined in \cite[Subsection 4.1]{BCCH19}. 
In the specific case of SHE, the set of trees \cite[(4.3)]{BCCH19} is inductively defined through\footnote{
In order to connect to \cite{Hairer}, we switch the notation in \cite{BCCH19} and use $\mathcal{I}$ for the abstract integration and $\mathscr{I}$ for the polynomial labelling. 
Here, $\mathscr{I}$ serves as a reminder that we cannot multiply polynomials. It is understood that \eqref{gpam01} is independent of the order.}
\begin{equation}\label{gpam01}
\tau = \treeZero (\prod_{i\in\mathsf{I}} \mathscr{I} X^{\n_i}) (\prod_{j\in\mathsf{J}}\mathcal{I}\tau_j)
\end{equation} 
and its integrated version $\mathcal{I} \tau$, 
where $\gls{intHairer}$ is the placeholder for integration, i.~e. application of the kernel of the solution operator, $\treeZero$ is the placeholder for the noise\footnote{
denoted in \cite{BCCH19} as $\Xi$} 
and with the understanding that $\n_i \neq \0$ for all $i\in\mathsf{I}$. 
Although the factor $\treeZero (\prod_{i\in\mathsf{I}} \mathscr{I} X^{\n_i})$ is considered in \cite{BCCH19} as a node decoration, we will instead think of $\mathscr{I} X^{\n_i}$ as a subtree; 
for example, the graphical representation of $\treeZero (\mathscr{I} X^{\mathbf{n}_1} ) ( \mathscr{I} X^{\mathbf{n}_2}) \, \mathcal{I}\treeZero$ is given by
\begin{equation}\label{gpam05}
\treeNineteen.
\end{equation}
The symmetry factor is defined accordingly, so that, for example,
\begin{equation*}
\sigma ( \treeTwenty ) = 2.
\end{equation*}
For later purpose, we recursively define $N(\treeZero)=1$ and 
\begin{equation}\label{gpam16}
N(\tau) = \prod_{i\in\mathsf{I}}\n_i ! \prod_{j\in\mathsf{J}} N(\tau_j),
\end{equation}
so that $N(\tau)$ is the product of the factorials of all polynomial decorations of $\tau$.

\medskip

Let us now define the two linear maps\footnote{ Following \cite{BCCH19}, we adopt the notation $\circ$ to refer to the setting of expanded polynomial decorations; we will drop it in the contracted setting of Subsection \ref{Sect7.4}.} 
$\gls{phiomin}:\mathsf{\tilde T} \to \mathscr{B}$ and $\gls{phio}: \T \to \mathscr{B}$ by $\mathring{\phi}_- \z_{\beta = 0} = 0$, and then recursively in the length of $\beta$ by
\begin{equation}\label{gpam02}
\mathring{\phi}_- \z_\beta = \sum_{k\geq 0}\sum_{e_k+\beta_1+\cdots+\beta_k=\beta}  \mathring{\phi} \z_{\beta_1}\cdots \mathring{\phi}\z_{\beta_k} \treeZero
\end{equation}
and
\begin{equation}\label{gpam03}
\mathring{\phi} \z_\beta = \left\{
\begin{array}{cl}
\mathscr{I} X^\n & \mbox{if}\;\beta = e_\n\\
\mathcal{I}\mathring{\phi}_- \z_\beta & \mbox{otherwise}
\end{array}
\right\}.
\end{equation}
As pointed out in Subsection \ref{Sect.ext}, the reason why we need both maps $\mathring{\phi}_-$ and $\mathring{\phi}$ is that the traditional model space in regularity structures relates to ours via $\mathsf{\tilde T}\oplus \R\oplus \T$; 
in trees, this means that every non-purely polynomial multi-index encodes both a linear combination of rooted and of integrated (planted) trees. 
In line with branched rough paths in Section \ref{Sect5}, when restricted to populated multi-indices, cf. \eqref{gpam22}, both maps $\mathring{\phi}_-$ and $\mathring{\phi}$ are one-to-one but not onto. 
Moreover we have the following characterization:
\begin{lemma}\label{lem:gpamdic}
	For every multi-index $\beta$,
	\begin{equation}\label{gpam04}
	\mathring{\phi}_- \z_\beta = \sum_{\tau\in\mathcal{T}_\beta } \frac{\sigma(\beta)}{\sigma(\tau)}\tau.
	\end{equation}
\end{lemma}
Here $\mathcal{T}_\beta$ is the set of trees with $\beta(k)$ nodes with $k$ children and $\beta(\n)$ decorations $\mathscr{I} X^\n
$. 
Recall that we think of the latter as subtrees, so they count as children. For example, \eqref{gpam05} belongs to $\mathcal{T}_{\beta}$ for $\beta= e_0 + e_3 + e_{\n_1} + e_{\n_2}$. 
The number $\sigma(\beta)$ is the same as \eqref{fs24}; in particular, it does not depend on $\beta(\n)$ for any $\n$. 
The proof of Lemma \ref{lem:gpamdic} is similar to that of Lemma \ref{lemrp}, and thus we omit it.

\medskip

Identity \eqref{gpam04} establishes the matrix representation of $\mathring{\phi}_-$, namely
\begin{equation}\label{ab01}
(\mathring{\phi}_-)_\tau^{\beta} = \left\{\begin{array}{cl}
\frac{\sigma(\beta)}{\sigma(\tau)} & \mbox{if }\tau\in \mathcal{T}_\beta\\
0 &\mbox{otherwise}
\end{array}\right\};
\end{equation}
this, in turn, defines the transposed map $\mathring{\phi}_-^\dagger$.

\medskip

\subsection{Relating $\mathsf{L}$ to grafting operators on $\mathscr{B}$}
\mbox{}

There are two operations on $\mathscr{B}$, namely grafting and increasing polynomial decorations, cf. \cite[Definition 4.7]{BCCH19}. 
Let us start with grafting, and define\footnote{
It is straightforward to check that this definition is the non-recursive version of \cite[Definition 4.7]{BCCH19}.} 
in line with \eqref{fs103}
\begin{equation}\label{gpam24}
\tau_1 \gls{arrown} \tau_2 = \sum_{\tau} m_\n (\tau_1,\tau_2;\tau) \tau,
\end{equation}
where $m_\0 (\tau_1,\tau_2;\tau) = m (\tau_1,\tau_2;\tau)$, cf. \eqref{fs20}, and $m_{\n\neq \0} (\tau_1,\tau_2;\tau)$ is the number of decorations $\mathscr{I} X^\n$ in $\tau_2$ that when replaced by $\mathcal{I}\tau_1$ yield $\tau$. As in Subsection \ref{Sect5.4}, let us denote by $Z_\tau$ the standard basis of a linear space indexed by trees $\tau$. We define for every $\n$
\begin{equation}\label{gpam08}
Z_{\tau_1} \gls{preLien} Z_{\tau_2} = \sum_\tau n_\n (\tau_1,\tau_2;\tau) Z_\tau,
\end{equation}
where $n_\0 (\tau_1,\tau_2;\tau) = n (\tau_1,\tau_2;\tau)$, cf. \eqref{fs27}, and $n_{\n\neq \0}(\tau_1,\tau_2;\tau)$ is the number of single cuts of $\tau$ such that the branch is $\tau_1$ and, after adjoining the decoration $\mathscr{I} X^\n$ to the trunk at the place the cut was made, yield $\tau_2$. 
The combinatorial identity $n_\n(\tau_1,\tau_2;\tau)\sigma(\tau_1)\sigma(\tau_2)$ $=m_\n(\tau_1,\tau_2;\tau)\sigma(\tau)$ still holds\footnote{ The proof follows the argument in \cite[Proposition 4.3]{Hoffman}.}, so that \eqref{gpam08} may be rewritten as
\begin{equation}\label{gpam09}
\sigma(\tau_1)Z_{\tau_1} \leadsto_\n \sigma(\tau_2)Z_{\tau_2} = \sum_\tau m_\n (\tau_1,\tau_2;\tau) \sigma(\tau)Z_\tau.
\end{equation}
In particular, \eqref{gpam08} defines a pre-Lie product which is isomorphic to \eqref{gpam24}.
\begin{lemma}\label{lem:gpam1}For every $\n$ and every $\tau_1$, $\tau_2$
	\begin{equation}\label{gpam10}
	\mathring{\phi}_-^\dagger (Z_{\tau_1}\leadsto_\n Z_{\tau_2}) = (\mathring{\phi}_-^\dagger Z_{\tau_1}) ( D^{(\n)} \mathring{\phi}_-^\dagger Z_{\tau_2}).
	\end{equation}
\end{lemma}
Identity \eqref{gpam10} may be regarded as a pre-Lie morphism into $\tilde{\mathsf{L}}$, cf. \eqref{p26}, namely
\begin{equation*}
\mathring{\phi}_-^\dagger (Z_{\tau_1}\leadsto_{\n_1} Z_{\tau_2})D^{(\n_2)} = (\mathring{\phi}_-^\dagger Z_{\tau_1}) D^{(\n_1)}\prelie (\mathring{\phi}_-^\dagger Z_{\tau_2})D^{(\n_2)}.
\end{equation*}
Note that, unlike \cite[Proposition 4.21]{BCCH19} and in line with Subsection \ref{Sect5.4}, our construction does not define a free (multi) pre-Lie algebra.
\begin{proof}
	The case $\n = \0$ follows from the arguments of Lemma \ref{lembrp}. Let us now fix $\n\neq \0$; multiplying \eqref{gpam10} by $\sigma(\tau_1)\sigma(\tau_2)$ and combining \eqref{gpam04} and \eqref{gpam09}, we see that \eqref{gpam10} follows from
	\begin{equation}\label{gpam12}
	\sum_\tau m_\n (\tau_1,\tau_2;\tau) \sigma(\beta)\z^\beta = \sigma(\beta_1)\sigma(\beta_2) (\z^{\beta_1} D^{(\n)})\z^{\beta_2},
	\end{equation}
	with the understanding that $\beta_1$ and $\beta_2$ are determined by $\tau_1 \in \mathcal{T}_{\beta_1}$ and $\tau_2 \in \mathcal{T}_{\beta_2}$. Note that $m_\n (\tau_1,\tau_2;\tau)\neq 0$ only if $\beta_1 + \beta_2 = \beta + e_\n$; moreover, in such a case $\sigma(\beta_1)\sigma(\beta_2) = \sigma(\beta)$. Since $(\z^{\beta_1} D^{(\n)})\z^{\beta_2} = \beta_2 (\n) \z^{\beta_1 + \beta_2 - e_\n}$, \eqref{gpam12} reduces to 
	\begin{equation*}
	\sum_{\tau: \beta_1 + \beta_2 = \beta + e_\n} m_\n (\tau_1,\tau_2;\tau) = \beta_2 (\n).
	\end{equation*}
	This clearly holds because $\beta_2 (\n)$ is the number of decorations $\mathscr{I} X^\n$ that $\tau_2$ contains.
\end{proof}

\medskip

We now turn to the second operation, namely $\gls{arrowup}$, which increases polynomial decorations. 
As for $\curvearrowright_\n$, there is a non-recursive expression of \cite[Definition 4.7]{BCCH19} given in this case by\footnote{ The coefficient $m_\n (X^{\n + (1,0)},\tau;\tau')$ is meaningful because we think of the decoration $\mathscr{I} [X^{\n + (1,0)}]$ as a subtree, and thus grafting $X^{\n + (1,0)}$ makes sense.} 
\begin{equation}\label{gpam13}
\uparrow_1 \tau = \sum_{{\bf n}} \sum_{\tau'} m_\n (X^{\n + (1,0)},\tau;\tau') \tau',
\end{equation}
and a similar expression for $\uparrow_2$ (in the sequel, we will focus on $\uparrow_1$).
On the dual side, we define the operator $\gls{sharp}$ as
\begin{equation}\label{gpam14}
\sharp_1 Z_\tau = \sum_\n (n_1 + 1)\sum_{\tau'} n_\n (X^{\n + (1,0)},\tau;\tau') Z_{\tau'},
\end{equation}
and analogously for $\sharp_2$. 
Since by \eqref{gpam16}
\begin{equation*}
n_\n(X^{\n + (1,0)},\tau ;\tau')\neq 0\;\implies (n_1 +1) N(\tau) = N(\tau'), 
\end{equation*}
and thanks to $n_\n(X^{\n + (1,0)},\tau ;\tau')\sigma(\tau)$ $=m_\n(X^{\n + (1,0)},\tau ;\tau')\sigma(\tau')$, we may pass from \eqref{gpam14} to \eqref{gpam13} by
\begin{equation}\label{gpam17}
\sharp_1 \sigma(\tau) N(\tau) Z_\tau =\sum_{{\bf n}} \sum_{\tau'} m_\n (X^{\n + (1,0)},\tau;\tau') \sigma(\tau') N(\tau') Z_{\tau'}. 
\end{equation}
\begin{lemma}\label{lem:gpam2}
	For every $\tau$
	\begin{equation}\label{gpam15}
	\mathring{\phi}_-^\dagger \sharp_1 Z_\tau = \partial_1 \mathring{\phi}_-^\dagger Z_\tau.
	\end{equation}
\end{lemma}
\begin{proof}
By definitions (\ref{ao30}) of $\partial_1$ and (\ref{gpam14}) of $\sharp_1$,
it suffices to establish for fixed ${\bf n}$ (and $\tau$)
\begin{align*}
\mathring\phi_-^\dagger\sum_{\tau'}n_{\bf n}(X^{{\bf n}+(1,0)},\tau;\tau')Z_{\tau'}
=\mathsf{z}_{{\bf n}+(1,0)}(D^{({\bf n})}\mathring\phi_-^\dagger Z_\tau).
\end{align*}
Denoting by $\beta$ and $\beta'$ the multi-indices with $\tau\in{\mathcal T}_\beta$
and $\tau'\in{\mathcal T}_{\beta'}$, respectively, and appealing to the definition
(\ref{gpam04}) of $\mathring\phi_-$ (and thus its transpose) this reduces to 
the following identity in $\mathsf{T}^*$
\begin{align*}
\sum_{\tau'}n_{\bf n}(X^{{\bf n}+(1,0)},\tau;\tau')
\frac{\sigma(\beta')}{\sigma(\tau')}\mathsf{z}^{\beta'}
=\frac{\sigma(\beta)}{\sigma(\tau)}\mathsf{z}_{{\bf n}+(1,0)}(D^{({\bf n})}\mathsf{z}^\beta),
\end{align*}
which by the combinatorial identity can be reformulated as
\begin{align}\label{ffw02}
\sum_{\tau'}m_{\bf n}(X^{{\bf n}+(1,0)},\tau;\tau')
\sigma(\beta')\mathsf{z}^{\beta'}
=\sigma(\beta)\mathsf{z}_{{\bf n}+(1,0)}(D^{({\bf n})}\mathsf{z}^\beta).
\end{align}
We distinguish the case ${\bf n}\not=\0$ and the remaining case of
\begin{align}\label{ffw03}
\sum_{\tau'}m(X^{(1,0)},\tau;\tau')
\sigma(\beta')\mathsf{z}^{\beta'}
=\sigma(\tau)\mathsf{z}_{(1,0)}(D^{({\bf 0})}\mathsf{z}^\beta),
\end{align}
which we treat first. 

\medskip

We note that by definition, $m(X^{(1,0)},\tau;\tau')\not=0$
implies that there exists a $k\ge 0$ such that
\begin{align}\label{ffw01}
\beta'=\beta-e_k+e_{k+1}+e_{(1,0)}.
\end{align}
Hence by definition (\ref{ao35}) of $D^{({\bf 0})}$,
for (\ref{ffw03}) it suffices to show for fixed $k$
\begin{align*}
\sum_{\tau':k}m(X^{(1,0)},\tau;\tau')
\sigma(\beta')\mathsf{z}^{\beta'}
=\sigma(\beta)\mathsf{z}_{(1,0)}(k+1)\mathsf{z}_{k+1}
(\partial_{\mathsf{z}_k}\mathsf{z}^\beta),
\end{align*}
where the sum is over all trees $\tau'$ that arise from attaching the decoration
$X^{(1,0)}$ to a node of the tree $\tau$ with $k$ children.
By (\ref{ffw01}) and definition (\ref{fs24}) of $\sigma(\beta)$, this reduces to the combinatorial
identity
\begin{align*}
\sum_{\tau':k}m(X^{(1,0)},\tau;\tau')=\beta(k),
\end{align*}
which is tautological by definition of $m$.

\medskip

We now turn for (\ref{ffw02}) for ${\bf n}\not=\0$.
By definition, $m_{\bf n}(X^{{\bf n}+(1,0)},\tau;\tau')\not=0$ implies
%
$\beta'$ $=\beta-e_{{\bf n}}+e_{{\bf n}+(1,0)}$.
%
Thus by definition (\ref{ao30}) of $D^{({\bf n})}$, and by
definition (\ref{fs24}) of $\sigma(\beta)$, (\ref{ffw02}) reduces to
\begin{align*}
\sum_{\tau'}m_{\bf n}(X^{{\bf n}+(1,0)},\tau;\tau')
=\beta({\bf n}),
\end{align*}
which again is tautological by definition of $m_{\bf n}$.
\end{proof}

\medskip

\subsection{Relating $\mathring{\phi}_-^\dagger$ to $\mathring{\Upsilon}$}\label{Sect7.3}
\mbox{}

In the spirit of Subsection \ref{Sect5.5}, we want to relate \eqref{gpam10} and \eqref{gpam15} with the morphism properties established in \cite[Lemma 4.8]{BCCH19}.
Fixing a nonlinearity $a$, for arbitrary 
polynomial\footnote{This notation is chosen to agree with \cite{BCCH19};
we would rather replace ${\mathbf u}$ by $p$.} ${\bf u}$ and
shift\footnote{Again, the notation is aligned upon
\cite{BCCH19}; for us, $y=h$ would be more natural} $y$, 
we take inspiration from (\ref{ao27bis}) to generalize definition (\ref{fs31}) to
\begin{align}\label{ffw17}
\gls{upsilonCirc}^a[\tau]({\bf u},y):=
\mathring\phi_{-}^\dagger N(\tau)\sigma(\tau)
Z_{\tau}[a(\cdot+{\bf u}(y)),{\bf u}(\cdot+y)-{\bf u}(y)].
\end{align}
By linearity, (\ref{ffw17}) extends from trees $\tau$ to linear combinations thereof.
It is an easy consequence of Lemmas \ref{lem:gpam1} and \ref{lem:gpam2} that 
for fixed $a$,
this linear map $\mathring\Upsilon^{a}$ from ${\mathscr B}$ into the space of functions  
of polynomial/base-point $({\bf u},y)$ is a morphism w.~r.~t.~to $\curvearrowright$
and $\uparrow$, in line with \cite[Lemma 4.8]{BCCH19}:

\begin{corollary}
\begin{align}
\mathring\Upsilon^a[\tau_1\curvearrowright_{\bf n}\tau_2]&=\mathring\Upsilon^a[\tau_1]
\big(\frac{d}{du^{({\bf n})}}\mathring\Upsilon^a[\tau_2]\big),\label{ffw18}\\
\mathring\Upsilon^a[\uparrow_1\tau]&=\frac{d}{dy_1}\mathring\Upsilon^a[\tau],\label{ffw20}
\end{align}
where the coefficients $u^{({\bf n})}$ are defined through 
${\bf u}(x)=\sum_{\bf n}\frac{1}{{\bf n}!}u^{({\bf n})}x^{\bf n}$.
\end{corollary}

\begin{proof} 
We start with (\ref{ffw18}) by inserting definition (\ref{gpam24}) of $\curvearrowright_{\bf n}$ and then
use the linearity and definition (\ref{ffw17}) of $\mathring\Upsilon^a[\tau]$ to obtain
\begin{align*}
\mbox{l.~h.~s.~of (\ref{ffw18})}=\mathring\phi_-^\dagger\sum_{\tau}N(\tau)\sigma(\tau)m_{\bf n}(\tau_1,\tau_2;\tau)
Z_\tau,
\end{align*}
where here and in the sequel we suppress the argument 
$[a(\cdot+{\bf u}(y)),$ ${\bf u}(\cdot+y)-{\bf u}(y)]$.
Since $m_{\bf n}(\tau_1,\tau_2;\tau)\not=0$ implies $N(\tau_1)N(\tau_2)$ $={\bf n}!N(\tau)$, 
this yields by (\ref{gpam09})
\begin{align*}
\mbox{l.~h.~s.~of (\ref{ffw18})}=\mathring\phi_-^\dagger\frac{1}{{\bf n}!}
N(\tau_1)\sigma(\tau_1)Z_{\tau_1}\leadsto_{\bf n}N(\tau_2)\sigma(\tau_2)Z_{\tau_2}.
\end{align*}
This allows us to appeal to (\ref{gpam10}) yielding
\begin{align}\label{ffw19}
\mbox{l.~h.~s.~of (\ref{ffw18})}=(\mathring\phi_-^\dagger N(\tau_1)\sigma(\tau_1)Z_{\tau_1})
\frac{1}{{\bf n}!}(D^{({\bf n})}\mathring\phi_-^\dagger N(\tau_2)\sigma(\tau_2) Z_{\tau_2}).
\end{align}
Using that the argument of $\mathring\Upsilon^{a}$, see (\ref{ffw17}), can also be written as
$[a(\cdot+u^{({\bf 0})}),$ 
$\sum_{{\bf n}\not=0}\frac{1}{{\bf n}!}u^{({\bf n})}x^{\bf n}]$, we learn
from the definitions (\ref{ao21}) 
and (\ref{hk05}) of $D^{({\bf n})}$ that (\ref{ffw19}) is identical to the
r.~h.~s.~of (\ref{ffw18}).

\medskip

We now turn to the l.~h.~s.~of (\ref{ffw20}). By definition (\ref{gpam13}) of
$\uparrow_1$, definition (\ref{ffw17}) of $\mathring\Upsilon^{a}$, and (\ref{gpam17})
we have (still with suppressed argument)
\begin{align*}
\mbox{l.~h.~s.~of (\ref{ffw20})}=\mathring\phi_-^\dagger\sharp_1 N(\tau)\sigma(\tau)Z_{\tau},
\end{align*}
so that we may appeal to (\ref{gpam15}) to the effect of
\begin{align*}
\mbox{l.~h.~s.~of (\ref{ffw20})}=(\partial_1\mathring\phi_-^\dagger N(\tau)\sigma(\tau)Z_{\tau})
[a(\cdot+{\bf u}(y),{\bf u}(\cdot+y)-{\bf u}(y)].
\end{align*}
By the characterization (\ref{ao25}) of $\partial_1$ and definition 
(\ref{ffw17}) of $\mathring\Upsilon^{a}$,
this turns into the l.~h.~s.~of (\ref{ffw20}).
\end{proof}

\medskip

\subsection{Relating $(\Delta,\Delta^+)$ to $(\Delta_H,\Delta_H^+)$}\label{Sect7.4}
\mbox{}

Our final goal is to connect our Hopf algebra structure to that of \cite{Hairer}. A first strong hint that the structures are compatible is \eqref{Hai4.14}; 
however, the analogue \eqref{comoduleHairer} for the coaction $\Delta$ is missing. 

\medskip

We will pass to a coarser tree-based description, which no longer distinguishes different polynomial decorations of a given node but contracts them by multiplication. 
We identify the space of linear combinations of trees with contracted decorations with the model space\footnote{ In the more general context of \cite{BCCH19}, this contracted space corresponds to $\mathscr{V}$, but we will directly work with the restriction to relevant trees.} in \cite[Subsection 4.2]{Hairer}, and denote it by $\gls{ModelSpaceH}$. 
More precisely, the model space $\T_H$ consists of linear combinations of trees of the inductive form
\begin{equation*}
\tau = \treeZero X^\n (\prod_{j\in\mathsf{J}}\mathcal{I}\tau_j),
\end{equation*}
as well as their integrated versions $\mathcal{I}\tau$. 
The passage from the detailed to the contracted description is encoded, as in \cite[Subsection 4.1]{BCCH19}, by a linear map $\gls{contraction}: \mathscr{B} \to \T_H$ recursively given, for $\tau$ as in \eqref{gpam01}, by
\begin{equation*}
\mathcal{Q}\tau = \treeZero  X^{\sum_{i\in\mathsf{I}} \n_i}(\prod_{j\in\mathsf{J}} \mathcal{I}\mathcal{Q}\tau_j).
\end{equation*}
For example, \eqref{gpam05} turns into
\begin{equation*}
\mathcal{Q} \treeNineteen = \treeNine{\n_1 + \n_2}.
\end{equation*}
The map $\mathcal{Q}$ defines a new dictionary $\phi_-,\phi: \T\to\T_H$ by means of\footnote{ Here, we are implicitly extending $\mathring{\phi}_-$ to the whole model space $\T$ by projecting onto the complement of the polynomial sector, i.~e.~ $\mathring{\phi}_-\z_{e_{\n}} = 0$.}
\begin{equation}\label{ab02}
\phi_- = \mathcal{Q}\mathring{\phi}_- \, , \quad \phi = \mathcal{Q}\mathring{\phi}.
\end{equation} 
Equivalently, due to \eqref{gpam02} and \eqref{gpam03}, $\phi_-$ and $\phi$ are determined by $\phi_- \z_{\beta = 0} = 0$ and then recursively in the length of $\beta$ by
\begin{equation}\label{gp01}
\phi_- \z_\beta 
= \sum_{k\geq0} \sum_{e_k+\beta_1+\dots+\beta_k = \beta} 
\phi \z_{\beta_1} \cdots \phi \z_{\beta_k} \treeZero ,
\end{equation}
and
\begin{align}\label{gp13}
\phi \z_\beta = \left\{
\begin{array}{cl}
X^\n &\mbox{if} \; \beta = e_\n \\
\I \phi_- \z_\beta &\mbox{otherwise}
\end{array}\right\}.
\end{align}
For example, we have
\begin{equation}\label{ex01}
\begin{split}
\phi_- \z_{e_0} = \treeZero&, \quad
\phi \z_{e_0} = \treeOne,\\
\phi \z_{e_0+e_1+e_2+e_{(2,0)}} &= \treeThree.
\end{split}
\end{equation}
Moreover, $\phi_-\z_\beta$ and $\phi\z_\beta$ vanish unless \eqref{gpam22} is satisfied.
Due to $\mathcal{Q}$, unlike $\mathring{\phi}_-$ and $\mathring{\phi}$, $\phi_-$ and $\phi$ are not one-to-one even if we restrict to populated indices; e.~g. $\z_{e_1+e_{(2,0)}}$ and $\z_{e_2+2e_{(1,0)}}$ are both mapped to $\treeFour$ by $\phi_-$.

\medskip

We denote by $\gls{homogeneityHairer}$ the homogeneity in \cite[p.23]{Hairer}, which is defined as follows: $|X^{\n'}|_H:=|\n'|$, $|\treeZero|_H := \alpha-2$ and then inductively via  $|\tau_1\tau_2|_H := |\tau_1|_H+|\tau_2|_H$ and $|\I\tau|_H := |\tau|_H +2$. 
We now argue that 
$\phi_-\z_\beta$ 
is a linear combination of trees $\tau$ satisfying $|\tau|_H=\lhom\beta\rhom-2$. 
The analogue for $\phi$ holds true with $|\tau|_H=\lhom\beta\rhom$. 
Indeed, this can be seen by induction in the length of $\beta$, where we may restrict to $\beta$'s satisfying \eqref{gpam22}. 
For $\beta$'s of length one we only have to consider $\phi_-\z_{e_0}$, $\phi \z_{e_0}$ and $\phi \z_{e_\n}$, for which the statement is clear by \eqref{gp13}, \eqref{ex01} and recalling that $\lhom e_0 \rhom = \alpha$ and $\lhom e_\n \rhom = |\n|$, cf. \eqref{ao52}. 
In the induction step, the statement for $\phi_-\z_\beta$ follows from \eqref{gp01}, by using the induction hypothesis, the definition of $|\cdot|_H$ and $\lhom\beta\rhom = \lhom e_k+\beta_1 + \dots + \beta_k | = \alpha + \lhom \beta_1 \rhom + \dots +\lhom \beta_k \rhom$ in every summand of \eqref{gp01}. As a consequence we obtain the corresponding statement for $\phi \z_\beta$ from \eqref{gp13}.

\medskip

The analogue of the space $\T^+$ of \cite[p. 24]{Hairer} is denoted by $\gls{TplusH}$ with its basis elements $X^\m \prod_i \J_{\n_i}^H \tau_i$, where $\tau_i$'s are elements of $\T_H$.
Recall that $\J_\n^H \tau$ vanishes for $\tau=X^{\n'}$ for every $\n'$, and for $|\n| \geq |\tau|_H +2$.
Moreover, we define $\gls{Phi}:\T^+ \to\T^+_H$ by postulating it to be multiplicative and
\begin{equation}\label{gp02}
\Phi \Z^{(0,(1,0))} = X_1, \quad
\Phi \Z^{(0,(0,1))} = X_2, \quad
\Phi \J_\n \z_\beta = \gls{embeddingHairer} \phi_- \z_\beta,
\end{equation}
for $|\n|<\lhom\beta\rhom$.
By the compatibility of the homogeneities $|\cdot|_H$ and $\lhom\cdot\rhom$, see above, the last equality in \eqref{gp02} can also be seen to hold for arbitrary $\beta$.

\medskip

Finally, we recall the definition of the flipped\footnote{ 
i.~e. $\textnormal{tw} \circ \Delta_H$ is the coaction in \cite{Hairer}, where $\textnormal{tw}(x\otimes y) = y\otimes x$, and the same for the coproduct} 
coaction $\gls{comoduleH}$ and coproduct $\gls{coproductH}$ of \cite[pp. 25-26]{Hairer}.
The coaction $\Delta_H:\T_H\to\T^+_H\otimes\T_H$ is by multiplicativity $\Delta_H \tau\tau' = (\Delta_H\tau)(\Delta_H\tau')$ recursively defined via
\begin{equation}\label{comoduleHairer}
\begin{split}
\Delta_H 1 = 1\otimes 1,\quad
&\Delta_H\treeZero = 1\otimes \treeZero,\quad
\Delta_H X_i = 1\otimes X_i + X_i \otimes 1, \\
\Delta_H \mathcal{I}\tau &= (\mathrm{id}\otimes \mathcal{I}) \Delta_H\tau 
+ \sum_{\n} \mathcal{J}_\n^H \tau \otimes \frac{X^\n}{\n!}. 
\end{split}
\end{equation}
The coproduct $\Delta^+_H:\T^+_H\to\T^+_H\otimes\T^+_H$ is by multiplicativity $\Delta^+_H \tau\tau' = (\Delta^+_H\tau)(\Delta^+_H\tau')$ via the coaction recursively defined by
\begin{equation}\label{coproductHairer}
\begin{split}
\Delta^+_H 1 &= 1\otimes 1,\quad
\Delta^+_H X_i = 1\otimes X_i + X_i \otimes 1, \\
\Delta^+_H \mathcal{J}_\n^H \tau &= (\mathrm{id}\otimes \mathcal{J}_\n^H) \Delta_H\tau 
+ \sum_{\m} \mathcal{J}_{\n+\m}^H \tau \otimes \frac{X^\m}{\m!}.
\end{split}
\end{equation}

With this definitions at hand we may formulate the following intertwining property, which is the main result of this section.
\begin{proposition}\label{lemrs}
\begin{align}
\Delta_H\phi_-&=(\Phi\otimes\phi_-)\Delta,\label{fw01}\\
\Delta_H^+\Phi&=(\Phi\otimes\Phi)\Delta^+.\label{gp03}
\end{align}
\end{proposition}
Since by definition, $\Phi$ is an algebra morphism, and bialgebra morphisms between Hopf algebras are automatically Hopf algebra morphisms, $\Phi$ is in particular a Hopf algebra morphism.

\begin{proof}
{\sc Step 1}. From coaction to coproduct by intertwining. 
We claim that (\ref{fw01}) implies (\ref{gp03}).
By multiplicativity of $\Phi$, $\Delta_H^+$, and $\Delta^+$, it is enough
to establish (\ref{gp03}) when applied to $\mathsf{Z}^{(e_{(\beta,{\bf n})},{\bf 0})}$,
$\mathsf{Z}^{(0,(1,0))}$, and $\mathsf{Z}^{(0,(0,1))}$. The two latter cases follow from
(\ref{Hai4.14b}), (\ref{gp02}), and (\ref{comoduleHairer}). For the former case we 
start from (\ref{Hai4.14}) to which we apply $\Phi\otimes\Phi$,
so that by (\ref{gp02}) (and the multiplicativity of $\Phi$)
\begin{align*}
(\Phi\otimes\Phi)\Delta^+{\mathcal J}_{\bf n}\mathsf{z}_\beta
=({\rm id}\otimes{\mathcal J}_{\bf n}^H)
(\Phi\otimes\phi_-)\Delta\mathsf{z}_\beta
+\sum_{{\bf m}}{\mathcal J}_{{\bf m}+{\bf n}}^H\phi_-
\mathsf{z}_\beta\otimes\frac{X^{\bf m}}{{\bf m}!}.
\end{align*}
On the other hand, we use (\ref{coproductHairer}) for $\tau=\phi_-\mathsf{z}_\beta$
so that by (\ref{gp02})
\begin{align*}
\Delta_H^+\Phi{\mathcal J}_{\bf n}\mathsf{z}_\beta=
({\rm id}\otimes{\mathcal J}_{\bf n}^H)\Delta_H\phi_-\mathsf{z}_\beta
+\sum_{{\bf m}}{\mathcal J}_{{\bf m}+{\bf n}}^H\phi_-
\mathsf{z}_\beta\otimes\frac{X^{\bf m}}{{\bf m}!}.
\end{align*}
We now see that (\ref{fw01}) implies $\Delta_H^+\Phi{\mathcal J}_{\bf n}$
$=(\Phi\otimes\Phi)\Delta^+{\mathcal J}_{\bf n}$.

\medskip

{\sc Step 2}. Taking care of the polynomial sector. 
We claim that (\ref{fw01}) implies\footnote{
Here and in the sequel, $\mathsf{Z}^{(0,\n)}\otimes\z_\n\in \mathsf{T}^+\otimes \mathsf{T}^*$ is identified with its representative in $\mathcal{L}(\mathsf{T},\mathsf{T}^+)$.}
\begin{align}\label{fw02}
\Delta_H\phi=(\Phi\otimes\phi)\Delta+\Phi({\mathcal J}_\0
+\sum_{{\bf n}\not=\0}\mathsf{Z}^{(0,{\bf n})}
\otimes\mathsf{z}_{\bf n})\otimes\mathsf{1}
\end{align}
componentwise, by which we mean that for fixed $\beta$ with $[\beta]\ge 0$,
(\ref{fw01}) applied on $\mathsf{z}_\beta$ implies
(\ref{fw02}) applied on $\mathsf{z}_\beta$, and that (\ref{fw02}) automatically
holds when applied to $\mathsf{z}_{e_{\bf n}}$.
We split (\ref{fw02}) into
\begin{align}
\Delta_H\phi&=(\Phi\otimes\phi)\Delta+\Phi{\mathcal J}_\0\otimes\mathsf{1}
\quad\mbox{on}\;\mathsf{\tilde T},\label{fw03}\\
\Delta_H\phi&=(\Phi\otimes\phi)\Delta+\Phi(\sum_{{\bf n}\not=\0}\mathsf{Z}^{(0,{\bf n})}
\otimes\mathsf{z}_{\bf n})\otimes\mathsf{1}
\quad\mbox{on}\;\mathsf{\bar T}.\label{fw06}
\end{align}
By (\ref{gp13}) we have $\phi$ $={\mathcal I}\phi_-$ on $\mathsf{\tilde T}$,
by (\ref{gp02}) we have $\Phi{\mathcal J}_{\bf n}$ $={\mathcal J}_{\bf n}^{H}\phi_-$,
so that by (\ref{comoduleHairer}) we obtain
\begin{align*}
\Delta_H\phi=({\rm id}\otimes{\mathcal I})\Delta_H\phi_-+\sum_{{\bf n}}
\Phi{\mathcal J}_{\bf n}\otimes\frac{X^{\bf n}}{{\bf n}!}\quad\mbox{on}\;\mathsf{\tilde T}.
\end{align*}
We now insert (\ref{fw01}):
\begin{align*}
\Delta_H\phi=(\Phi \otimes{\mathcal I}\phi_-)\Delta+\sum_{{\bf n}}
\Phi{\mathcal J}_{\bf n}\otimes\frac{X^{\bf n}}{{\bf n}!}\quad\mbox{on}\;\mathsf{\tilde T}.
\end{align*}
Hence (\ref{fw03}) follows from
\begin{align}\label{fw04}
\big( {\rm id}\otimes(\phi-{\mathcal I}\phi_-) \big) \Delta=\sum_{{\bf n}\not=\0}{\mathcal J}_{\bf n}
\otimes\frac{X^{\bf n}}{{\bf n}!}\quad\mbox{on}\;\mathsf{\tilde T}.
\end{align}

\medskip

Here comes the argument for (\ref{fw04}): By (\ref{gp13}), $\phi-{\mathcal I}\phi_-$
is the projection onto the polynomial sector $\bar{\T}_H$. Hence (\ref{fw04}) assumes the form
\begin{align*}
({\rm id}\otimes (1-P^\dagger))\Delta=\sum_{{\bf n}\not=\0}{\mathcal J}_{\bf n}
\otimes\frac{\mathsf{z}_{e_{\bf n}}}{{\bf n}!}\quad\mbox{on}\;\mathsf{\tilde T},
\end{align*}
where $P^\dagger$ denotes the projection from $\mathsf{T}$ onto $\mathsf{\tilde T}$
(defined through the direct sum $\mathsf{T}=\mathsf{\tilde T}\oplus\mathsf{\bar T}$).
In view of (\ref{ao77}) and (\ref{fw26}), the last statement amounts to
\begin{align}\label{fw05}
((\rho U)\mathsf{z}_{\bf n})=\iota_{\bf n}U
\quad\mbox{mod}\;\mathsf{\bar T}^*\quad\mbox{for all}\;U\in{\rm U}(\mathsf{L})
\quad\mbox{and}\;{\bf n}\not=\0.
\end{align}
We check (\ref{fw05}) on the basis elements $U=D_{(J,{\bf m})}$, cf.~(\ref{ao75}):
By definition (\ref{io01}) of $\iota_{\bf n}$, the r.~h.~s.~of (\ref{fw05}) does not vanish
iff $(J,{\bf m})=(e_{(\beta,{\bf n})},\0)$, in which case both sides coincide
with $\mathsf{z}^\beta$. The l.~h.~s.~of (\ref{fw05}) is also non-zero
for $J=0$, it is then given by $(\frac{1}{{\bf m}!}\partial^{\bf m}\mathsf{z}_{\bf n})$
$\in\mathsf{\bar T}^*$, see (\ref{ao30}), so that the statement is not affected.

\medskip

We now turn to (\ref{fw06}), which can be rephrased as
\begin{align}
\Delta_H\phi \z_{e_{\bf n}}&=(\Phi\otimes\phi)\Delta \z_{e_{\bf n}}
+\Phi \mathsf{Z}^{(0,{\bf n})}\otimes\mathsf{1}
\quad\mbox{for all}\;{\bf n}\not=\0.\label{fw07}
\end{align}
In view of (\ref{comoduleHairer}) and the multiplicativity of $\Delta_H$ we have
\begin{equation*}
\Delta_H X^{\bf n}=\sum_{{\bf n}'+{\bf n}''={\bf n}} \tbinom{\n}{\n'} 
X^{{\bf n}'}\otimes X^{{\bf n}''}.
\end{equation*}
In view of the multiplicativity of $\Phi$ in conjunction with (\ref{gp02}) we
have $\Phi Z^{(0,{\bf n})}=X^{\bf n}$ (tautologically
including ${\bf n}=\0$). Together with (\ref{gp13}), we thus see
that (\ref{fw07}) reduces to \eqref{comodulePoly}.

\medskip

{\sc Step 3}. Inductive proof of (\ref{fw01}). We carry out an induction in the length
of $\beta$ with $[\beta]\ge 0$, the component of (\ref{fw01}). We start with
the base case of $\beta=0$. The l.~h.~s.~vanishes since (\ref{gp01}) includes
$\phi_-\mathsf{z}_{\beta=0}=0$. For the r.~h.~s.~to also vanish, we just need to
convince ourselves that 
\begin{align*}
\Delta \z_{\beta=0}=\sum_{{\bf n}\not=\0}
\mathsf{Z}^{(e_{(0,{\bf n})},\0)}\otimes\mathsf{z}_{e_{\bf n}}.
\end{align*}
By \eqref{SG01} and \eqref{ao77}, this amounts to the fact that the image of $\rho D_{(J,{\bf m})}$
$\in{\rm End}(\mathsf{T}^*)$, see (\ref{ao75}), 
contains $\mathsf{1}$ iff $(J,{\bf m})$ $=(e_{(0,{\bf n})},\0)$ for some ${\bf n}\not=\0$,
in which case the pre-image must be $\mathsf{z}_{\bf n}$.

\medskip

Obviously, the pre-dual $\mathbb{R}[[\z_k,\z_\n]]^\dagger$ of $\mathbb{R}[[\z_k,\z_\n]]$, which is isomorphic to $\mathbb{R}[\z_k,\z_\n]$, is endowed with a coproduct;
since $\mathsf{T}_H$ is endowed with a product, we may multiply linear maps from
$\mathbb{R}[[\mathsf{z}_k,\mathsf{z}_{\bf n}]]^\dagger$ into $\mathsf{T}_H$.
Extending $\phi$ (next to $\phi_-$) trivially from $\mathsf{T}$ to 
$\mathbb{R}[[\mathsf{z}_k,\mathsf{z}_{\bf n}]]^\dagger$, we may apply this to $\phi$
in order to give a sense to $\phi^k$. Furthermore, we may identify
elements of $\mathbb{R}[[\mathsf{z}_k,\mathsf{z}_{\bf n}]]$, 
like $\mathsf{z}_k$, with a linear map (of rank 1)
from $\mathbb{R}[[\mathsf{z}_k,\mathsf{z}_{\bf n}]]^\dagger$ into $\mathsf{T}_H$. 
Note finally that (\ref{gp01}) extends from populated, cf.~(\ref{newreference}), 
to all $\beta$; indeed, if $\beta$
is not populated, (\ref{gpam22}) is violated, hence if $e_k+\beta_1+\cdots+\beta_k$ $=\beta$,
one of the $\beta_l$'s violates (\ref{gpam22}) and thus the r.~h.~s.~of (\ref{gp01})
vanishes. 
This allows us to re-interpret (\ref{gp01}) in the more compact way
\begin{align}\label{fw11}
\phi_-=\treeZero\sum_{k\ge 0}\mathsf{z}_k\phi^k
\end{align}
as an identity in ${\mathcal L}(\mathbb{R}[[\mathsf{z}_k,\mathsf{z}_{\bf n}]]^\dagger,
\mathsf{T}_H)$.
By the multiplicativity of $\Delta_H$ and
(\ref{comoduleHairer}), this implies
\begin{align*}
\Delta_H\phi_-=(1\otimes\treeZero)\sum_{k\ge 0}\mathsf{z}_k(\Delta_H\phi)^k,
\end{align*}
an identity in the space of linear maps from
$\mathbb{R}[[\mathsf{z}_k,\mathsf{z}_{\bf n}]]^\dagger$ into the algebra
$\mathsf{T}_H^+\otimes\mathsf{T}_H$.

\medskip

It will be convenient to test this identity with elements $\mathsf{Z}$ of a space endowed with a non-degenerate pairing with $\mathsf{T}^+_H$, which we choose to be the\footnote{
unique up to linear isomorphisms} 
direct sum indexed by the basis elements of $\T^+_H$ with the corresponding canonical pairing. 
Since $\mathsf{T}^+_H$ is endowed with the polynomial product, this space carries a coproduct, so that in Sweedler's notation
\begin{align}\label{fw10}
\langle\mathsf{Z},\Delta_H\phi_-\rangle=\treeZero\sum_{k}\mathsf{z}_k\sum_{(Z)}
\langle\mathsf{Z}_{(1)},\Delta_H\phi\rangle\cdots
\langle\mathsf{Z}_{(k)},\Delta_H\phi\rangle,
\end{align}
with the understanding that the pairing acts on the $\T^+_H$ component, 
and for $k=0$ the $\sum_{(Z)}$-term is the counit applied to $\mathsf{Z}$.
Due to the presence of $\mathsf{z}_k$, the r.~h.~s.~of the $\beta$-component of this identity
only involves components $\gamma$'s of smaller length so that we may appeal to
the induction hypothesis (\ref{fw01}). 
We use (\ref{fw01}) in its upgraded form (\ref{fw02}),
tested with $\mathsf{Z}$. 
Note that the matrix representation of $\Phi$ in the standard bases of the polynomial algebras $\T^+$ and $\T^+_H$ has the finiteness property that allows us to define $\Phi^\dagger\mathsf{Z}\in\mathrm{U}(\mathsf{L})$. 
Indeed, this finiteness property is inherited from the corresponding one of $\mathring{\phi}_-$, see \eqref{ab01}, via \eqref{gp02} and \eqref{ab02}.
Using the argument from the end of Subsection \ref{Sect5.1}, we have 
$\langle\mathsf{Z},(\Phi\otimes\phi)\Delta\rangle$
$=\phi(\rho U)^\dagger$, where $U$ $:=\Phi^\dagger\mathsf{Z}$ $\in{\rm U}(\mathsf{L})$. 
Hence \eqref{fw02} tested with $\mathsf{Z}$ takes the form of an identity in
${\mathcal L}(\mathbb{R}[[\mathsf{z}_k,\mathsf{z}_{\bf n}]]^\dagger,\mathsf{T}_H)$
$\supset\mathsf{T}^*\otimes\mathsf{T}_H$
\begin{align}\label{fw09}
\langle\mathsf{Z},\Delta_H\phi\rangle=\phi(\rho U)^\dagger+\big(\iota_\0 U
+\sum_{{\bf n}\not=\0}(\varepsilon_{\bf n} U)\mathsf{z}_{\bf n}\big)\otimes\mathsf{1}
\quad\mbox{on}\;\mathsf{T}.
\end{align}
Extending $(\rho U)^\dagger$ trivially, we may think of (\ref{fw09})
as holding not just on $\mathsf{T}$ but all of $\mathbb{R}
[[\mathsf{z}_k,\mathsf{z}_{\bf n}]]^\dagger$, which allows us to insert (\ref{fw09})
into (\ref{fw10}).
Since $\Phi$ is multiplicative, $\Phi^\dagger$ preserves the coproduct, so that
we obtain
\begin{align}\label{fw14}
&\langle\mathsf{Z},\Delta_H\phi_-\rangle \nonumber \\
&=\treeZero\sum_{k}\mathsf{z}_k
\sum_{(U)}\prod_{l=1}^k\Big(\phi(\rho U_{(l)})^\dagger+\big(\iota_\0 U_{(l)}
+\sum_{{\bf n}\not=\0}(\varepsilon_{\bf n} U_{(l)})\mathsf{z}_{\bf n}\big)\otimes\mathsf{1}\Big),
\end{align}
again with the understanding that for $k=0$ the term $\sum_{(U)}$ 
reduces to $\varepsilon_{\bf 0} U$. 
Here we interpret
$\iota_\0+\sum_{{\bf n}\not=\0}\varepsilon_{\bf n}\otimes\mathsf{z}_{\bf n}$ as a linear
map from ${\rm U}(\mathsf{L})$, a space endowed with a coproduct, into
the algebra $\mathbb{R}[[\mathsf{z}_k,\mathsf{z}_{\bf n}]]$, 
so that powers make sense. 
To further rewrite \eqref{fw14}, we first claim that linear maps $\phi_1,\phi_2\in{\mathcal L}(\mathbb{R}[[\mathsf{z}_k,\mathsf{z}_{\bf n}]]^\dagger,\mathsf{T}_H)$ satisfy a generalized Leibniz rule, i.~e. for $U\in{\rm U}(\mathsf{L})$ we have
\begin{align}\label{fw12}
(\phi_1\phi_2)(\rho U)^\dagger=\sum_{(U)}(\phi_1(\rho U_{(1)})^\dagger)
(\phi_2(\rho U_{(2)})^\dagger),
\end{align}
which we show by duality.
Therefore, we consider elements $Z$ of a space endowed with a non-degenerate pairing with $\mathsf{T}_H$, which we choose to be the direct sum indexed by the basis elements of $\T_H$ with the corresponding canonical pairing.
This space carries a coproduct, since the product on $\mathsf{T}_H$ satisfies the right finiteness property.\footnote{Compare to ${\rm U}(\mathsf{L})$ and $\T^+$ in Section \ref{Sect4}.}
Applying the l.~h.~s. of \eqref{fw12} to $\z_\beta\in \mathbb{R}[[\z_k,\z_\n]]^\dagger$ and testing with $Z$, we obtain using the argument from the end of Subsection \ref{Sect5.1} by duality
\begin{equation*}
\langle Z, (\phi_1\phi_2)(\rho U)^\dagger \z_\beta\rangle
= 
\langle(\rho U) (\phi_1^*\phi_2^*) Z , \z_\beta\rangle.
\end{equation*}
Using \eqref{leib02} and again duality yields
\begin{align*}
\langle(\rho U) (\phi_1^*\phi_2^*) Z , \z_\beta\rangle
&= \big\langle \sum_{(U)} \big( ((\rho U_{(1)})\phi_1^*) ((\rho U_{(2)})\phi_2^*) \big) Z , \z_\beta\big\rangle \\
&= \big\langle Z , \sum_{(U)} \big( (\phi_1 (\rho U_{(1)})^\dagger) (\phi_2 (\rho U_{(2)})^\dagger) \big) \z_\beta\big\rangle,
\end{align*}
finishing the argument for \eqref{fw12}.
By (\ref{fw12}) in its multi-linear form, (\ref{fw14}) thus assumes the form
\begin{align}\label{fw15}
&\langle\mathsf{Z},\Delta_H\phi_-\rangle \nonumber \\
&=\treeZero\sum_{k}\mathsf{z}_k\sum_{(U)}
\sum_{l=0}^k \tbinom{k}{l} \phi^{l}(\rho U_{(1)})^\dagger\Big(\big(\iota_\0
+\sum_{{\bf n}\not=\0}\varepsilon_{\bf n}\otimes\mathsf{z}_{\bf n}\big)^{k-l}U_{(2)}
\otimes\mathsf{1}\Big).
\end{align}
Changing the order of summation in $k\ge l$ and applying (\ref{fw13}) 
with $U$ replaced by $U_{(2)}$, this collapses to
\begin{align*}
\langle\mathsf{Z},\Delta_H\phi_-\rangle=\treeZero\sum_{l}\sum_{(U)}
\phi^{l}(\rho U_{(1)})^\dagger((\rho U_{(2)})\mathsf{z}_{l}).
\end{align*}
Using once more (\ref{fw12}), this yields $\langle\mathsf{Z},\Delta_H\phi_-\rangle$ 
$=\treeZero\sum_{l}\mathsf{z}_{l}\phi^{l}(\rho U)^\dagger$. 
Appealing once more to (\ref{fw11}), this gives
$\langle\mathsf{Z},\Delta_H\phi_-\rangle$ $=\phi_-(\rho U)^\dagger$, 
which is the tested/dual form of (\ref{fw01}).
\end{proof}

\medskip

\appendix

\section{}
We first provide summation formulas for multi-indices which are repeatedly used throughout the text.
We assume that $J,J',J''$ are multi-indices over a set $\mathsf{I}$ and $A,B$ are finite sequences indexed by multi-indices over $\mathsf{I}$ in a real vector space.
\begin{lemma}\label{lemsum01}
	It holds
	\begin{multline*}
	(J(i)+1) \sum_{J'+J'' = J + e_i} \frac{1}{J'!}A_{J'}\otimes \frac{1}{J''!} B_{J''} \\= \sum_{J'+J''=J} \big(\frac{1}{J'!}A_{J'+e_i}\otimes \frac{1}{J''!} B_{J''} + \frac{1}{J'!}A_{J'}\otimes \frac{1}{J''!} B_{J''+e_i}\big).
	\end{multline*}
\end{lemma}
\begin{proof}
	We rewrite each of the terms of the r.~h.~s. as
	\begin{multline*}
	(J'(i)+1)\frac{1}{(J'+e_i)!}A_{J'+e_i}\otimes \frac{1}{J''!} B_{J''} \\+ (J''(i)+1)\frac{1}{J'!}A_{J'}\otimes \frac{1}{(J''+e_i)!} B_{J''+e_i}.
	\end{multline*}
	The sum then turns into
	\begin{multline*}
	\sum_{\substack{J'+J'' = J + e_i \\ J'\geq e_i}} J'(i)\frac{1}{J'!}A_{J'}\otimes \frac{1}{J''!} B_{J''} + \sum_{\substack{J'+J'' = J + e_i \\ J''\geq e_i}} J''(i)\frac{1}{J'!}A_{J'}\otimes \frac{1}{J''!} B_{J''}.
	\end{multline*}
	Note that the restrictions $J'\geq e_i$ and $J''\geq e_i$ are immaterial because of the factors $J'(i)$ and $J''(i)$. We may then combine the sums, and since $J'(i)+J''(i)=J(i)+1$, this concludes the proof.
\end{proof}

\medskip

\begin{lemma}\label{lemsum02}
	It holds
	\begin{equation*}
	\sum_J \frac{1}{J!} A_J = \sum_{k\geq 0} \frac{1}{k!} \sum_{i_1,...,i_k\in\mathsf{I}} A_{e_{i_1}+...+e_{i_k}}.
	\end{equation*}
\end{lemma}
\begin{proof}
	We split the sum of the l.~h.~s. according to $k=|J|$, parametrized according to $J=\sum_{j=1}^k e_{i_1}$ and count repetitions to obtain
	\begin{align*}
	\sum_J \frac{1}{J!} A_J &= \sum_{k\geq 0}\sum_{|J|=k} \frac{1}{J!} A_J\\
	&=\sum_{k\geq 0}\sum_{|J|=k}\frac{J!}{k!} \sum_{i_1,...,i_k\in\mathsf{I}} \frac{1}{J!}A_{e_{i_1}+...+e_{i_k}}\\
	&= \sum_{k\geq 0} \frac{1}{k!} \sum_{i_1,...,i_k\in\mathsf{I}} A_{e_{i_1}+...+e_{i_k}}.\qedhere
	\end{align*}
\end{proof}

\medskip

Finally, we show an extended chain rule, which connects to Fa\`a di Bruno's formula.

\begin{lemma}
For $l\ge 0$ and $U\in{\rm U}(\mathsf{L})$ it holds
\begin{align}\label{fw13}
((\rho U)\mathsf{z}_l)=\sum_{k\ge l} \tbinom{k}{l} \mathsf{z}_{k}(\iota_{\bf 0}+\sum_{{\bf n}\not=\0}
\varepsilon_{\bf n}\otimes\mathsf{z}_{\bf n})^{k-l}U.
\end{align}
\end{lemma}
Recall that \eqref{fw13} is short for 
\begin{align}\label{fw18}
\lefteqn{((\rho U)\mathsf{z}_l)=\sum_{k\ge l} \tbinom{k}{l} \mathsf{z}_{k}}\nonumber\\
&\times\sum_{(U)}\big(\iota_{\bf 0}U_{(1)}+\sum_{{\bf n}\not=\0}
(\varepsilon_{\bf n} U_{(1)})\mathsf{z}_{\bf n}\big)\cdots
\big(\iota_{\bf 0}U_{(k-l)}+\sum_{{\bf n}\not=\0}
(\varepsilon_{\bf n} U_{(k-l)})\mathsf{z}_{\bf n}\big).
\end{align}

\begin{proof}
Assume we have established (\ref{fw13}) for $U=\frac{1}{{\bf m}!}\partial^{\bf m}$
for any ${\bf m}$.
It then remains to argue for admissible $(\beta,{\bf n})$ that if $U$ satisfies (\ref{fw18}), 
also $\mathsf{z}^\beta U D^{({\bf m})}$ satisfies (\ref{fw18}), which we shall do
now. We first address the case of ${\bf m}\not=\0$, in which case the l.~h.~s.~of
(\ref{fw18}) (with $U$ replaced by $\mathsf{z}^\beta U D^{({\bf m})}$) vanishes.
According to Lemma \ref{lem4.2}, the rank-one constituents of the
coproduct of $\mathsf{z}^\beta U D^{({\bf m})}$ always involve one factor 
of the form $\mathsf{z}^\beta U' D^{({\bf m})}$, which is annihilated by both
$\iota_{\bf 0}$ and $\varepsilon_{\bf n}$. Hence it remains to consider ${\bf m}=\0$.
For the l.~h.~s.~of (\ref{fw18}) we have by (\ref{fw24}) and (\ref{ao35})
\begin{align}\label{fw19}
(\rho\mathsf{z}^\beta U D^{({\bf 0})})\mathsf{z}_l
=(l+1)\mathsf{z}^\beta((\rho U)\mathsf{z}_{l+1}).
\end{align}
Once more according to Lemma \ref{lem4.2}, the rank-one constituents of the
coproduct of $\mathsf{z}^\beta U D D^{({\bf 0})}$ are those of $U$ besides
exactly one factor 
of the form $\mathsf{z}^\beta U' D^{({\bf 0})}$, which is annihilated by $\varepsilon_{\bf n}$,
and is annihilated by $\iota_{\bf 0}$ unless $U'=\mathsf{1}$, in which case it gives
rise to $\mathsf{z}^\beta$. Hence the r.~h.~s.~of (\ref{fw18})
(with $U$ replaced by $\mathsf{z}^\beta U D^{({\bf 0})}$) assumes the form
\begin{align}\label{fw20}
&\mathsf{z}^\beta \sum_{k\ge l} \tbinom{k}{l} \mathsf{z}_{k}(k-l) \\
&\times\sum_{(U)}\big(\iota_{\bf 0}U_{(1)}+\sum_{{\bf n}\not=\0}
(\varepsilon_{\bf n} U_{(1)})\mathsf{z}_{\bf n}\big)\cdots
\big(\iota_{\bf 0}U_{(k-l-1)}+\sum_{{\bf n}\not=\0}
(\varepsilon_{\bf n} U_{(k-l-1)})\mathsf{z}_{\bf n}\big). \nonumber
\end{align}
It follows from (\ref{fw18}) (with $l$ replaced by $l+1$) and the
combinatorial identity $(k-l)\binom{k}{l}$ $=(l+1)\binom{k}{l+1}$ 
that (\ref{fw19}) and (\ref{fw20}) coincide.

\medskip

We now turn to the proof of (\ref{fw18}) for $U=\frac{1}{{\bf m}!}\partial^{\bf m}$.
Since the factors of the coproduct of $U$ are of the form
$U'=\frac{1}{{\bf m}'!}\partial^{{\bf m}'}$, see Lemma \ref{lemcop},
on which $\iota_{{\bf 0}}$ vanishes and
on which $\varepsilon_{{\bf n}}$ renders $1$ provided ${\bf n}={\bf m}'\not=\0$,
(\ref{fw18}) assumes the form
\begin{align*}
(\tfrac{1}{{\bf m}!}\partial^{\bf m}\mathsf{z}_l)
=\sum_{k\ge l} \tbinom{k}{l}\,\mathsf{z}_{k}
\sum_{{\bf m}_1+\cdots+{\bf m}_{k-l}={\bf m}}\mathsf{z}_{{\bf m}_1}
\cdots\mathsf{z}_{{\bf m}_{k-l}},
\end{align*}
which, relabelling $k-l$ as $k$, we rewrite as
\begin{align}\label{fw21}
(\tfrac{1}{{\bf m}!} \partial^{\bf m} l! \mathsf{z}_l)
=\sum_{k\ge 0} \tfrac{(l+k)!}{k!} \, \mathsf{z}_{k+l}
\sum_{{\bf m}_1+\cdots+{\bf m}_{k}={\bf m}}\mathsf{z}_{{\bf m}_1}
\cdots\mathsf{z}_{{\bf m}_{k}}.
\end{align}
From definitions (\ref{ao20bis}) and (\ref{ao25})
we see that this amounts to the Fa\`a di Bruno formula
\begin{align*}
\lefteqn{\frac{1}{{\bf m}!}\frac{d^{\bf m}}{dy^{\bf m}}_{|y=0}
\frac{d^la}{du^l}(p(y))}\nonumber\\
&=\sum_{k\ge 0}\frac{1}{k!}\frac{d^{l+k}a}{du^{l+k}}(0)
\sum_{{\bf m}_1+\cdots+{\bf m}_{k}={\bf m}}
\frac{1}{{\bf m}_1!}\frac{d^{{\bf m}_1}p}{dy^{{\bf m}_1}}(0)
\cdots\frac{1}{{\bf m}_k!}\frac{d^{{\bf m}_{k}}p}{dy^{{\bf m}_{k}}}(0),
\end{align*}
cf. \cite[(2.1)]{Frabetti} up to re-summation.
\end{proof}

\bigskip

\section*{Acknowledgments}
The authors thank Yvain Bruned and Stefan Hollands for fruitful discussions. PL thanks Nikolas Tapia for discussions on free pre-Lie algebras.
MT thanks Lorenzo Zambotti for suggestions.

\medskip
\printglossary[title={Symbolic index}]

\medskip

\end{document}